\documentclass{easychair}

\usepackage{filecontents}
\usepackage[utf8]{inputenc}
\usepackage{url}
\usepackage[all]{xy}

\usepackage{qsymbols}
\usepackage{fancybox}
\usepackage{prooftree}
\usepackage{hyperref}
\usepackage{stackengine,scalerel}
\usepackage{enumitem}
\usepackage{amsfonts}
\usepackage{xcolor}
\usepackage{newunicodechar}
\usepackage{bclogo}
\usepackage{ulem}
\usepackage[shortcuts]{extdash}
\usepackage{textcomp}
\usepackage{soul}

\usepackage{breqn}

\newcommand{\Ell}{\mathcal{L}}

\newqsymbol{"<"}{\langle}
\newqsymbol{">"}{\rangle}
\newqsymbol{`(.)}{\odot}\DeclareUnicodeCharacter{2014}{\dash}
\newcommand{\zero}{\underline{0}}
\newcommand{\one}{\underline{1}}
\newcommand{\two}{\underline{2}}

\newtheorem{corollary}{Corollary}{\itshape}{\rmfamily}
\newtheorem{theorem}{Theorem}{\itshape}{\rmfamily}
\newtheorem{remark}{Remark}{\itshape}{\rmfamily}
\newtheorem{definition}{Definition}{\itshape}{\rmfamily}
\newtheorem{proposition}{Proposition}{\itshape}{\rmfamily}
\newtheorem{example}{Example}{\itshape}{\rmfamily}
\newtheorem{lemma}{Lemma}{\itshape}{\rmfamily}

\newcommand{\nat}{\mathbb{N}}
\newqsymbol{`e}{\varepsilon}

\newcommand{\ind}[2]{\ensuremath{\llparenthesis \udl{#1},#2\rrparenthesis}}
\newcommand{\udl}[1]{\underline{#1}}
\newcommand{\dup}[3]{\ensuremath{\ind{#1}{#2}\,\triangledown {#3}}}
\newcommand{\era}[3]{\ind{#1}{#2}\odot {#3}}
\newcommand{\LR}{\ensuremath{`L_{\textrm{\footnotesize\textregistered}}}}
\newcommand{\lRdB}{\ensuremath{`l_{\textrm{\textregistered}}}}
\newcommand{\LRdB}{\ensuremath{`L_{\textrm{\textregistered}}}}

\newcommand{\LRdBin}{\LRdB^{\Ell}}
\newcommand{\lRdBin}{\lRdB^{\Ell}}
\newcommand{\Lin}{\ensuremath{`L}^{\Ell}}
\newcommand{\lin}{\ensuremath{`l}^{\Ell}}

\newcommand{\sub}[3]{\ensuremath{#1[#2/#3]}}

\newcommand{\down}[1]{\downarrow#1}
\newcommand{\Downs}{~\Downarrow_s~}
\newcommand{\Downr}{~\Downarrow_r~}

\newcommand{\Nd}[1]{*+[o][F]{\large \bf #1}}
\newcommand{\dupG}[4]{\ensuremath{#1"<"^{#2}_{#3}~#4}}
\newcommand{\eraG}[2]{#1`(.)#2}

\newcommand{\dupL}[4]{\ensuremath{\mathcal{C}^{#2|#3}_{#1}(#4)}}
\newcommand{\eraL}[2]{\ensuremath{\mathcal{W}_{#1}(#2)}}
\newcommand{\lu}{\lambda\upsilon}
\newcommand{\Lu}{\Lambda\upsilon}
\renewcommand{\stackrel}[2]{\mathrel{\mathop{#2}\limits^{#1}}}
\newcommand{\luar}{\stackrel{\lu}{\xymatrix@C=18pt{\ar[r] &}}}
\newcommand{\luarone}{\stackrel{\lu_1}{\xymatrix@C=18pt{\ar[r] &}}}
\newcommand{\luarplus}{\stackrel{\lu\;+}{\xymatrix@C=25pt{\ar[r] &}}}
\newcommand{\luin}{\lambda^\Ell_\upsilon}
\newcommand{\luarin}{\stackrel{\luin}{\xymatrix@C=23pt{\ar[r] &}}}
\newcommand{\Luin}{\Lambda^{\Ell}_\upsilon}
\newcommand{\lift}{\Uparrow\!}
\newcommand{\shift}{\uparrow}

\newcommand{\upd}[1]{\llbracket #1 \rrbracket}
\newcommand{\zip}{\mathop{\,\ddagger\,}}
\newcommand{\mapdown}{\downarrow}
\newcommand{\mapup}{\uparrow}
\newcommand{\plpl}{\zip}

\newcommand{\bO}{0}
\newcommand{\bI}{1}

\newcommand{\multiDup}[1]{\raisebox{-5pt}{$\stackrel{\displaystyle
      \bigtriangledown}{\scriptstyle \ind{k}{`g}`:#1}$} \ind{k}{`g}~}
\newcommand{\replace}[2]{\llfloor #1 / #2\rrceil} 
\newcommand{\repl}[3]{#1\,\bm{\lfloor}#2 "<-" #3\bm{\rceil}} 

\renewcommand{\emph}[1]{\textit{#1}}
\newcommand\Warning{%
 \makebox[1.4em][c]{%
 \makebox[0pt][c]{\raisebox{.1em}{\small!}}%
 \makebox[0pt][c]{\color{red}\Large$\bigtriangleup$}}}%

\makeatletter
\def\thickhline{%
  \noalign{\ifnum0=`}\fi\hrule \@height \thickarrayrulewidth \futurelet
   \reserved@a\@xthickhline}
\def\@xthickhline{\ifx\reserved@a\thickhline
               \vskip\doublerulesep
               \vskip-\thickarrayrulewidth
             \fi
      \ifnum0=`{\fi}}
\makeatother

\newlength{\thickarrayrulewidth}
\definecolor{darkGreen}{rgb}{0,.5,0}
\definecolor{mauve}{rgb}{1,0,1}
\definecolor{darkRed}{cmyk}{.3,1,.3,0}
\definecolor{cyanp}{cmyk}{.5,.3,0,0}
\definecolor{yellow}{cmyk}{0,0,.7,0}
\definecolor{beige}{cmyk}{0,.2,.7,0}
\definecolor{lightBrown}{cmyk}{0,.5,.7,0}
\definecolor{darkBrown}{cmyk}{.3,.75,.75,.15}

  {\color{blue}}
  {}
  {\color{darkBrown}}
  {}

  {\color{red}}
  {} %

\usepackage{color}
\usepackage{alltt}
\usepackage{bm}

\DeclareUnicodeCharacter{2115}{$\nat$}
\DeclareUnicodeCharacter{3B5}{\ensuremath{\varepsilon}}
\DeclareUnicodeCharacter{3B1}{\ensuremath{\alpha}}
\DeclareUnicodeCharacter{338}{\ensuremath{\!\!\!\!\not{}~\;}}
\DeclareUnicodeCharacter{2113}{\ensuremath{\ell}}
\DeclareUnicodeCharacter{039B}{{\color{black} \ensuremath{\Lambda}}}
\DeclareUnicodeCharacter{03BB}{\ensuremath{\lambda}}
\DeclareUnicodeCharacter{03B2}{\ensuremath{\beta}}
\DeclareUnicodeCharacter{19B}{\ensuremath{\lambdabar}}
\DeclareUnicodeCharacter{2238}{\ensuremath{\dot{-}}}
\DeclareUnicodeCharacter{2081}{\ensuremath{_1}}
\DeclareUnicodeCharacter{2082}{\ensuremath{_2}}
\DeclareUnicodeCharacter{22A4}{\ensuremath{\top}}
\DeclareUnicodeCharacter{228E}{\ensuremath{\uplus}}
\DeclareUnicodeCharacter{2208}{\ensuremath{\in}}
\DeclareUnicodeCharacter{207A}{\ensuremath{^+}}
\DeclareUnicodeCharacter{03C5}{\ensuremath{\upsilon}}
\DeclareUnicodeCharacter{03A3}{\ensuremath{\Sigma}}
\DeclareUnicodeCharacter{301A}{\ensuremath{\llbracket}}
\DeclareUnicodeCharacter{301B}{\ensuremath{\rrbracket}}
\DeclareUnicodeCharacter{21D1}{\ensuremath{\Uparrow}}
\DeclareUnicodeCharacter{2237}{\ensuremath{:\hspace*{-.5pt}:}}
\DeclareUnicodeCharacter{227B}{\ensuremath{\succ}}
\DeclareUnicodeCharacter{227A}{\ensuremath{\prec}}
\DeclareUnicodeCharacter{2198}{\ensuremath{\searrow}}
\DeclareUnicodeCharacter{25BD}{\ensuremath{\triangledown}}
\DeclareUnicodeCharacter{2299}{\ensuremath{\odot}} 
\DeclareUnicodeCharacter{2B3}{r}
\DeclareUnicodeCharacter{2987}{$\llparenthesis$}
\DeclareUnicodeCharacter{2988}{$\rrparenthesis$}
\newif\ifcomment

\commenttrue           

\newif\ifUcomment 
\Ucommentfalse
\ifUcomment
\newcommand{\Jel}[1]{\ifcomment {\color{black} {#1}} \else #1\fi}

\newcommand{\Sim}[1]{\ifcomment {\color{black} {#1}} \else #1\fi}
\newcommand{\Pier}[1]{\ifcomment {\color{black} {#1}} \else #1\fi}
\else
\newcommand{\Jel}[1]{\ifcomment {\color{black} {#1}} \else #1\fi}
\newcommand{\Sil}[1]{\ifcomment {\color{black} {#1}} \else #1\fi}
\newcommand{\Sim}[1]{\ifcomment {\color{black} {#1}} \else #1\fi}
\newcommand{\Pier}[1]{\ifcomment {\color{black} {#1}} \else #1\fi}
\fi

 \newenvironment{minted}[1]{\begin{samepage}\begin{alltt}\color{darkGreen}}{\end{alltt}\end{samepage}}

\newcommand{\thefigurelOOO}{
 \begin{figure}[!ht]
   \centering
  \begin{displaymath}
    \resizebox{.15\textwidth}{!}{\xymatrix{
        &&    `l\ar@{-}[d] \\
        &&   @\ar@{-}[ld]\ar@{-}[rd]\\
        &@\ar@{-}[ld]\ar@{-}[rd]&& \ar@{.>}@(dr,ur)[uul]\\
        \ar@{.>}@(dl,ul)[uuurr]&&\ar@{.>}@(dr,dr)[uuu]
}} \qquad \raisebox{-20pt}{exemplifies} \qquad
  \xymatrix@C=7pt@R=10pt{&&\ar@{-}[d]\\\ar@{.>}[rr]&&{`l}\ar@{-}[d]&&\ar@{.>}[ll]\\
   &&&&\ar@{.>}[ull]}%
\end{displaymath}
\medskip
  \begin{displaymath}
    \resizebox{\textwidth}{!}{
      \xymatrix{
        &&    `l\ar@{-}[d] \\
        &&   @\ar@{-}[ld]\ar@{-}[rd]\\
        &@\ar@{-}[ld]\ar@{-}[rd]&& \ar@{.>}@(dr,ur)[dddl]\\
        \ar@{.>}@(dl,ul)[dr]&&\ar@{.>}@(dr,ur)[dl]&&\\
        &\triangledown\ar@{.>}@(d,ul)[dr]\\
        &&\triangledown\ar@{.>}@(dl,l)[uuuuu]
      }
      \qquad
      \xymatrix{
        &&    `l\ar@{-}[d] \\
        &&   @\ar@{-}[ld]\ar@{-}[rd]\\
        &@\ar@{-}[ld]\ar@{-}[rd]&& \ar@{.>}@(dr,ur)[dd]\\
        \ar@{.>}@(dl,ul)[ddrr]&&\ar@{.>}@(dr,ul)[dr]&&\\
        &&&\triangledown\ar@{.>}[dl]&&\\
        &&\triangledown\ar@{.>}@(dl,ul)[uuuuu]
      }
      \xymatrix{
        &&    `l\ar@{-}[d] \\
        &&   @\ar@{-}[ld]\ar@{-}[rd]\\
        &@\ar@{-}[ld]\ar@{-}[rd]&& \ar@{.>}@(dr,ur)[ddll]\\
        \ar@{.>}@(dl,ul)[dr]&&\ar@{.>}@(dr,ur)[dd]&&\\
        &\triangledown\ar@{.>}@(d,ul)[dr]\\
        &&\triangledown\ar@{.>}@(dl,l)[uuuuu]
      }
      \qquad
      \xymatrix{
        &&    `l\ar@{-}[d] \\
        &&   @\ar@{-}[ld]\ar@{-}[rd]\\
        &@\ar@{-}[ld]\ar@{-}[rd]&& \ar@{.>}@(dr,ur)[dddl]\\
        \ar@{.>}@(dl,ur)[dr]&&\ar@{.>}@(dr,ul)[dl]&&\\
        &\triangledown\ar@{.>}@(d,ul)[dr]\\
        &&\triangledown\ar@{.>}@(dl,l)[uuuuu]
      }
      \qquad
      \xymatrix{
        &&    `l\ar@{-}[d] \\
        &&   @\ar@{-}[ld]\ar@{-}[rd]\\
        &@\ar@{-}[ld]\ar@{-}[rd]&& \ar@{.>}@(dr,ul)[dddl]\\
        \ar@{.>}@(dl,ur)[dr]&&\ar@{.>}@(dr,ul)[dl]&&\\
        &\triangledown\ar@{.>}@(d,ur)[dr]\\
        &&\triangledown\ar@{.>}@(dl,l)[uuuuu]
      }
    }
  \end{displaymath}
   \caption{$`l ((\udl{0}\,\udl{0})\,\udl{0})$ and antecedents by
     \textsf{readback} as terms with two duplications}\label{fig:Lxxx}
 \end{figure}
 }
\newcommand{\BigStepSub}{
\begin{figure}[!ht]
\centering
 \Ovalbox{
   \begin{minipage}{.9\textwidth}
     \begin{displaymath}
  (Sapp)~
  \prooftree
  \sub{t_1}{t}{\ind{i}{`a}} \Downs u_1 \qquad  \sub{t_2}{t}{\ind{i}{`a}} \Downs u_2
  \justifies  \sub{(t_1\, t_2)}{t}{\ind{i}{`a}} \Downs u_1\,u_2 
  \endprooftree
\end{displaymath}

\bigskip
  
    \begin{displaymath}
    (Sabs)~
  \prooftree
  \sub{t'}{t}{\ind{i+1}{`a}} \Downs u
  \justifies \sub{(`l\,t')}{t}{\ind{i}{`a}} \Downs `l\, u
  \endprooftree
\end{displaymath}

\bigskip

  \begin{displaymath}
   (Sind_1)~
   \prooftree
  \justifies \sub{\ind{i}{`a}}{t}{\ind{i}{a}} \Downs t
  \endprooftree
\end{displaymath}

\bigskip

\begin{displaymath}
  (Sind_2)~
  \prooftree
 j > i
   \justifies \sub{\ind{j}{`b}}{t}{\ind{i}{`a}} \Downs \ind{j-1}{`b}
   \endprooftree
 \end{displaymath}
 
\bigskip
 
 \begin{displaymath}
(Sind_3)~
   \prooftree
 j < i  \vee (i = j \wedge `a \neq `b) 
   \justifies \sub{\ind{j}{`b}}{t}{\ind{i}{`a}} \Downs \ind{j}{`b}
   \endprooftree
 \end{displaymath}

 \bigskip

\begin{displaymath}
  (S\odot_1)~
  \prooftree
   \justifies \sub{(\era{i}{`a}{t'})}{t}{\ind{i}{`a}} \Downs t'
  \endprooftree
  \end{displaymath}

  \bigskip

  \begin{displaymath}
(S\odot_2)~
\prooftree
\sub{t'}{t}{\ind{j}{`b}} \Downs u \qquad   i > j
  \justifies \sub{(\era{i}{`a}{t'})}{t}{\ind{j}{`b}} \Downs \era{i-1}{`a}{u}
  \endprooftree
  \end{displaymath}

  \bigskip

  \begin{displaymath}
      (S\odot_3)~
  \prooftree
\sub{t'}{t}{\ind{j}{`b}} \Downs u \qquad   i < j \vee (i = j \wedge `a \neq `b)
  \justifies \sub{(\era{i}{`a}{t'})}{t}{\ind{j}{`b}} \Downs \era{i}{`a}{u}  
  \endprooftree
\end{displaymath}

\bigskip

\begin{displaymath}
  (S\triangledown_1)~
  \prooftree
  \sub{t'}{t}{\ind{i}{`a0}} \Downs v \qquad v \replace{t}{\ind{i}{`a1}} \Downr u
  \justifies \sub{(\dup{i}{`a}{t'})}{t}{\ind{i}{`a}} \Downs u 
  \endprooftree
\end{displaymath}
\bigskip
  \begin{displaymath}
      (S\triangledown_2)~
  \prooftree
  \sub{t'}{t}{\ind{j}{`b}} \Downs u  \qquad i > j
  \justifies \sub{(\dup{i}{`a}{t'})}{t}{\ind{j}{`b}}  \Downs  \dup{i-1}{`a}{u}
  \endprooftree
  \end{displaymath}

\bigskip  

  \begin{displaymath}
    (S\triangledown_3)~
  \prooftree
  \sub{t'}{t}{\ind{j}{`b}} \Downs u \qquad i < j \vee (i = j \wedge `a \neq `b)
  \justifies \sub{(\dup{i}{`a}{t'})}{t}{\ind{j}{`b}}  \Downs  \dup{i}{`a}{u}
  \endprooftree
\end{displaymath}
\end{minipage}}
\caption{Big-step semantics of \emph{substitute}}\label{fig:BSS}
\end{figure}}
\newcommand{\BigStepRep}{
\begin{figure}[!ht]
  \centering
\Ovalbox{
\begin{minipage}{.9\textwidth}
     \begin{displaymath}
  (Rapp)~
  \prooftree
  t_1 \replace{t}{\ind{i}{`a}} \Downr u_1 \qquad  t_2 \replace{t}{\ind{i}{`a}} \Downr u_2
  \justifies  (t_1\, t_2)\replace{t}{\ind{i}{`a}} \Downr u_1\,u_2  
  \endprooftree
\end{displaymath}

\bigskip
  
\begin{displaymath}
  (Rabs)~
  \prooftree
  t' \replace{t}{\ind{i+1}{`a}} \Downr u 
  \justifies (`l\,t')\replace{t}{\ind{i}{`a}} \Downr `l\, u
  \endprooftree
\end{displaymath}
\bigskip
  \begin{displaymath}
   (Rind_1)~
   \prooftree
  \justifies \ind{i}{`a}\replace{t}{\ind{i}{a}} \Downr t
  \endprooftree
\qquad\qquad \qquad
    (Rind_2)~
  \prooftree
\ind{i}{`a}\neq \ind{j}{`b}
  \justifies \ind{j}{`b}\replace{t}{\ind{i}{`a}} \Downr \ind{j}{`b}
\endprooftree
  \end{displaymath}

  \begin{displaymath}
      (R\odot_1)~
      \prooftree
    \justifies (\era{i}{`a} t') \replace{t}{\ind{i}{`a}} \Downr t'
  \endprooftree
\qquad
  (R\odot_2)~
  \prooftree
  t'  \replace{t}{\ind{i}{`a}} \Downr u \qquad \ind{i}{`a}\neq \ind{j}{`b}
   \justifies (\era{j}{`b} t') \replace{t}{\ind{i}{`a}} \Downr\era{j}{`b} u
  \endprooftree
\end{displaymath}
\bigskip
\begin{displaymath}
  (R\triangledown_1)~
  \prooftree
  t'  \replace{t}{\ind{i}{`a0}} \Downr v \qquad v  \replace{t}{\ind{i}{`a1}} \Downr u 
  \justifies (\dup{i}{`a}{t'}) \replace{t}{\ind{i}{`a}} \Downr u 
  \endprooftree
  \end{displaymath}
    \bigskip
  \begin{displaymath}
  (R\triangledown_2)~
  \prooftree
  t'\replace{t}{\ind{j}{`b}} \Downr u \qquad \ind{i}{`a}\neq \ind{j}{`b}
  \justifies (\dup{i}{`a}{t'})\replace{t}{\ind{j}{`b}}  \Downr  \dup{i}{`a}{u}
  \endprooftree
\end{displaymath}
\end{minipage}}
\caption{Big-step semantics of \emph{replace}}\label{fig:BSR}
\end{figure}}

\AtBeginDocument{%
  \providecommand\BibTeX{{%
    \normalfont B\kern-0.5em{\scshape i\kern-0.25em b}\kern-0.8em\TeX}}}

\newcommand{\streetaddress}[1]{~#1}
\newcommand{\city}[1]{~#1}
\newcommand{\postcode}[1]{~#1}
\newcommand{\country}[1]{~#1}

\begin{document}
\authorrunning{Ghilezan et al.}
\titlerunning{L-ypes for resource aware languages}
\title{List types for resource aware languages\\ an implicit name approach}

\author{Silvia Ghilezan \inst{1} \inst{2} 
  \and {Jelena Iveti\'{c}} \inst{1}
  \and Pierre Lescanne \inst{3}
  \and Simona Proki\'{c} \inst{1}}

\institute{University of Novi Sad, Faculty of Technical Sciences\\
  \streetaddress{Trg Dositeja Obradovi\'{c}a 6}
  \city{Novi Sad}
  \postcode{21000}
  \country{Serbia}
\and
  Mathematical Institute of the Serbian Academy of Sciences and Arts\\
  \streetaddress{Kneza Mihaila 36}
  \city{Beograd}
  \postcode{11000}
  \country{Serbia}
\and
  University of Lyon, \'Ecole normale sup\'erieure de Lyon\\
  \streetaddress{46 all\'ee d'Italie}
  \city{Lyon}
  \postcode{69364}
  \country{France}}

\maketitle
\begin{abstract}
  A novel formalisation of variable control in languages with implicit names based on de Bruijn indices is presented.
We design and implement three languages: first, a restricted language with implicit names; then, a restricted calculus with implicit names and explicit substitution, and  finally, an extended calculus with implicit names, implicit substitution and resource control. We propose a novel concept of list types, which are used to give a simple and manageable definition of linearity.
We develop an implementation in Haskell.
\end{abstract}

\section{Introduction}
\label{sec:introduction}
 
In computation, the control of variable use goes back to Church’s $\lambda I$-calculus and restricted terms~\cite{hindley97:_basic_simpl_theor}. Likewise, in logic, the control of formula use is present in
Gentzen’s structural rules~\cite{gent35} which enable a wide class of substructural logics~\cite{DS-H93}.
In programming, the augmented ability to control 
the use of operations and objects has a wide range of applications which enable, among others, compiling functional languages without garbage collector and avoids memory leaking~\cite{rose11:_implem_trick_that_make_crsx_tick,rose:LIPIcs:2011:3130}; inline expansion in compiler optimisations~\cite{DBLP:conf/sas/CalvertM12}; safe memory management~\cite{walk05};
controlled type discipline as a framework for resource-sensitive compilation~\cite{DBLP:conf/esop/GhicaS14};  the interpretation of linear formulae as session types that provides a purely logical account of session types~\cite{DBLP:conf/concur/CairesP10}.
At the core of all these phenomena is the  Curry-Howard correspondence of formulae-as-types and proofs-as-terms.

\paragraph*{\bf Control:  by restriction vs by extension}

There are several restricted classes of  $\lambda$-terms, where the restrictions are due to the control of variable use. The best known among them are: $\lambda \mathsf{I}$-terms, aka \emph{relevant} terms, where variables occur at least once;
$\mathsf{BCK}`l$-terms, aka
\emph{affine} terms, where variables occur at most once;
$\mathsf{BCI}`l$-terms, aka \emph{linear} terms, where each variable occurs exactly once~\cite{hindley97:_basic_simpl_theor,DBLP:journals/tcs/Girard87}.
E.g. the combinator $\mathsf{K}$ is not a $\lambda \mathsf{I}$-term and the combinator $\mathsf{S}$
is not a $\mathsf{BCK}\lambda$-term.
This ``\emph{control by restriction}'' approach  is widely present in substructural logics~\cite{DS-H93}, substructural type theory~\cite{walk05}, linear logic~\cite{DBLP:conf/tapsoft/GirardL87,DBLP:conf/lics/LincolnM92}, among others.
 On the other hand, the control of variable use can be achieved by extending the language by operators
 meant to  tightly encode the control.
If a variable has to be reused, it will be explicitly duplicated, whereas
if the variable is not needed, it will be explicitly erased.
These two resource control operators, \emph{duplication} and
\emph{erasure}, are extensions of the syntax of the $`l$-calculus which
allow  all $\lambda$-terms to become: relevant (only erasure is used), affine (only duplication is used)  and linear (both erasure and duplication are used).
The advantage of this  ``\emph{control by extension}'' approach is that all $`l$-terms can be encoded in the extended calculus.
Hence, the extended calculi are equivalent, in computational power, to $\lambda$-calculus, which is not the case with the restricted calculi.
This approach has been developed in different theoretical~\cite{DBLP:conf/concur/Boudol93,DBLP:journals/iandc/KesnerL07} and applicative~\cite{rose11:_implem_trick_that_make_crsx_tick,DBLP:conf/sas/CalvertM12} settings.
From a proof theoretical perspective, such a simply typed \emph{extended}
$`l$-calculus has a Curry-Howard correspondence with
intuitionistic logic with \emph{Contraction} and
\emph{Thinning} structural rules~\cite{Troelstra:2000:BPT:351148}, whereas a \emph{restricted} $`l$-calculus corresponds to substructural logic~\cite{DS-H93,Restall00} such as relevant or affine logic.

\paragraph*{\bf Names: explicit vs implicit}

The well-known $\lambda$-calculus is a calculus with explicit use of variables (names).
On the other hand, the calculus with implicit names is de Bruijn notation of $\lambda$-calculus that avoids the explicit naming of variables by employing de Bruijn indices~\cite{NGDeBruijn108,curien93:categ_combin,Lescanne95WADT}.
Each variable is replaced by a natural
number which is the number of $`l$'s crossed in order to reach the
binder of that variable.  For instance in de Bruijn notation, the combinator $\mathsf{I} \equiv `l x.x$  is  $`l \udl{0}$, the combinator $\mathsf{K} \equiv `l x.`l y.x$  is $`l `l \udl{1}$ and the combinator
$\mathsf{S}\equiv `l x .`l y. `l z . x z (y z)$ is $`l`l`l\,\udl{2}\,\udl{0}(\udl{1}\,\udl{0})$. The profound advantage of de Bruijn notation is that $\alpha$-conversion, the renaming of bound variables, is not needed, which significantly facilitates implementation and also, in the case of an extended $`l$-calculus, simplifies the rules.

\paragraph*{\bf Foundations}
In this paper, we study both restricted and extended control of
variable use in calculi with \emph{implicit names}.  This means that instead
of (explicit) variables we use either de Bruijn indices, or novel
®\=/indices. Inspired by de Bruijn indices, ®\=/indices provide
information about duplication of names.  We design three languages and implement two associated calculus:
\begin{itemize}
    \item a restricted language with implicit names $\Lin$; 
    \item a restricted language with implicit names and explicit substitution $\Luin$, together with its $\luin$-calculus;  
    \item an extended language with
implicit names, implicit substitution and explicit duplication and erasure $\LRdBin$, together with its $\lRdBin$-calculus.
\end{itemize}
In all introduced calculi, linear terms are defined as
\Jel{specifically typeable} terms in systems we call \emph{list
  types}. The list type, abbreviated as $\Ell$-type, of a term
represents the list of its free indices and is convenient for checking
its linearity.  This is why we use $\Ell$-typeability as a definition
of linearity.

\paragraph*{\bf Linearity: indirect vs direct}
\Pier{ 
For a reader used to explicit names only, i.e. regular $\lambda$-calculus, there exists, at first glance, a simple method to characterise linear terms: it suffices to
check that each variable occurs once, which is easy if one creates a new variable for each new abstraction.}
\Pier{
However with
implicit names, aka de Bruijn indices, checking the occurrences of indices with the same number is of no help, because,
as a matter of fact, a same number may occur several times in a linear term and a term with only occurrences of
different numbers, may be non linear. Thus $`l\zero\, (`l\zero)$ is linear and $`l\zero\,(`l\one)$ is not linear.  We
let the reader imagine examples with thousand $`l$'s (thousand abstractions) and check whether indices $\underline{127}$
and $\underline{828}$ correspond to the same variable or whether those variables are duplicated or not.

Thus for terms with implicit names, checking unique occurrences of numbers (associated with de Bruijn indices) does not
work.  There is then
} 
\Sil{ an \emph{indirect method} for checking linearity with implicit names} 
\Pier{
which consists of two steps:
\begin{enumerate}
\item to translate a $`l$-term with implicit names into a $\lambda$-term with explicit names,
\item to check unique occurrence of variables.
\end{enumerate}
} \Sil{ This indirect} \Pier{ method is algorithmically costly, since
  it consists in a translation from implicit names to explicit
  names. 
} In turn, we give a \emph{direct} characterisation of
\emph{linearity} with implicit names by means of $\Ell$-types.
 We likewise did the direct definition of the occurence of an ®-index in
  $\lRdB$-term.

\paragraph*{\bf Implementation}
We worked simultaneously on the development of the $\Ell$-type calculi and on their implementation,
because the implementation shapes the development\Pier{\footnote{The implementation in
    \textsf{Haskell}~\cite{githubrepo2} is complete, while an \textsf{Agda} implementation is still
    in development~\cite{githubrepo}}}.

\medskip

The main contributions of this paper are:
\begin{itemize}
\item[-] the use of de Bruijn indices (an old concept, \Jel{in a novel setting}) and ®-indices (a new concept), 
\item[-] the novel concept  of list types, \Sil{dubbed $\Ell$-types,} for characterising \Sil{linearity in languages with implicit names,} 
\item[-] a formal definition of linearity based on $\Ell$-typeability in 
\Sil{three languages}, especially in a calculus with
    explicit substitution and in a calculus with resource control, because in those calculi a formal and direct
    definition of linearity is hard to state,
\item[-] proofs of $\Ell$-type preservation, hence of linearity,
\item[-] an implementation in \textsf{Haskell} for framework \cite{githubrepo2}. 
\end{itemize}

The rest of this paper is organised as follows. We first review the background on de Bruijn indices, the  $\lambda$-calculus with implicit names,  
in Section~\ref{sec:dB}. In Section~\ref{sec:LinLambda}, we introduce the notion of \Sil{list types (abbreviated by $\Ell$-types)} and design the language $\Lin$ of restricted terms with implicit names.  In Section~\ref{sec:luin}, we extend the notion of $\Ell$-types and design the language $\Luin$ of restricted terms with implicit names and explicit substitution and the corresponding calculus for which we prove type preservation.
In Section~\ref{sec:LRdB}, we modify the notion of $\Ell$-types for the design of the language $\LRdBin$ of extended terms with implicit names, implicit substitution, and resource control, and the corresponding calculus for which we prove type preservation. We also provide several examples of computation, and an implementation in \textsf{Haskell}. In Section~\ref{sec:related}, we discuss related work. Section~\ref{sec:conc} concludes the paper. To facilitate the reader's comprehension and understanding of the text, the following table highlights the most important introduced notations, along with their informal definitions and references to the sections in which they are formally introduced.

\begin{displaymath}
  \begin{array}[h]{|l|l|l|}
  \hline %
  \textbf{name}&\textbf{informal definition}&\textbf{reference to its definition}\\\hline
  \Lambda & \textrm{The set of terms (with implicit names)} &\textrm{Section~\ref{sec:dB}}\\
  \Lin& \textrm{The set of restricted terms}& \textrm{Section~\ref{sec:LinLambda}}\\
 \Lu&\textrm{The set of terms with explicit substitutions} & \textrm{Section~\ref{sec:luin}}\\
    \Luin&\textrm{The set of restricted terms with explicit substitutions} & \textrm{Section~\ref{sec:luin}}\\
    \luin&\textrm{The calculus of~} \Luin&\textrm{Section~\ref{sec:luin}}\\
    \LRdB&\textrm{The set of extended terms} &\textrm{Section~\ref{sec:LRdB}}\\
    \LRdBin&\textrm{The set of extended terms with resource control} &\textrm{Section~\ref{sec:LRdB}}\\
   \lRdBin & \textrm{The calculus of~} \LRdBin &\textrm{Section~\ref{sec:LRdB}}\\
    \hline  
\end{array}
\end{displaymath}

\section{Terms with implicit names $\Lambda$}
\label{sec:dB}

Our development relies on the paradigm of implicit names in formal calculi.
We recall the notion of term with implicit names based on de Bruijn indices~\cite{NGDeBruijn108,curien93:categ_combin,Lescanne95WADT}. 
Let us consider the (regular)  $\lambda$-terms $\textsf{K} \equiv `l x. `l y. x$ and $\textsf{S}\equiv `l x .`l y. `l z . x z (y z)$
 and the three contractions of the term \textsf{SK}
\begin{displaymath}
  (`l x .`l y. `l z . x z (y z))\,(`l x. `l y. x) "->" %
`l y .`l z. (`l x . `l y . x) z (y z)  "->"%
`l y . `l z . (`l y . z) (y z) "->" `l y. `l z . z.
\end{displaymath}

Assume that we want to represent those terms without using variables. Such
a variable-free representation is
called sometimes \emph{Bourbaki assembly}, because this variable-free
two dimensional representation of terms has been first used by
Bourbaki~\cite{Bourbaki39}
\begin{figure}
  \centering
  \includegraphics[width=2cm]{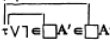}
  \caption{Bourbaki assembly in~\href{https://books.google.fr/books?id=VDGifaOQogcC&pg=SA1-PA14&dq=Bourbaki+ensemble+signes+assemblages}{\cite{Bourbaki39}}}
  \label{fig:Bourb}
\end{figure}
(see Figure~\ref{fig:Bourb}) and has been called
``assembly''~\cite{bourbaki68:_theor_sets,DBLP:journals/jfrea/Grimm10}.  \Pier{In a contemporary style (eighty five years later and a different context), Figure~\ref{fig:Bourb} can be explained as follows.  In a system of implicit names and an infix operator $\in$, this \emph{assembly} represents $`t((\neg (\zero \in A')) \vee (\zero \in A))$, where $\tau$ is a binder acting like $`l$.  In a system of explicit names, where $`t$ binds $x$, this assembly represents $`t~x . (\neg(x`: A')) \vee (x`:A)$.

Figure~\ref{fig:Bourb} resembles Figure~\ref{fig:assembly} (we use here an infix notation for the binary operator
``application'' and Bourbaki uses prefix notations).}

\begin{figure}

  \vspace*{25pt}

\begin{displaymath}
 \xymatrix@C=7pt@R=2pt{
`l \ar@{-}@(u,u)[rrr]&`l\ar@{--}@(u,u)[rrrr]&`l\ar@{.}@(u,u)[rr]\ar@{-}@(ur,u)[rrrr] &\Box&\Box&(\Box&\Box)&
`l \ar@{-}@(u,u)[rr]&`l&\Box
}
\end{displaymath}
\caption{Bourbaki assembly of $\mathsf{SK}$}\label{fig:assembly}
\end{figure}

Later and independently, de
Bruijn proposed a one dimension variable-free representation using natural
numbers\footnote{This has been popularised by
  Curien~\cite{curien93:categ_combin}. Notice that de Bruijn and Curien
  make the indices to start at $1$, but the third author
  proposed in~\cite{Lescanne95WADT} the indices to start at $0$, a convention
  largely adopted since~\cite{pierce02:_types_progr_languag}.}, called
since \emph{de Bruijn indices}.  Each variable is replaced by a natural
number which is the number of $`l$'s crossed in order to reach the
binder of that variable.  For instance, $`l x .`l y. `l z . x z (y z)$ is
replaced by $`l`l`l\,\udl{2}\,\udl{0}(\udl{1}\,\udl{0})$.  Indeed, $x$ is
replaced by $\udl{2}$ because one crosses two $`l$'s to meet its binder,
$y$ is replaced by~$\udl{1}$ because one crosses one $`l$ to meet its
binder and $z$ is replaced by~$\udl{0}$ because one crosses no $`l$ to
meet its binder.

The abstract syntax of terms with de Bruijn notation is the following:
\begin{eqnarray*}
    t &::=& \udl{n} \mid `l t \mid t\; t
\end{eqnarray*}
where $\udl{n}$, associated with $n \in \nat$, is an index. The set of all terms with de Bruijn notation will be denoted by $`L $ and it will be ranged over by $t, \; s, \ldots$.
We will call them  $\lambda$-terms or terms without mentioning 
 de Bruijn indices if there is no
place for confusion.

Terms with de Bruijn notation are also called terms with \emph{implicit names} since variables are implicit rather than explicit as in the regular $\lambda$-calculus.  Using implicit names is
convenient because terms with de Bruijn indices represent classes of
$`a$\=/conversion of terms with explicit variables.
 

We will see that de Bruijn indices also enable simple
descriptions of features connected with linearity, duplication and erasure
that are otherwise described with cumbersome
notations~\cite{DBLP:journals/iandc/KesnerL07,DBLP:journals/tcs/KesnerR11,DBLP:journals/corr/GhilezanILL14}.
The formal definition of $`b$-reduction is given in Example~\ref{ex:b}.

\Sil{Since ${\sf S} \equiv `l x .`l y. `l z . x z (y z)$ is
replaced by $`l`l`l\,\udl{2}\,\udl{0}(\udl{1}\,\udl{0})$ and ${\sf K} \equiv `l x .`l y.x$ is replaced by $\lambda \lambda \one$ } the above chain of contractions of the term \textsf{SK} is replaced by
\begin{displaymath}
  (`l`l`l\,\udl{2}\,\udl{0}(\udl{1}\,\udl{0}))\,(`l`l\udl{1}) "->" %
  `l`l((`l`l\udl{1})\udl{0}(\udl{1}\,\udl{0})) "->"%
  `l`l((`l\udl{1})\,(\udl{1}\,\udl{0})) "->" `l`l\udl{0}.
\end{displaymath}

\noindent Figure~\ref{fig:SK} presents $\mathsf{SK}$ and its three contractions. It shows how de Bruijn indices are built from
variables (aka explicit names), indicates the links between names and their binders and presents the
chain of $`b$-reductions in de Bruijn notation.
Notice that in
$`l`l((`l\underline{1})\,(\underline{1}\,\underline{0}))$ the same
variable $z$ is associated with two de Bruijn indices,
$\underline{1}$ and $\underline{0}$ and that the same de Bruijn index
$\underline{1}$ is associated with two variables,  $y$
and~$z$.  In the de~Bruijn notation the value of an index associated with
a variable depends on the context. \Sim{Also, notice that $\udl{0}, \udl{1}\, \udl{0}$ and $\udl{2}\,\udl{0}(\udl{1}\,\udl{0})$ are examples of open terms in de Bruijn notation.}

\begin{figure}[!th]
  \centering
 \resizebox{0.9\textwidth}{!}{
  \xymatrix @C=7pt@R=7pt{
&&&&{@}\ar@{-}[dl]\ar@{-}[dr] \\
&&&`l\ar@{-}[d]&&`l\ar@{-}[d]&\\
&&&`l\ar@{-}[d]&&`l\ar@{-}[d]&\\
&&&`l\ar@{-}[d]&&\ar@{.>}@(d,r)[uu]\\
&&&{@}\ar@{-}[dll]\ar@{-}[drr]\\
&{@}\ar@{-}[dl]\ar@{-}[dr]&&&&{@}\ar@{-}[dl]\ar@{-}[dr]& \\
\ar@{.>}@(dl,l)[uuuuurrr] &&\ar@{.>}@(dr,l)[uuur]&&\ar@{.>}@(dl,r)[uuuul] &&\ar@{.>}@(dr,r)[uuulll]  &
}
\quad\raisebox{-.1\textwidth}{$"->"$}
  \xymatrix @C=7pt@R=7pt{
&&&`l\ar@{-}[d]\\
&&&`l\ar@{-}[d]\\
&&&{@}\ar@{-}[dll]\ar@{-}[drr]\\
&@\ar@{-}[dl]\ar@{-}[dr]&&&&@\ar@{-}[dl]\ar@{-}[dr]\\
`l\ar@{-}[d]&&\ar@{.>}@(dr,l)[uuur]&&\ar@{.>}@(dl,r)[uuuul] &&\ar@{.>}@(dr,r)[uuulll]\\
`l\ar@{-}[d]\\
\ar@{.>}@(d,r)[uu]&
}
\quad\raisebox{-.1\textwidth}{$"->"$}\quad
  \xymatrix @C=7pt@R=7pt{
&&`l\ar@{-}[d]\\
&&`l\ar@{-}[d]\\
&&@\ar@{-}[dll]\ar@{-}[drr] \\
`l\ar@{-}[d]&&&&@\ar@{-}[dl]\ar@{-}[dr]\\
\ar@{.>}@(dl,l)[uuurr] &&&\ar@{.>}@(dl,r)[uuuul]&&\ar@{.>}@(dr,r)[uuulll]
}
\quad\raisebox{-.1\textwidth}{$"->"$}
  \xymatrix @C=7pt@R=7pt{
`l\ar@{-}[d]\\
`l \ar@{-}[d]\\
\ar@{.>}@(d,r)[u]
}
}
\caption{The term \textsf{S}\textsf{K} and its three contractions}\label{fig:SK}
\end{figure}

The basic reduction considered here is
\begin{displaymath}
  \xymatrix@C=7pt@R=7pt{
& \ar@{-}[dd]\\
&\\
&@ \ar@{-}[dl] \ar@{-}[dr] \\
`l\ar@{-}[d]  && \Nd{B}\\
 \Nd{A} \ar@{.>}@(dl,ul)[uur]
}
\qquad
\raisebox{-.05\textwidth}{$"->"$}
\qquad
\xymatrix@C=7pt@R=7pt{
& \ar@{-}[dd]\\
&\\
&\Nd{A}\ar@{-}[dl]\\
\Nd{B}
}
\end{displaymath}

Three patterns are of interest in Figure~\ref{fig:SK}:
\begin{displaymath}
  \xymatrix@C=7pt@R=7pt{\ar@{-}[d]\\{`l}\ar@{-}[d]\\ \\}
  \qquad
  \xymatrix@C=7pt@R=7pt{&&\ar@{-}[d]\\\ar@{.>}[rr]&&{`l}\ar@{-}[d]\\ &&}%
  \qquad  %
  \xymatrix@C=7pt@R=7pt{&&\ar@{-}[d]\\\ar@{.>}[rr]&&{`l}\ar@{-}[d]&&\ar@{.>}[ll]\\
    &&}%
\end{displaymath}
The first pattern corresponds to a $`l$ that binds no index, the second
pattern corresponds to a $`l$ that binds exactly one index and the third
pattern corresponds to a $`l$ that binds two indices.  This later pattern
is representative, but clearly, there are patterns with more bound indices
(see Figure~\ref{fig:Lxxx}, page \pageref{fig:Lxxx}).
We propose the control of variable by restricting the language to $\Lin$ in  Section~\ref{sec:LinLambda} and $\Luin$ in Section~\ref{sec:luin}, and by  extending the language to $\LRdB$ in Section~\ref{sec:LRdB}.
In the new language  $\LRdB$, extended with two new operators $\triangledown$  (duplicator) and $`(.)$ (erasure),
terms are \emph{linearised}, meaning that only patterns corresponding to a $`l$ that binds exactly one index are present in the Bourbaki representation  (see Figure~\ref{fig:merged}).
This recalls
Lamping's optimal
calculus~\cite{DBLP:conf/popl/Lamping90}, which is described in~\cite{DBLP:conf/popl/GonthierAL92} and in~\cite{DBLP:books/daglib/0095289} with its connections with linear logic.
In $\LRdB$, we have an atomic substitution, whereas in Lamping's
calculus there is none.  Indeed, in Lamping calculus, \emph{fans} (a~kind of duplicators) are propagated.  However, the connection
should be deepened.

\begin{figure}[!ht]
 \centering
   \begin{displaymath}
    \resizebox{.9\textwidth}{!}{
      \xymatrix @C=7pt@R=7pt{
        &&&&{@}\ar@{-}[dl]\ar@{-}[dr] \\
        &&&`l\ar@{-}[d]&&`l\ar@{-}[d]&\\
        &&&`l\ar@{-}[d]&&`l\ar@{-}[d]&`(.)\ar@{.>}[l]\\
        &&&`l\ar@{-}[d]&&\ar@{.>}@(d,l)[uu]\\
        &&&{@}\ar@{-}[dll]\ar@{-}[drr]\\
        &{@}\ar@{-}[dl]\ar@{-}[dr]&&&&{@}\ar@{-}[dl]\ar@{-}[dr]& \\
        \ar@{.>}@(dl,l)[uuuuurrr] &&\ar@{.>}[ddrr]&&\ar@{.>}@(dl,l)[uuuul]
        &&\ar@{.>}@(dr,ur)[ddll]  \\
        \\
        &&&&\triangledown\ar@{.>}@(d,l)[uuuuul]
      }
      \quad\raisebox{-.1\textwidth}{$"->"$} %
      \xymatrix @C=7pt@R=7pt{
        &&&&`l\ar@{-}[d]\\
        &&&&`l\ar@{-}[d]\\
        &&&&{@}\ar@{-}[dll]\ar@{-}[drr]\\
        &&@\ar@{-}[dl]\ar@{-}[dr]&&&&@\ar@{-}[dl]\ar@{-}[dr]\\
        &`l\ar@{-}[d]&&\ar@{.>}[ddr]&&\ar@{.>}@(dl,r)[uuuul] &&\ar@{.>}@(rd,ur)[ddlll]  \\ 
        &`l\ar@{-}[d]&`(.)\ar@{.>}[l]&&\\
        &\ar@{.>}@(d,l)[uu]&&&\triangledown\ar@{.>}@(dl,l)[uuuuu]
      }
      \quad %
      \raisebox{-.1\textwidth}{$"->"$} %
       \xymatrix @C=7pt@R=7pt{
        &&&`l\ar@{-}[d]\\
        &&&`l\ar@{-}[d]\\
        &&&@\ar@{-}[dll]\ar@{-}[drr] \\
        `(.)\ar@{.>}[r]&`l\ar@{-}[d]&&&&@\ar@{-}[dl]\ar@{-}[dr]\\
        &\ar@{.>}@(d,ul)[ddrr] &&&\ar@{.>}@(dl,r)[uuuul]&&\ar@{.>}@(dr,ur)[ddlll]\\
        &&&&&\\
        &&&\triangledown\ar@{.>}@(dl,l)[uuuuu]}
      \raisebox{-.1\textwidth}{$"->"$}
\xymatrix @C=7pt@R=7pt{
        `l\ar@{-}[d]&`(.)\ar@{.>}[l]\\
        `l \ar@{-}[d]\\
        \ar@{.>}@(d,r)[u]
      }
    }
  \end{displaymath}
  \caption{Terms with duplicators and erasures}\label{fig:merged}
\end{figure}

\section{Restricted terms $\Lin$}\label{sec:LinLambda}

In this section, we focus on  \emph{restricted terms}~\cite{hindley97:_basic_simpl_theor}   with implicit names~\cite{NGDeBruijn108,Lescanne95WADT}.
We first define the concept of list types, which we will refer to as $\Ell$-types. Then we define a type system which assigns $\Ell$-types to $\lambda$-terms with implicit names and show how this type system singles out 
linear terms  with implicit names. 
The set of terms typeable with $\Ell$-types will be denoted by~$\Lin$.
 
\subsection{$\Ell$-types for $\Lin$}
\label{sec:Lin-types}

Lists of natural numbers are called $\Ell$-\emph{types} for $\Lin$.

\begin{definition}[$\Ell$-types] The abstract syntax of $\Ell$-types is given by
\begin{eqnarray*}
    \ell  &::=& [] \mid i :: \ell \qquad\qquad \textrm{where~} i \in \nat
\end{eqnarray*}
\end{definition}
The empty list is $[]$ and the \emph{cons} operation, $::$, puts an
element in front of a list. We write\footnote{Beware ! The reader
    should not confuse lists as $\Ell$-types and lists of references.} the list made of $1 :: ( 3 :: (5 :: [] ))$ as $[1, 3, 5]$. A list is
\emph{affine} if its elements are not repeated.  On lists, we define 
two operations: a binary operation \emph{merge}, $\zip$, and a unary operation
decrement, $\downarrow$.
 
\begin{definition}[Merge]\label{def: merge}
The binary operation $\zip$ 
 which merges two lists is defined as follows:
\begin{eqnarray*}
  [] \zip \ell &=& \ell\\
  (i :: \ell) \zip [] &=& i :: \ell\\
  (i_1 :: \ell_1) \zip (i_2 :: \ell_2) &=& \mathbf{if~} i_1 < i_2
                                           \mathbf{~then~} i_1 :: (\ell_1
                                           \zip (i_2 :: \ell_2))
  \\
 && \mathbf{if~} i_1 > i_2 \mathbf{~then~}i_2 :: ((i_1 :: \ell_1) \zip \ell_2 )
\end{eqnarray*}
\end{definition}

\begin{remark}
  Be aware that $\zip$ is \textbf{not total}.   For instance if $j$ occurs both in $\ell_1$ and in $\ell_2$ then $\ell_1\zip \ell_2$ is not defined.
  Also, note that if two sorted lists are merged, the result is a sorted
  list.
\end{remark}
If all elements of a list are strictly positive, the list is said to be a \emph{strictly positive list}. We define a unary operation $\downarrow$ on strictly positive lists.  The result is either the empty list or the list where  all  indices of the initial list are decremented.

\begin{definition}[Decrement]\label{def:decrement} The unary operation $\downarrow$ is defined as follows:
\begin{eqnarray*}
	\downarrow []   &=& [] \\
  \downarrow ((i+1) :: \ell)  &=& i ::~ \downarrow \ell
\end{eqnarray*}
We assume that the list $(i+1) :: \ell$ is strictly positive, thus the list $\ell$ is also strictly positive and $\downarrow \ell$ is defined.
\end{definition}
The function $\downarrow$ \textbf{fails} if the list contains $0$.  Said otherwise $\downarrow$ is not total, that is $\downarrow$ is only defined on strictly positive lists.

The type system that  defines the set of restricted terms $\Lin$ is given as follows.
\begin{definition}[Terms $\Lin$]\label{def:lin}
A $\lin$-term is a $`l$-term that can be typed by the following rules.
\\

\begin{center}
  \noindent
    \Ovalbox{
    \begin{minipage}{.75\textwidth}
      \begin{center}
  \begin{displaymath}
    \mathsf{(ind)}~
\prooftree
    \quad
    \justifies \underline{i} : [i]
    \endprooftree
    \qquad \qquad
    \mathsf{(abs)}~
    \prooftree
    t : 0 :: \ell
    \justifies `l t :~ \downarrow \ell
    \endprooftree
    \qquad\qquad
    \mathsf{(app)}~
    \prooftree
    t_1 : \ell_1 \qquad t_2 : \ell_2  
    \justifies t_1 \; t_2 : \ell_1 \zip \ell_2
    \endprooftree
  \end{displaymath}
\end{center}

\medskip

\end{minipage}}
\end{center}
The set of all $\lin$-terms is denoted by $\Lin$.
\end{definition}

\Sim{The \textsf{(abs)} rule enforces the constraint that} \Sil{$\lambda x.t$ is well-typed under the condition that} \Sim{$x$ belongs to the set of free variables of $t$, via the condition  $t : 0 :: \ell$.}  If the function $\downarrow$
fails the rule \textsf{(abs)} fails as well. \Sim{Furthermore, in the \textsf{(app)} rule, the usage of the merge operator $\zip$ disallows the use of the
same free variable in both sub-terms.} Likewise, if the operator
$\zip$ fails the rule \textsf{(app)} fails as well.  Thus the non
determinism of the type system lies in the failures of the functions it uses. 

The given type system \emph{has no side condition}. \Sil{The same holds for the other two type systems of this paper}, namely they have also no side condition for the rules  \textsf{(abs)} and \textsf{(app)}.


An $\Ell$-type assigned to a term represents the list of natural numbers corresponding to its free implicit names. 
For instance,
$`l\, \zero\,\underline{5}\,\two$ has $\Ell$-type $[1,4]$ since the $\Ell$- type of $\zero\,\underline{5}\,\two$ is clearly $[0, 2, 5]$ and to obtain the $\Ell$-type of $`l\, \zero\,\underline{5}\,\two$  one removes the $\zero$ which is bound and one decrements the other indices.
Moreover, it is a sorted list, as shown by the following Proposition~\ref{prop:sortLin}.

\begin{proposition}[Sortedness of lists]\label{prop:sortLin}
If $t : \ell$ then $\ell$ is sorted.
\end{proposition}
\begin{proof}
  $[i]$ is sorted and $\zip$ and $\downarrow$ preserve
  sortedness.
\end{proof}

\begin{example}[Typing terms]\label{typing_terms}
\begin{displaymath}
  \prooftree 
  \prooftree 
  \prooftree 
  \one : [1]
  \qquad
 \zero : [0]
  \justifies \one\,\zero : [0,1]
  \endprooftree 
  \justifies `l \one\,\zero : [0]
  \endprooftree 
  \justifies `l `l \one\, \zero : []
  \endprooftree 
  \qquad
  \prooftree 
  \prooftree 
  \two : [2]
   \qquad
  \zero : [0]
  \justifies \two\,\zero : [0,2]
  \endprooftree 
  \qquad
  \prooftree 
   \one : [1]
    \qquad
   \zero : [0]
    \justifies \one\,\zero : [0,1]
  \endprooftree 
  \justifies \two\,\zero \, (\one\,\zero) : ? 
  \using \Warning
  \endprooftree 
  \qquad
  \prooftree 
  \prooftree 
  \zero :[0]
  \justifies `l\zero : []
  \endprooftree 
  \justifies `l `l \zero : ?
  \using \Warning
  \endprooftree 
\end{displaymath}

Let us 
highlight the following facts: 
\begin{enumerate}
 
\item The term $\two\,\zero \, (\one\,\zero)$ is not $\Ell$-typeable since there are two free occurrences of index $\zero$.  We cannot merge lists $[0, 2]$ and $[0, 1]$, 
 thus $\two\,\zero \, (\one\,\zero)$ does not belong to $\Lin$.

\item The empty list $[]$ does not start with $0$, thus $`l`l\udl{0}$ is not $\Ell$-typeable, i.e. it does not belong to  $\Lin$.
\end{enumerate}
\end{example}

We prove further which combinators belong to $\Lin$ and which do not.

\Sil{\begin{proposition}[Combinators in $\Lin$]\label{prop:combinators} 
For the basic combinators ${\sf S} \equiv \lambda \lambda \lambda \two \zero (\one \zero) $, ${\sf K} \equiv \lambda \lambda \zero$, ${\sf W} \equiv \lambda \lambda \one \zero \zero $
 ${\sf B} \equiv \lambda \lambda \lambda \two (\one \zero)$, ${\sf C} \equiv \lambda \lambda \lambda \two\zero  \one$, and ${\sf I} \equiv \lambda 0$ the following holds:
\begin{enumerate}
	\item The combinators ${\sf S}$, ${\sf K}$ and ${\sf W}$ are not $\lin$-terms, i.e. ${\sf S},\; {\sf K},\; {\sf W}\;\not \in \Lin$.
	\item The combinators ${\sf B}$, ${\sf C}$ and ${\sf I}$ are $\lin$-terms, i.e. ${\sf B},\; {\sf C},\; {\sf I}\;  \in \Lin$, moreover ${\sf B}:[]$, ${\sf C}:[]$ and ${\sf I}:[]$ .
\end{enumerate}
\end{proposition}
\begin{proof}
	\begin{enumerate}
		\item Example~\ref{typing_terms} proves that ${\sf S} \equiv \lambda \lambda \lambda \two \zero (\one \zero) \not \in \Lin$.  Furthermore, ${\sf K} \equiv \lambda \lambda \one \not \in \Lin$ and  ${\sf W} \equiv \lambda \lambda  \one \zero \zero  \not \in \Lin$ since

 \begin{displaymath}
  \prooftree 
  \one :[1]
  \justifies `l\one : ?
    \using \Warning
  \endprooftree 
 \qquad
  \prooftree 
  \prooftree 
  \one : [1]
  \qquad
 \zero : [0]
  \justifies 
  \one\,\zero : [0,1]
  \endprooftree 
  \qquad
   \zero : [0]
  \justifies 
  \one\, \zero\, \zero : ?
   \using \Warning
  \endprooftree 
\end{displaymath}

		\item The proof-trees given below prove that  ${\sf B} \equiv \lambda \lambda \lambda \two (\one \zero) : []$, ${\sf C }\equiv \lambda \lambda \lambda \two \zero \one  :[]$	 and
${\sf I} \equiv \lambda \zero :[]$ 

\begin{displaymath}
\prooftree
 \prooftree
  \prooftree 
  \prooftree 
  \two : [2]
  \qquad
  \prooftree 
  \one : [1]
  \qquad
 \zero : [0]
  \justifies 
  \one\,\zero : [0,1]
  \endprooftree 
  \justifies 
  \two\, (\one\, \zero) : [0,1,2]
  \endprooftree 
  \justifies
  `l \two\, (\one\, \zero) : [0,1]
  \endprooftree 
  \justifies `l `l \two\, (\one\, \zero) : [0]
  \endprooftree 
   \justifies \lambda `l `l \two\, (\one\, \zero) : []
   \endprooftree   
  \qquad
 \prooftree
 \prooftree
  \prooftree 
  \prooftree 
  \prooftree 
  \two : [2]
  \qquad
 \zero : [0]
  \justifies 
  \two\,\zero : [0,2]
  \endprooftree 
  \qquad
   \one : [1]
  \justifies 
  \two\, \zero\, \one : [0,1,2]
  \endprooftree 
  \justifies
  `l \two\,  \zero\, \one : [0,1]
  \endprooftree 
  \justifies `l `l \two\,  \zero\, \one : [0]
  \endprooftree 
   \justifies \lambda `l `l \two\, \zero\, \one : []
   \endprooftree   
  \qquad
  \prooftree 
  \zero :[0]
  \justifies `l\zero : []
  \endprooftree 
\end{displaymath}
	
	\end{enumerate}
\end{proof}
}

\Sil{

\begin{proposition}
If $t$ is a  ${\sf B} {\sf C} {\sf I}$-term, i.e. a term generated by the combinators ${\sf B}, {\sf C}$ and ${\sf I}$, then $t : []$.
\end{proposition}

\begin{proof} By Proposition~\ref{prop:combinators} (2) and the fact that application is closed for terms typeable with empty lists, i.e. if $t:[]$ and $s:[]$, then $ts:[]$.
	
\end{proof}
}

Two lists are merged only in the rule {\sf (app)}. In order to successfully apply $\zip$, the two lists of free indices must be disjoint and there cannot be more than one occurrence of an index in the application. Moreover, in order to apply the rule {\sf (abs)}, the unique membership of the index to be abstracted is checked. In a term of type $[]$ all the indices are abstracted, then the check for unique membership is made for all of them. This justifies the following definition of linearity of $\lambda$-terms with implicit names.

\begin{definition}[Linearity of $\lambda$-terms]
  A $\lambda$-term $t$ is said to be \emph{linear}  if $t : []$.
\end{definition}

%



\Sil{To this end, the notion of $\Ell$-typeability has enabled a direct characterisation of 
linearity of  $\lambda$-terms with implicit names. As discussed in the Introduction (Section~\ref{sec:introduction}) the direct method for checking linearity of $\lambda$-terms with explicit names does not work for $\lambda$-terms with implicit names and the indirect method which involves the translation of $\lambda$-terms with implicit names back to $\lambda$-terms with  explicit names followed by linearity check on the later is algorithmically costly. 
}


\paragraph{Remark}\label{sec:red-Lin}
We do not treat reduction in $\Lin$, or more precisely,
reduction using implicit substitution.  We will treat fully reduction in
the framework of explicit substitution in $\Luin$
(Section~\ref{sec:luin}) and in the extended language $\LRdBin$ (Section~\ref{sec:LRdB}).  Consequently the reader will find no
$`b$-reduction and no statement of a theorem of type preservation in this section.
For a discussion the reader is invited to look at 
Section~\ref{subsec:reduction}.

\section{Restricted terms with explicit substitution $\Luin$}\label{sec:luin}
\newcommand{\RulesLuin}{
 \noindent
\begin{displaymath}
\begin{array}{rcl@{\quad}l}
 \thickhline
(\lambda t_1)\,t_2 &\luarone& t_1\{t_2,0\}& (\mathsf{B}_{in})\\
(t_1\,t_2)\upd{i} &\luarone& t_1\upd{i}\,t_2 \upd{i}& (\mathsf{App}_{\upd{}})\\
(t_1\,t_2)\{t_3,i\} &\luarone& t_1\{t_3,i\}\,t_2\{t_3,i\} & (\mathsf{App}_{\{\}})\\
(\lambda t)\upd{i} &\luarone& \lambda (t\upd{i+1})& (\mathsf{Lambda}_{\upd{}})\\
(\lambda t_1)\{t_2,i\} &\luarone& \lambda (t_1\{t_2,i+1\})& (\mathsf{Lambda}_{\{\}})\\
 \udl{0}\{t,0\} &\luarone& t&  (\mathsf{FVar}_{\{\}})\\
 \udl{n+1}\{t,0\} &\luarone& \udl{n} & (\mathsf{RVar}_{\{\}})\\
 \udl{0}\{t,i+1\} &\luarone& \udl{0}& (\mathsf{FVarLift}_{\{\}})\\
\udl{n+1}\{t,i+1\} &\luarone& \udl{n}\{t,i\} \upd{0}& (\mathsf{RVarLift}_{\{\}})\\
\udl{0}\upd{i+1} & \luarone& \udl{0} &(\mathsf{FVarLift}_{\upd{}})\\
\udl{n+1}\upd{i+1} & \luarone& \udl{n} \upd{i}\upd{0} &(\mathsf{RVarLift}_{\upd{}})\\
\udl{n}\upd{0} & \luarone& \udl{n+1}  &(\mathsf{VarShift}_{\upd{}})\\[3pt]
\thickhline
\end{array}
\end{displaymath}}


In this section, we focus on terms with implicit names and explicit substitution. 
We start from the $\lambda{\upsilon}$-calculus, a
simple calculus with explicit substitution introduced by Lescanne in~\cite{DBLP:conf/popl/Lescanne94}.  First, we modify the syntax and define restricted terms, dubbed $\luin$-terms, by typeability with $\Ell$-types. We then prove $\Ell$-type preservation under reduction. The design of the language is inspired by~\cite{Lescanne95WADT}.

The set of plain $\lu$-terms, denoted by $\Lu$, is given by the following syntax:
\begin{eqnarray*}
  t &::=& \udl{n} \mid \lambda t \mid tt \mid t[s]\\
  s &::=& t/ \mid \; \lift(s) \mid \;\shift
\end{eqnarray*}

A term $t$  can be a natural number $\udl{n}$ (i.e. a de Bruijn index), an abstraction, an application or a substituted term, where a substitution can be one of the following three: a slash $t/$, a lift $\lift(s)$ or a shift $\shift$.

The rewriting rules of the $\lu$-calculus are given in Figure~\ref{fig:rewriting rules for Lu}.

\begin{figure}[!h]
\begin{displaymath}
  \begin{array}{rcl@{\quad}l}
    \thickhline
(`l t_1) \, t_2 &\luar& t_1 [ t_2/ ] &\mathsf{(B)}\\
 (t_1\,t_2)[s] &\luar& (t_1[s]) \, (t_2 [s]) & \mathsf{(App)}\\
 (`l t) [s] &\luar& `l (t [ \lift(s)])  &\mathsf{(Lambda)}\\
 \udl{0} [t/] &\luar& t & \mathsf{(FVar)}\\
 \udl{n+1}[t/] &\luar& \udl{n}  &\mathsf{(RVar)}\\
 \udl{0}[\lift(s)] &\luar& \udl{0} & \mathsf{(FVarLift)}\\
 \udl{n+1} [\lift(s)] &\luar& \udl{n}[s][\shift] & \mathsf{(RVarLift)} \\
    \udl{n}[\shift]&\luar& \udl{n+1} & \mathsf{(VarShift)}\\
 \thickhline
  \end{array}
\end{displaymath}
\caption{The rewriting system for $\lu$-calculus}\label{fig:rewriting rules for Lu}
\end{figure}%

\begin{example}[$`b$-reduction]\label{ex:b}
  Let us write $t\downarrow_\upsilon$ the normalisation of $t$ using the set of rules
  \begin{displaymath}
\upsilon \ = \ \{\mathsf{(App)}, \mathsf{(Lambda)}, \mathsf{(FVar)}, \mathsf{(RVar)}, \mathsf{(FVarLift)}, \mathsf{(RVarLift)}, \mathsf{(VarShift)} \}.
\end{displaymath}
Said otherwise, one eliminates in $t$ all the occurrences of
substitutions. Given a $`b$-redex $(`l t_1) \, t_2$, its
$`b$-reduction is the term $(t_1 [ t_2/ ])\downarrow_\upsilon$.  This means that after the substitution $[ t_2/ ]$ has been introduced, it is eliminated. 
\end{example}

In what follows $\luarplus$ is the transitive closure of the rewriting relation $\luar$.
In order to characterise linearity by a type system, we consider two kinds of objects:  
\begin{itemize}
\item $[\lift^i(\shift)]$ is called an \emph{updater} and abbreviated as $\upd{i},\;i=0,1,...$, whereas
\item $[\lift^i(t/)]$ is called simply a \emph{substitution} and abbreviated as $\{t,i\},\;i=0,1,...$.
\end{itemize}

   According to the introduced abbreviations, we propose an alternative syntax that will be used in the definition of terms $\Luin$:
\begin{eqnarray*}
  t &::=& \udl{n} \mid \lambda t \mid tt \mid t\upd{i} \mid t\{t,i\}
\end{eqnarray*}

\noindent Furthermore, we propose an alternative rewriting system for $\lu$-calculus, given in Figure~\ref{fig:rewriting rules for Luin}, which is in accordance with the new syntax introduced above.

\begin{figure}[!h]

\medskip

\RulesLuin  

\caption{An alternative rewriting system for $\lu$-calculus}\label{fig:rewriting rules for Luin}
\end{figure}

\begin{proposition}
The rewriting system given by the  rules in Figure~\ref{fig:rewriting rules for Luin} is computationally equivalent to the rewriting system of $\lu$ given in Figure~\ref{fig:rewriting rules for Lu}.
\end{proposition}

\begin{proof} 
	
  Let us show, on four rules\Pier{, as a paradigm for the others,} the computational equivalence of the systems of Figure~\ref{fig:rewriting rules for Lu} and
  Figure~\ref{fig:rewriting rules for Luin}.
  \begin{itemize}
    
  \item \emph{Consider rule $(\mathsf{B}_{in})$:} In Figure~\ref{fig:rewriting rules for
      Lu} and Figure~\ref{fig:rewriting rules for
      Luin}, $(\mathsf{B})$ and $(\mathsf{B}_{in})$ have the same left-hand side.  In Figure~\ref{fig:rewriting rules for
      Lu} the right-hand side is $t_1[t_2/]$ that is $t_1[\lift^0(t_2/)]$ which writes $t_1\{t_2,0\}$ in the syntax of Figure~\ref{fig:rewriting rules for Luin}.
  \item \emph{Consider rule $(\mathsf{Lambda}_{\upd{}})$:} Written in the syntax of Figure~\ref{fig:rewriting rules for
      Lu}, left-hand side $(\lambda t)\upd{i}$ is $(`l t)[\lift^i(\shift)] $ and right-hand side
    $\lambda (t\upd{i+1})$ is $`l(t[\lift^{i+1}(\shift)])$. Therefore right-hand side of rule
    $(\mathsf{Lambda}_{\upd{}})$ of Figure~\ref{fig:rewriting rules for Luin} is obtained from left-hand side of rule
    $(\mathsf{Lambda}_{\upd{}})$ by application of rule $(\mathsf{Lambda})$ of Figure~\ref{fig:rewriting rules for Lu}.
  \item \emph{Consider rule $(\mathsf{Lambda}_{\{\}})$:} Written in the syntax of Figure~\ref{fig:rewriting rules for
      Lu}, left-hand side $(\lambda t_1)\{t_2,i\}$ is $(`l t_1)[\lift^i(t_2/)] $ and right-hand side
    $\lambda (t_1\{t_2,i+1\})$ is \(`l(t_1[\lift^{i+1}(t_2/ )])\). Therefore right-hand side of rule
    $(\mathsf{Lambda}_{\{\}})$ of Figure~\ref{fig:rewriting rules for Luin} is obtained from left-hand side of rule
    $(\mathsf{Lambda}_{\{\}})$ by application of rule $(\mathsf{Lambda})$ of Figure~\ref{fig:rewriting rules for Lu}.
  \item \emph{Consider rule $(\mathsf{RVarLift}_{\{\}})$:} Written in the syntax of Figure~\ref{fig:rewriting rules for
      Lu}, left-hand side $\udl{n+1}\{t,i+1\}$ is $\udl{n+1}[\lift^{i+1}(t/)] $ and right-hand side
    $\udl{n}\{t,i\} \upd{0}$ is $\udl{n}[\lift^i(t/)][\shift]$. Therefore right-hand side of rule
    $(\mathsf{RVarLift}_{\{\}})$ of Figure~\ref{fig:rewriting rules for Luin} is obtained from left-hand side of rule
    $(\mathsf{RVarLift}_{\{\}})$ by application of rule $(\mathsf{RVarLift})$ of Figure~\ref{fig:rewriting rules for Lu}.
\end{itemize}

\end{proof}

\subsection{$\Ell$-types for $\Luin$}
\label{sec:type-Luin}

Just like in the case of $\Lin$, $\Ell$-types for $\Luin$ provide information on free indices of a  $\luin$-term. In a declaration $t:\ell$, the type $\ell$ represents a sorted  list of free indices of~$t$. In order to define the language $\Luin$ we need to extend the notion of $\Ell$-types with a new operation and the notion of filters on $\Ell$-types.

The operation \Pier{\textit{increment,}} denoted by $\mapup$,
    increments all the elements of a list and should not be confused with $\downarrow$, defined in Definition~\ref{def:decrement}, which decrements the indices of a list.

\begin{definition}[Increment]\label{def:increment} The unary operation $\mapup$ is defined as follows:
\begin{eqnarray*}
	\mapup []   &=& [] \\
 \mapup (i :: \ell)  &=& (i+1) ::~ \mapup \ell
\end{eqnarray*}
\end{definition}
\Sim{By $\mapup^i$ we denote $i$ applications of the operation $\mapup$, i.e. $\mapup^0 \ell = \ell$ and $\mapup^{i+1} = \mapup \left(\mapup^i \ell\right)$.} Since \emph{shift} is not defined in the same context as \emph{increment}, there is no confusion despite both use the same symbol.


\subsubsection*{Filters on lists}
\label{sec:mod}
In order to ease list manipulation, we introduce filters on lists.
Given a predicate $p$ on naturals and a list $\ell$, then  $(p\,\mid\,\ell)$ is the list filtered by the predicate.
\begin{eqnarray*}
  (p \mid []) &=& [] \\
  (p \mid i :: \ell) &=& \mathsf{if~} p(i) \mathsf{~then~} i :: (p \mid \ell) \mathsf{~else~} (p \mid \ell)
\end{eqnarray*}

We will consider three basic predicates
\begin{eqnarray*}
  <i &=^{def}& \bm{\lambda} k \;.\; k < i\\ 
  >i &=^{def}& \bm{\lambda} k \;.\; k > i\\ 
\ge i &=^{def}& \bm{\lambda} k \;.\; k \ge i
\end{eqnarray*}

We can modify a predicate $p$ into $p_{[i "<-" e]}$ which is $p$ in which
each free occurrence of $i$ is replaced by $e$.  
Assume that predicates are made of
\begin{itemize}
\item constants,
\item free variables,
\item basic predicates $<i$, $>i$, and $\ge i$,
\item logical connectors,
\item functions on the naturals, like $\bm{\lambda} k \;.\; k +1$
\end{itemize}
we define $p_{[i"<-"e]}$ by induction as follows (\Sim{we denote an expression by $expr$ and} we
  assume $k$ is not the same variable as $i$):
\begin{itemize}
\item $(<expr)_{[i"<-"e]} =^{def} \bm{\lambda} k . k < expr_{[i"<-"e]}$
\item $(>expr)_{[i"<-"e]} =^{def} \bm{\lambda} k . k > expr_{[i"<-"e]}$
\item $(\ge expr)_{[i"<-"e]} =^{def} \bm{\lambda} k . k \ge expr_{[i"<-"e]}$
\item $(p \vee q)_{[i"<-"e]} =^{def} p_{[i"<-"e]} \vee q_{[i"<-"e]}$
\item $(p \wedge q)_{[i"<-"e]} =^{def} p_{[i"<-"e]} \wedge q_{[i"<-"e]}$
\end{itemize}
The substitution in expressions over the naturals is done as usual, as the substitution in universal algebra.

\begin{example}[Predicates]~
  \begin{itemize}
  \item $(<3\mid [0,2,3,4]) = [0,2]$
  \item $(\ge 3 \mid  [0,2,3,4]) = [3,4]$
  \item $(> 3 \mid  [0,2,3,4]) = [4]$
  \item $\uparrow (\ge 3 \mid  [0,2,3,4]) = \; \uparrow [3,4] = [4,5]$
  \item $\downarrow(\ge 3 \mid  [0,2,3,4])  = \; \downarrow [3,4] = [2,3]$
  \item $(< (i+1))_{[i"<-"j+1]} = \; < ((j + 1) + 1) = \; < (j + 2)$
  \item $(\ge (i+1))_{[i"<-"0]} = \; \ge 1 $
  \end{itemize}
\end{example}
Now, we can prove the following auxiliary lemma, containing list related properties needed in the proof of type preservation.

\begin{lemma}\label{lem:on-zip}~ Let $\ell, \ell_1$, $\ell_2$ and $\ell_3$ be sorted lists. The following equations hold, if all lists that appear in the equations are defined.
\begin{enumerate}[label=\alph*)]
\item $\ell_1 \zip \ell_2 = \ell_2 \zip \ell_1$;
  \item $\ell_1 \zip (\ell_2 \zip \ell_3) = (\ell_1 \zip \ell_2) \zip \ell_3$;
\item $(p\mid\ell_1) \zip (p\mid\ell_2) ~=~ (p\mid\ell_1 \zip \ell_2)$;
\item $\mapup \ell_1 \zip \mapup \ell_2  ~=~ \mapup(\ell_1 \zip \ell_2)$;
\item $\mapdown \ell_1 \zip \mapdown \ell_2  ~=~ \mapdown(\ell_1 \zip \ell_2)$;
\item $\mapup(p\mid\ell)= (p_{[i"<-"i+1]}\mid\mapup \ell)$;
\item $\mapdown(p_{[i "<-" i+1]}\mid\ell)= (p\mid\mapdown \ell)$. 
\end{enumerate}
\end{lemma}
\begin{proof}
  \newcommand{\IF}{\textsf{if~}}
  \newcommand{\THEN}{\textsf{~then~}}
    \newcommand{\ELSE}{\textsf{~else~}}
  \begin{description}
  \item[a)] Cases $[] \zip \ell$ and $\ell \zip []$ are by definition.
    Consider the case ${(n_1 ::  \ell_1) \zip (n_2 :: \ell_2)}$ with $n_1 <
    n_2$:
    \begin{eqnarray*}
      (n_1 ::  \ell_1) \zip (n_2 :: \ell_2) &=& n_1 :: (\ell_1 \zip (n_2 :: \ell_2)) \quad
                                                \textrm{(by definition)}\\
                                            &=& n_1 :: ((n_2 :: \ell_2) \zip \ell_1) \qquad \textrm{(by induction)} \\
      &=& (n_2 :: \ell_2) \zip (n_1 :: \ell_1) \quad
                                                \textrm{(by definition)}.
    \end{eqnarray*}
    Case $n_2 < n_1$ is symmetric.
  \item[b)] Cases where at least one of $\ell_1$, $\ell_2$ or $\ell_3$ is $[]$ are easy.
    For the general case, consider ${n_1 < n_2 < n_3}$.  The other cases
    are on the same pattern.
    \begin{eqnarray*}
      (n_1 ::  \ell_1) \zip ((n_2 :: \ell_2) \zip (n_3 :: \ell_3)) %
      &=&  (n_1 :: \ell_1) \zip (n_2 :: (\ell_2 \zip (n_3 :: \ell_3))) \\
      &=& n_1 :: (\ell_1 \zip (n_2 :: (\ell_2 \zip (n_3 :: \ell_3)))\\
      &=& n_1 :: (\ell_1 \zip ((n_2 :: \ell_2) \zip (n_3 :: \ell_3))) \\
      &=& n_1 :: ((\ell_1 \zip (n_2 :: \ell_2)) \zip (n_3 :: \ell_3)) \\
      &=& (n_1 :: (\ell_1 \zip (n_2 :: \ell_2))) \zip (n_3 :: \ell_3)\\
      &=& ((n_1 ::  \ell_1) \zip (n_2 :: \ell_2)) \zip (n_3 :: \ell_3)
    \end{eqnarray*}
  \item[c)] By case and induction 
    \begin{eqnarray*}
      (p \mid []) \zip (p \mid \ell) &=& [] \zip (p \mid \ell) = (p \mid \ell) = (p \mid [] \zip \ell)\\[5pt]
      (p \mid i :: \ell) \zip (p \mid []) &=& (p \mid i :: \ell) \zip []\\ 
            &=& (p \mid i :: \ell )\\
            &=& (p \mid (i :: \ell) \zip [])
    \end{eqnarray*}
    General case and sub-case $i_1 < i_2$ and $\neg p(i_1)$:
    \begin{eqnarray*}
      (p \mid i_1 :: \ell_1) \zip (p \mid i_2 :: \ell_2) &=& (p \mid \ell_1) \zip (p \mid i_2 :: \ell_2)\\
                                                         &=& (p \mid \ell_1 \zip (i_2 :: \ell_2)) = (p \mid i_1 :: (\ell_1 \zip (i_2 :: \ell_2))) \\
      &=& (p \mid (i_1 :: \ell_1) \zip (i_2 :: \ell_2))
    \end{eqnarray*}
    by induction.  Sub-case $i_1 < i_2$ and $p(i_1)$:
    \begin{eqnarray*}
      (p \mid i_1 :: \ell_1) \zip (p \mid i_2 :: \ell_2) &=& (i_1 :: (p \mid \ell_1)) \zip (p \mid i_2 :: \ell_2)\\
                                                         &=& i_1 :: ((p \mid \ell_1) \zip (p \mid i_2 :: \ell_2))\\
                                                         &=& i_1 :: (p \mid \ell_1 \zip (i_2 :: \ell_2)) = (p \mid i_1 :: (\ell_1 \zip (i_2 :: \ell_2)))\\
                                                         &=& (p \mid (i_1 :: \ell_1) \zip (i_2 :: \ell_2)).
    \end{eqnarray*}
    The sub-cases $i_2 < i_1$ ($\neg p(i_2)$ and $p(i_2)$) are similar.
  \item[d)] Also by case and induction
    \begin{displaymath}
      (\mapup []) \zip (\mapup \ell) = [] \zip (\mapup \ell) = \mapup\ell
      = \mapup ([] \zip \ell) 
    \end{displaymath}
    \begin{displaymath}
   ( \mapup (i :: \ell)) \zip (\mapup []) = (\mapup (i :: \ell)) \zip [] = \mapup (i :: \ell) = \mapup ((i :: \ell) \zip [])
  \end{displaymath}
  case $i_1 < i_2$ (hence $(i_1 + 1) < (i_2 + 1)$):
  \begin{eqnarray*}
    (\mapup (i_1 :: \ell_1)) \zip (\mapup (i_2 :: \ell_2)) &=& ((i_1 + 1) :: \mapup \ell_1) \zip ((i_2 + 1) :: \mapup \ell_2) \\
    &=& (i_1 + 1) :: ((\mapup \ell_1) \zip ((i_2 + 1) :: \mapup \ell_2)) \\
                                                       &=& (i_1 + 1) :: ((\mapup \ell_1) \zip (\mapup (i_2 :: \ell_2))) \\
    &=& (i_1 + 1) :: \mapup (\ell_1 \zip (i_2 :: \ell_2))\\
    &=& \mapup (i_1  :: (\ell_1 \zip (i_2 :: \ell_2))) = ~\mapup~ ((i_1 :: \ell_1) \zip (i_2 :: \ell_2))
  \end{eqnarray*}
  
  Case $i_2 < i_1$ is similar.
\item[e)] This proof is  similar to the proof of $d)$.

\item[f)] Also by case and induction
  \begin{eqnarray*}
    \mapup (p \mid []) &=& \mapup [] = [] = (p_{[i"<-"i+1]} \mid []) = (p_{[i"<-"i+1]} \mid \uparrow [])
  \end{eqnarray*}
  Sub-case $\neg p(i)$, hence $\neg p_{[i"<-"i+1]}(i+1)$ 
  \begin{displaymath}
    \begin{array}[h]{rcl@{\qquad}l}
      \mapup (p \mid i :: \ell) &=& \mapup (p \mid \ell) &\textrm{by definition of the filter~}\\
                                &=& (p_{[i"<-"i+1]} \mid \mapup \ell) & \textrm{by induction}\\
                                &=& (p_{[i"<-"i+1]} \mid (i+1) :: ~\mapup \ell)& \textrm{by definition of the filter}\\
                                &=& (p_{[i"<-"i+1]} \mid \mapup (i :: \ell)) &\textrm{by definition of~} \mapup
    \end{array}
  \end{displaymath}
  Sub-case $p(i)$, hence $p_{[i"<-"i+1]}(i+1)$
  \begin{displaymath}
    \begin{array}[h]{rcl@{\qquad}l}
      \mapup (p \mid i :: \ell) &=& \mapup (i :: (p \mid \ell)) &\textrm{by definition of the filter}\\
                                &=& (i+1) :: ~\mapup (p \mid \ell) & \textrm{by definition of~} \mapup\\
                                &=& (i+1) :: (p_{[i"<-"i+1]} \mid ~\mapup \ell) &\textrm{by induction}\\
                                &=& (p_{[i"<-"i+1]} \mid ~\mapup (i :: \ell))&\textrm{by definition of the filter}
    \end{array}
  \end{displaymath}
\item[g)] Works like \emph{f)}.
  \end{description}
\end{proof}

\begin{definition}[Terms $\Luin$]
A $\luin$-term is a plain $\lu$-term that can be $\Ell$-typed by the rules of Figure~\ref{fig:typesForLuin}.
The set of all $\luin$-terms is denoted by $\Luin$.
\end{definition}
\begin{figure}[h]\centering
  \Ovalbox{
    \begin{minipage}{.9\textwidth}
      \begin{center}
\begin{math}
   (ind) ~  \prooftree \phantom{ \udl{n}:[n]} \justifies \udl{n}:[n] \endprooftree
  \hspace*{.3\textwidth}
   (abs) ~  \prooftree  t:0::\ell \justifies \lambda t :\; \mapdown \ell  \using 
   \endprooftree
\end{math}

\medskip

\begin{math}
 (app) ~  \prooftree t_1:\ell_1\qquad t_2:\ell_2  \justifies  t_1t_2 : \ell_1 \zip \ell_2
  \endprooftree
   \hspace*{.3\textwidth}
 (upd) ~  \prooftree t:\ell  \justifies t\upd{i}: (<i\mid\ell) \plpl \mapup(\geq i\mid\ell) \endprooftree
\end{math}

\medskip

\begin{math}
  (sub_{\in}) ~
  \prooftree
  t_1:\ell_1 \qquad t_2:\ell_2  \using {\scriptstyle i`:\ell_1}
  \justifies t_1\{t_2,i\}:((<i\mid\ell_1) \plpl \mapdown(>i\mid\ell_1)) \zip \mapup^i \ell_2
  \using {\scriptstyle i`:\ell_1}
  \endprooftree
\end{math}

\medskip

\begin{math}
  (sub_{\notin}) ~
  \prooftree
  t_1:\ell_1 \qquad t_2:\ell_2
  \justifies t_1\{t_2,i\}:((<i\mid\ell_1) \plpl \mapdown(>i\mid\ell_1))
  \using {  \scriptstyle i`;\ell_1}
  \endprooftree
\end{math}
\end{center}
\end{minipage}}
\caption{Typing rules for $\Luin$}\label{fig:typesForLuin}
\end{figure}

About rule $(sub_\notin)$ we may notice that we assume $t_2:\ell_2$ although this assumption in not used in the
consequence.\footnote{A similar situation occurs with intersection types for explicit
  substitution~\cite{DBLP:journals/iandc/LengrandLDDB04}, with the rule \textsf{(drop)}.}
This guarantees the well-formedness of $t_2$.

Like for $\Lin$, we 
observe that in the typing tree of a $\Ell$-typed closed term, we meet only sorted lists with unique
occurrence of free indices.
This suggests to define 
linearity in $\Lu$ by $\Ell$\=/typeability.


\begin{definition}[Linearity of $\lu$-terms]\label{def:lu-linear}
A $\lu$-term $t$ is said to be linear if $t:[]$.
\end{definition}

\subsection{Reduction of $\Luin$}\label{subsec:reduction}

Let us recall that the $\luin$-calculus is equipped with the rewrite system given in Figure~\ref{fig:rewriting rules for Luin}.
 In the following theorem, we prove that $\Ell$-types are preserved by reduction.

\begin{theorem}[$\Ell$-type preservation]\label{thm:type-pres-lu}
If $t:\ell$ and $t \luarone t'$, then $t':\ell$.
\end{theorem}
\begin{proof}
  Assume that $t$ matches the left-hand side of one of the rules given in Figure~\ref{fig:rewriting rules for Luin}.  

  \begin{description}
  \item[$\mathsf{B}_{in} : $] ~ $(`l t_1) t_2 \luarone t_1\{t_2, 0\} $
    
    \medskip
    
    The left-hand side and the right-hand side of the rule can be typed as follows:
    
    \medskip
    
    \prooftree
    \prooftree
    t_1 : 0 :: \ell_1
    \justifies `l t_1 :~ \downarrow \ell_1
    \endprooftree
    \qquad
    t_2 : \ell_2
    \justifies (`l t_1)\,t_2 :~ \downarrow \ell_1 \zip \ell_2
    \endprooftree
    \hfill
    \prooftree
    t_1 : 0 :: \ell_1
    \qquad
    t_2 : \ell_2
    \justifies t_1\{t_2, 0\} : (< 0 \mid \ell_1) \plpl \downarrow (> 0 \mid \ell_1) \zip \uparrow^0 \ell_2
    \endprooftree

\medskip

Successfully typing the left-hand side means $\downarrow \ell_1 \cap \ell_2 = []$. If this is the case, then
$((< 0 \mid \ell_1) \zip \downarrow (> 0 \mid \ell_1)) \cap \uparrow^0 \ell_2 = []$ holds, so the right-hand side can
be successfully typed.

The equality
 \begin{math}
 ~ \downarrow \ell_1 \zip \ell_2   = (< 0 \mid \ell_1) \plpl \downarrow (> 0 \mid \ell_1) \zip \uparrow^0 \ell_2
 \end{math}
comes from
    \begin{itemize}
    \item $(<0 \mid \ell_1) = []$
    \item $(>0 \mid \ell_1) = \ell_1$ 
    \item $ \uparrow^0 \ell_2 =  \ell_2$
    \end{itemize}
    \bigskip
    
  \item[$\mathsf{App}_{\upd{}} : $]~ $(t_1 t_2)\upd{i} \luarone t_1 \upd{i} t_2 \upd{i}$
  
  \medskip
  
   For the left-hand side of the rule we get

    \medskip
    
    \prooftree
    \prooftree
    t_1 : \ell_1 \qquad t_2 : \ell_2
   \justifies t_1\,t_2 : \ell_1 \zip \ell_2
    \endprooftree
    \justifies (t_1\,t_2)\upd{i}: (<i\mid\ell_1 \zip \ell_2) \plpl \mapup(\geq i\mid\ell_1 \zip \ell_2)
    \endprooftree

    \bigskip
    
    For the right-hand side we get

    \medskip
    
    \prooftree
    \prooftree
    t_1:\ell_1
   \justifies t_1\upd{i} : (<i\mid\ell_1) \plpl \mapup(\geq i\mid\ell_1)
   \endprooftree
   \quad
   \prooftree
   t_2 : \ell_2
   \justifies t_2 \upd{i}: (<i\mid\ell_2) \plpl \mapup(\geq i\mid\ell_2)
    \endprooftree
    \justifies (t_1\upd{i}) \, (t_2\upd{i}): ((<i\mid\ell_1) \plpl \mapup(\ge i\mid\ell_1)) \zip ((<i\mid\ell_2) \plpl (\mapup(\ge i\mid\ell_2))
    \endprooftree

  \bigskip
  
  Typing the left-hand side, means $\ell_1 \cap\ell_2 = []$. As a consequence,
    $((<i\mid\ell_1) \zip \mapup(\ge i\mid\ell_1)) \cap ((<i\mid\ell_2) \zip (\mapup(\ge i\mid\ell_2))= []$ holds, so
    the right-hand side can be successfully typed.  
From Lemma~\ref{lem:on-zip}, we conclude that
  \begin{eqnarray*}
    (<i\mid\ell_1 \zip \ell_2) \plpl \mapup(\geq i\mid\ell_1 \zip \ell_2) &=& %
     ((<i\mid\ell_1) \zip (<i\mid \ell_2)) \plpl (\mapup(\geq i\mid\ell_1) \zip \mapup(\geq i\mid \ell_2)) \\
    &=& ((<i\mid\ell_1) \plpl \mapup(\geq i\mid\ell_1)) \zip ((<i\mid\ell_2) \plpl (\mapup(\geq i\mid\ell_2))
  \end{eqnarray*}
  
  \bigskip
\item[$\mathsf{App}_{\{\}} : $]~ $(t_1 t_2)\{t_3, i \} \luarone t_1 \{t_3 , i \} t_2 \{t_3, i \}$

\medskip

  For the left-hand side, with $i`:\ell_1$ we get

    \medskip
    
    \prooftree
    \prooftree
    t_1 : \ell_1 \qquad t_2 : \ell_2
   \justifies t_1\,t_2 : \ell_1 \zip \ell_2
   \endprooftree
   \qquad\qquad
   t_3 : \ell_3
      \justifies (t_1\,t_2)\{t_3,i\}: (<i\mid\ell_1 \zip \ell_2) \plpl \downarrow(> i\mid\ell_1 \zip \ell_2) \zip \uparrow^i \ell_3
    \using {\scriptstyle i`:\ell_1}
    \endprooftree

    \bigskip
    
    If the right-hand side is successfully typed, then $\ell_1 \cap \ell_2 = []$, and since 
    $i `: \ell_1$, then $i`; \ell_2$.   For the right-hand side, we get

    \medskip
    
    \prooftree
    \prooftree
    t_1:\ell_1  \qquad t_3 : \ell_3
    \justifies t_1\{t_3,i\} : (<i\mid\ell_1) \plpl \downarrow(> i\mid\ell_1) \zip\uparrow^i \ell_3
    \using { \scriptstyle i`:\ell_1}
   \endprooftree
   \quad
\prooftree
    t_2:\ell_2 \qquad t_3 : \ell_3
    \justifies t_2\{t_3,i\} : (<i\mid\ell_2) \plpl \downarrow(> i\mid\ell_2)
    \using {\scriptstyle i`;\ell_2}
    \endprooftree
    \justifies (t_1\{t_3,i\}) \, (t_2\{t_3,i\}): ((<i\mid\ell_1) \plpl \downarrow(> i\mid\ell_1)) \zip\uparrow^i \ell_3 \zip ((<i\mid\ell_2) \plpl \downarrow\,(> i\mid\ell_2))
    \endprooftree

    \bigskip

Like the previous cases, it is straightforward to show that whenever the left-hand side is typeable, the right-hand side is typeable as well.

 Here also, from Lemma~\ref{lem:on-zip}, we get
    
  \begin{math}
    (<i\mid\ell_1 \zip \ell_2) \plpl \downarrow(> i\mid\ell_1 \zip \ell_2) \zip\uparrow^i \ell_3 = ((<i\mid\ell_1) \plpl \downarrow(> i\mid\ell_1))\zip\uparrow^i \ell_3 \zip
    ((<i\mid\ell_2) \plpl (\downarrow(> i\mid\ell_2)).
  \end{math}
  
 The case $i`; \ell_1$, $i`: \ell_2$ is similar. Let us look now at case $i`;\ell_1\zip \ell_2$. For the left-hand side we have
  
    \medskip
    
    \prooftree
    \prooftree
    t_1 : \ell_1 \qquad t_2 : \ell_2
   \justifies t_1\,t_2 : \ell_1 \zip \ell_2
   \endprooftree
  \quad
   t_3 : \ell_3
      \justifies (t_1\,t_2)\{t_3,i\}: (<i\mid\ell_1 \zip \ell_2) \plpl \downarrow(> i\mid\ell_1 \zip \ell_2)
    \using {\scriptstyle i`;\ell_1\zip \ell_2}
    \endprooftree

    \medskip
    For the right-hand side we have

       \prooftree
    \prooftree
    t_1:\ell_1 \quad t_3 : \ell_3
    \justifies t_1\{t_3,i\} : (<i\mid\ell_1) \plpl \downarrow(> i\mid\ell_1)
    \using {\scriptstyle i`;\ell_1}
   \endprooftree
   \quad
\prooftree
    t_2:\ell_2 \quad t_3 : \ell_3 
    \justifies t_2\{t_3,i\} : (<i\mid\ell_2) \plpl \downarrow(> i\mid\ell_2)
    \using { \scriptstyle i`;\ell_2}
    \endprooftree
    \justifies (t_1\{t_3,i\}) \, t_2\{t_3,i\}: ((<i\mid\ell_1) \plpl \downarrow(> i\mid\ell_1)) \zip ((<i\mid\ell_2) \plpl (\downarrow(> i\mid\ell_2)) 
    \endprooftree
    
    \medskip 
    
Whenever $\ell_1 \cap \ell_2 = []$ holds and we can type the left-hand side of the rule, $((<i\mid\ell_1) \plpl \downarrow(> i\mid\ell_1)) \cap ((<i\mid\ell_2) \plpl (\downarrow(> i\mid\ell_2)) = []$  holds and the right-hand side of the rule can be typed.

  From Lemma~\ref{lem:on-zip} we conclude that we obtain equal types for both left-hand side and right-hand side of the rule.
    
    \medskip
\item[$\mathsf{Lambda}_{\upd{}} : $]~  $(`l t)\upd{i} \luarone `l (t \upd{i+1})$

  \medskip
Left-hand and right-hand sides of the rule can be typed as follows:

\medskip

  \prooftree
  \prooftree
  t : 0 :: \ell
  \justifies `l t : ~ \downarrow \ell
  \endprooftree
  \hspace*{5pt}
\justifies (`l t)\upd{i} : (< i \mid \downarrow \ell) \zip \uparrow (\ge i \mid \downarrow \ell) 
\endprooftree
\hfill
\prooftree
\prooftree
t : 0 :: \ell
\justifies t\upd{i+1} : 0 :: (< i+1 \mid \ell) \zip \uparrow (\ge i+1 \mid \ell)
\endprooftree
\justifies `l(t \upd{i+1}) :~ \downarrow((< i+1 \mid \ell) \zip \uparrow (\ge i+1 \mid \ell))
\endprooftree

\bigskip

The equality $ (< i \mid \downarrow \ell) \zip \uparrow (\ge i \mid \downarrow \ell) =~ \downarrow((< i+1 \mid \ell) \zip \uparrow (\ge i+1 \mid \ell)) $  is a consequence of
Lemma~\ref{lem:on-zip} d) e) and g).

\medskip

\item[$\mathsf{Lambda}_{\{\}} : $]~ $(`l t_1)\{t_2, i \} \luarone `l (t_1 \{t_2, i+1 \})$
  
  \medskip
  
  First, we consider case $i+1`; \ell_1$ (with the same calculation as in the case $\mathsf{Lambda}_{\upd{}}$):
  
  \medskip

  \prooftree
  \prooftree
  t_1 : 0 :: \ell_1
  \justifies `l t_1 : \; \downarrow \ell_1
  \endprooftree  \quad t_2 : \ell_2
  \justifies (`l t_1)\{t_2,i\} : ~ (< i \mid \downarrow \ell_1) \zip \downarrow (> i \mid \downarrow \ell_1)
  \using {\scriptstyle i `; \downarrow \ell_1}
  \endprooftree

  \bigskip

  \prooftree
  \prooftree
  t_1 : 0 :: \ell_1 \quad t_2 : \ell_2
  \justifies t_1\{t_2,i+1\} : 0 :: ((< i+1 \mid \ell_1) \zip \downarrow (> i+1 \mid \ell_1))
 \using {\scriptstyle  
   i+1 `; \ell_1}
  \endprooftree
  \justifies `l (t_1\{t_2,i+1\}) : \;\downarrow ((< i+1 \mid \ell_1) \zip \downarrow (> i +1 \mid \ell_1))
  \endprooftree

\medskip

From Lemma~\ref{lem:on-zip} we can conclude that the type of the term on the left-hand side of the rule and the type of the term on the right hand-side of the rule are equal.

  \bigskip

  Next, let us look at the case $i+1 `: \ell_1$

 \medskip

   \prooftree
  \prooftree
  t_1 : 0 :: \ell_1
  \justifies `l t_1 : \downarrow \ell_1
  \endprooftree
  \qquad
   t_2 : \ell_2
   \justifies (`l t_1)\{t_2,i\} : (< i \mid \downarrow \ell_1) \zip \downarrow (> i \mid \downarrow \ell_1) \zip \uparrow^i \ell_2
  \using {\scriptstyle i `: (\downarrow \ell_1)}
  \endprooftree

  \bigskip

  \prooftree
  \prooftree
  t_1 : 0 :: \ell_1 \qquad
  t_2 : \ell_2
  \justifies t_1\{t_2,i+1\} : 0 :: (< i+1 \mid \ell_1) \zip \downarrow (> i+1 \mid \ell_1) \zip \uparrow^{i+1} \ell_2
  \using {\scriptstyle i+1 `: 0 :: \ell_1}
  \endprooftree
  \justifies `l (t_1\{t_2,i+1\}) :~ \downarrow ((< i+1 \mid \ell_1) \zip \downarrow (> i+1 \mid \ell_1)) \zip \uparrow^i \ell_2
  \endprooftree
  
  \medskip

  \medskip
  
\item[$\mathsf{FVar}_{\{\}} : $] ~ $\udl{0}\{t, 0 \} \luarone t$
 
  \medskip

  \prooftree
  \udl{0} : [0] \qquad t : \ell
  \justifies \udl{0}\{t,0\} : (< 0 \mid [0]) \zip (> 0 \mid [0]) \zip \uparrow^0 \ell 
  \using {\scriptstyle  0 `: [0]}
  \endprooftree
  \qquad\qquad 
  $t : \ell$

  \medskip

The equality $(< 0 \mid [0]) \zip (> 0 \mid [0]) \zip \uparrow^0 \ell = \ell$ comes from the fact that : $(< 0 \mid [0]) \zip (> 0 \mid [0]) = []$.  

  \medskip
\item[$\mathsf{RVar}_{\{\}} : $]~ $\udl{n+1}\{t, 0\} \luarone \udl{n}$

 \medskip
 
  \prooftree
  \udl{n+1} : [n+1] \quad t : \ell
   \justifies \udl{n+1}\{t,0\} :  (<0 \mid [n+1]) \zip \downarrow (> 0
   \mid [n+1]) 
   \using {\scriptstyle 0 `; [n+1]}
   \endprooftree
   \qquad \qquad
 $\udl{n} : [n]$

   \medskip

   The equality of the types comes from the fact that 
   $(<0 \mid [n+1]) = []$ and $\downarrow (> 0 \mid [n+1]) = [n]$.

   \medskip
   
 \item[$\mathsf{FVarLift}_{\{\}} : $]~ $\udl{0} \{t, i+1 \} \luarone \udl{0}$

   \medskip

   \prooftree
   \udl{0} : [0]  \quad t : \ell
   \justifies \udl{0}\{t, i+1\} : (<i+1 \mid [0]) \zip \downarrow (> i+1 \mid [0]) 
   \using {\scriptstyle i+1 `; [0]}
   \endprooftree
   \qquad \qquad \qquad \qquad
   \udl{0} : [0]

   \medskip

   The equality of the types comes from $(<i+1 \mid [0]) = [0]$ and $\downarrow (> i+1 \mid [0]) = []$.

   \medskip
   
 \item[$\mathsf{RVarLift}_{\{\}} : $]~ $\udl{n+1}\{t, i+1 \} \luarone \udl{n}\{t, i \}\upd{0} $
 
 \medskip
  We will consider three cases, depending on the numbers $i$ and $n$. First, we consider the case where  $i < n$

   \medskip

 \prooftree
   \udl{n+1} : [n+1]  \quad t : \ell
   \justifies \udl{n+1}\{t,i+1\} :  (<i+1 \mid [n+1]) \zip \downarrow (>i+1 \mid [n+1])
   \using {\scriptstyle i+1 `; [n+1]}
   \endprooftree
   \qquad \qquad
   \prooftree
   \prooftree
   \udl{n} : [n] \quad t : \ell
   \justifies \udl{n}\{t, i\} : \; \downarrow[n]
   \using {\scriptstyle i `; [n]}
   \endprooftree
\justifies \udl{n}\{t, i\}\upd{0} : [n]
   \endprooftree

\medskip
Since $i < n$, we have $i+1 < n+1$, and it holds that $(<i+1 \mid [n+1]) = []$ and $ \downarrow (> i+1 \mid [n+1]) = \downarrow [n+1] = [n]$, so the types are equal.
   \medskip

   Next, we consider the case where $i = n$.

   \medskip

   \prooftree
   \udl{n+1} : [n+1]  \qquad \qquad t : \ell
   \justifies \udl{n+1}\{t,i+1\} :\; ((< i+1 \mid [n+1]) \zip \downarrow (> i+1 \mid [n+1])) \zip \uparrow^{i+1} \ell
   \using {\scriptstyle i+1 `: [n+1]}
   \endprooftree
   \qquad \qquad
   \prooftree
   \prooftree
   \udl{n} : [n] \qquad \qquad t : \ell
   \justifies \udl{n}\{t,i\} : \; \uparrow^i \ell
   \using {\scriptstyle i `: [n]}
   \endprooftree
   \justifies \udl{n}\{t,i\}\upd{0} :\; \uparrow^{i+1} \ell
   \endprooftree
\medskip

From $i = n$, we obtain $i+1 = n+1$, and it follows that  $(< i+1 \mid [n+1]) \zip \downarrow (> i+1 \mid [n+1]) = []$, hence the types are equal.
   \medskip

 Finally, we consider the case $i > n$.

   \medskip

 \prooftree
   \udl{n+1} : [n+1] \qquad \qquad t : \ell
   \justifies \udl{n+1}\{t,i+1\} : (< i+1 \mid [n+1]) \zip \downarrow (> i+1 \mid [n+1]) 
   \using {\scriptstyle i+1 `; [n+1]}
   \endprooftree
   \qquad \qquad
   \prooftree
   \prooftree
   \udl{n} : [n]  \qquad \qquad t : \ell
   \justifies \udl{n}\{t, i\} : [n]
   \using {\scriptstyle i `; [n]}
   \endprooftree
\justifies \udl{n}\{t, i\}\upd{0} : [n+1]
   \endprooftree

\medskip
Since $i>n$, we have $i+1 > n+1$, and it follows that $\downarrow (>i+1 \mid [n+1]) = []$. Hence, the types are equal.

\medskip
We see that in all three cases we have typed both the term on the left-hand side and the term on the right-hand side of the rule with the same type.

\medskip

\item[$\mathsf{FVarLift}_{\upd{}} : $]~ $\udl{0}\upd{i+1} \luarone \udl{0}$

Left-hand side and right-hand side of the rule can be typed as follows:
\medskip

\prooftree
\udl{0} : [0]
\justifies \udl{0}\upd{i+1} : (< i+1 \mid [0]) \zip \mapup (\geq i+1 \mid [0])
\endprooftree 
 \qquad \qquad  
  \udl{0} : [0]
  
\medskip

From Lemma~\ref{lem:on-zip} we have $(\geq i+1 \mid [0]) = []$, thus $\mapup (\geq i+1 \mid [0]) = []$. From the latter and $(< i+1 \mid [0]) = [0]$ we obtain $((< i+1 \mid [0]) \zip \mapup (\geq i+1 \mid [0])) = [0]$.
\medskip


\item[$\mathsf{RVarLift}_{\upd{}} : $]~ $\udl{n+1}\upd{i+1} \luarone \udl{n}\upd{i}\upd{0}$

\medskip
Left-hand side and right-hand side of the rule can be typed as follows:
\medskip

\prooftree
\udl{n+1} : [n+1]
\justifies \udl{n+1}\upd{i+1} : (< i+1 \mid [n+1]) \zip \mapup(\geq i+1 \mid [n+1])
\endprooftree \qquad \qquad
\prooftree
\prooftree
\udl{n} : [n]
\justifies \udl{n}\upd{i} : (< i \mid [n]) \zip \mapup (\geq i \mid [n])
\endprooftree
\justifies \udl{n}\upd{i}\upd{0} : \mapup \ ((< i \mid [n]) \zip \mapup (\geq i \mid [n]))
\endprooftree
\medskip

From Lemma~\ref{lem:on-zip} we get\\
$(< i+1 \mid [n+1]) \zip \mapup(\geq i+1 \mid [n+1]) = \; \mapup ((< i \mid [n]) \zip \mapup (\geq i \mid [n]))$

\medskip

\item[$\mathsf{VarShift}_{\upd{}} : $]~ $\udl{n}\upd{0} \luarone \udl{n+1}$

\medskip
Left-hand side and right-hand side of the rule can be typed as follows:

\medskip

\prooftree
\udl{n} : [n]
\justifies \udl{n}\upd{0} : (< 0 \mid [n]) \zip \mapup (\geq 0 \mid [n])
\endprooftree\qquad \qquad
$\udl{n+1} : [n+1]$

\medskip

Since we have that
\begin{itemize}
\item $(< 0 \mid [n]) = []$,
\item $(\geq 0 \mid [n]) = [n]$, and
\item $\mapup [n] = [n+1]$,
\end{itemize}
it follows that $((< 0 \mid [n]) \zip \mapup( \geq 0 \mid [n])) = [n+1]$.

  \end{description}

\end{proof} 

Let us 
point out that this constructive proof of $\Ell$-type preservation enables a constructive evaluator for terms in~$\Luin$.
Currently this implementation is in \textsf{Haskell}.  Indeed Theorem~\ref{thm:type-pres-lu} and
Definition~\ref{def:lu-linear} entail the correctness of $\luin$-calculus.  
If we remember that 
linearity is defined as $\Ell$-typedness
  when the type is~$[]$, we obtain the following corollary.


\begin{corollary}[Preservation of linearity] If a $\lu$-term $t$ is linear and $t \luarone t'$ then $t'$ is linear.
\end{corollary}

An open question is whether this proof of preservation of linearity for a $\lambda$-calculus with implicit names (de
  Bruijn indices) yields easily a proof of preservation of linearity for a calculus with explicit names (nominal
  logic~\cite{DBLP:journals/iandc/Pitts03}). We refer to the discussion of Berghofer and
  Urban~\cite{DBLP:journals/entcs/BerghoferU07} to let the reader figure out the work that remains to be done.

\section{Extended terms with resource control $\LRdB$}
\label{sec:LRdB}

In the previous sections resource control
was achieved by \emph{restricting} the sets $\Lambda$ and $\Lu$ to $\Lin$ and $\Luin$, respectively.
In this section, we take a dual approach and \emph{extend} $\Lambda$ with explicit operators performing erasure and duplication on terms in order to obtain full resource control.  The goal is to design a language capable to linearise  all $`l$-terms. We adapt $\Ell$-types and use them to define terms $\LRdBin$ and to characterise  linear terms in $\LRdB$.

The abstract syntax of $\lRdB$-terms,  plain terms with resources and implicit names, is generated by the following grammar:
\[ t, s ::= \ind{n}{\alpha} \mid `l t \mid t\; s \mid \era{n}{\alpha}{t}  \mid \dup{n}{\alpha}{t} \]
where $\ind{n}{\alpha}$ is a \emph{®\=/index}, $\odot$ denotes \emph{the erasure} of index in a term, and $\triangledown$ denotes \emph{the duplication} of index in a term. The set of all $\lRdB$-terms is denoted by $\LRdB$.

\subsubsection*{®-indices} An ®-index is the pair $\ind{n}{\alpha}$, where $n$ is a natural number and $\alpha$ is a string of booleans. For convenience, we will use the following abbreviations: $\bO \equiv false$ and $\bI \equiv true$.  Therefore $`a$ will be a string of $0$'s and $1$'s. Whether $0$ and $1$ refer to natural numbers or to booleans will be easily distinguished; so we consider that using those notations will introduce no confusion. In $(\udl{n}, \alpha)$, $\udl{n}$ corresponds to an index in $`L$ and $`a$ represents duplications of the index. The empty string of booleans, corresponding to absence of duplications, is denoted by~$\varepsilon$. For instance, if $(\udl{n}, \varepsilon)$ is  duplicated, it is represented by $\ind{n}{\bO}$ and $\ind{n}{\bI}$; if it is triplicated, it can be represented by $\ind{n}{\bO}$, $\ind{n}{{\bI}{\bO}}$ and $\ind{n}{{\bI} {\bI}}$ (or by $\ind{n}{\bO\bO}$, $\ind{n}{\bO \bI}$ and $\ind{n}{\bI}$).

In the following example and in Subsection~\ref{sec:bestiary} we introduce informally notions corresponding to $\lRdB$-terms, which will be formally defined in Subsection~\ref{sec:typ-LRdB}.
\begin{example}\label{exa:R} ~
  \begin{itemize}

  \item The term $`l x.  y$ is represented in $\LRdB$ by the term\\
    $`l \era{0}{`e}{\ind{1}{`e}}$.
  \item The term $`l x . (x (`l y . x y))$ is represented in $\LRdB$ by the term\\
    $`l (\dup{0}{`e}{(\ind{0}{0}\,(`l \ind{1}{1}\,\ind{0}{`e}))})$.
  \item The linear term $`l x . `l y . x \,y$ is represented in $\LRdB$ by
    the term $`l `l \ind{1}{`e}\, \ind{0}{`e}$, that has neither $\triangledown$ nor $\odot$, since it is linear and needs no resource control.
    \Jel{\item The open term with multiple occurrences of a free variable $x \, z \, (y \, z)$
    is represented in $\LRdB$ by the term\\
    $\dup{0}{`e} {(\ind{2}{`e} \ind{0}{0} (\ind{1}{ `e} \ind{0}{ 1}))}$ where $\ind{2}{\varepsilon}$ represents the free variable $x$, $\ind{1}{\varepsilon}$ represents the free variable $y$, and $\ind{0}{0}$ and $\ind{0}{1}$ correspond to the two occurrences of the variable $z$.
}
\item Term $`l x . x\,x\,x$ is discussed in Example ~\ref{exa:Lxxx}.
  \end{itemize}
\end{example}
Several more examples of $\lRdB$-terms will be elaborated in the following subsection.

\subsection{A bestiary of $\lRdB$-terms}
\label{sec:bestiary}

In this section, we examine basic and well known terms.

\medskip\noindent\textbf{The term \textsf{I}}
\begin{eqnarray*}
  \mathsf{I} &=& `l \ind{0}{`e}.
\end{eqnarray*}
This corresponds to the term $`l x.x$ in the $\lambda$-calculus with explicit
names. $\ind{0}{`e}$ means that there is no $`l$ between the ®-index
$\ind{0}{`e}$ and its binder and that there is no duplication.

\medskip\noindent\textbf{The term \textsf{K}}
\begin{eqnarray*}
  \mathsf{K} &=& `l `l \era{0}{`e}{\ind{1}{`e}}.
\end{eqnarray*}
In $\lambda$-calculus, $\mathsf{K}$ is written $`l x .`l y . x$.  In $\Lambda$, \textsf{K} is
written $`l `l \udl{1}$.  The index $\udl{0}$ does not occur in~$\udl{1}$, but since
we want $\LRdB$-terms to be linear, we make it to occur anyway, thus we
write $\era{0}{`e}{\ind{1}{`e}}$.  Notice that $`e$ is the second
component of all the ®-indices since there is no duplication.   Recall that the term $ `l `l \era{0}{`e}{\ind{1}{`e}} $ correspond to the term $`l `l\udl{1}$. The translation from  $\LR$ to  $`L$ will be introduced in Section~\ref{sec:read}.
 This term also corresponds to the term
\begin{displaymath}
  `l x.`l y.\eraG{y}{x}
\end{displaymath}
using the notations of \cite{DBLP:journals/corr/GhilezanILL14} and to the term 
\begin{displaymath}
`l x.`l y.\eraL{y}{x}.
\end{displaymath}
using the notations of~\cite{DBLP:journals/iandc/KesnerL07}.

\medskip\noindent\textbf{The term \textsf{S}}
\begin{eqnarray*}
  \mathsf{S}&=& `l `l `l \dup{0}{`e} {(\ind{2}{`e} \ind{0}{0} (\ind{1}{ `e} \ind{0}{ 1}))}
\end{eqnarray*}
In $\lambda$-calculus, $\mathsf{S}$ is written $`l x. `l y .`l z . x z (y z)$ and in
$\Lambda$, $\mathsf{S}$ is written $`l `l `l (2\, 0\, (1\, 0))$.  We notice
the double occurrence of $z$ in $\lambda$-calculus and of $\udl{0}$ in $\Lambda$.
Therefore a duplication is necessary.  From the ®-index $\ind{0}{`e}$ it
creates two indices $\ind{0}{0}$ and $\ind{0}{1}$.  Where the second
component $0$ is the string of length $1$ made of $0$ alone and the second
component $1$ is the string of length $1$ made of $1$ alone.  This term
can be written using the notations of~\cite{DBLP:journals/corr/GhilezanILL14} as the term
\begin{displaymath}
`l x.`l y.`lz.(\dupG{z}{z_{0}}{z_{1}}{x\,z_{0}\,(y\,z_{1})})
\end{displaymath}
or using the notations of~\cite{DBLP:journals/iandc/KesnerL07} as the term
\begin{displaymath}
  `l x.`l y.`lz.(\dupL{z}{z_{0}}{z_{1}}{x\,z_{0}\,(y\,z_{1})}).
\end{displaymath}
\medskip\noindent\textbf{The term \textsf{5}}
\begin{displaymath}
  \mathsf{5} ~=~ `l `l (\dup{1}{`e}{\dup{1}{0}{\dup{1}{00}{\dup{1}{000}{(\ind{1}{0000}(\ind{1}{0001} (\ind{1}{001} (\ind{1}{01} (\ind{1}{1} \ind{0}{`e})))))}}}})
\end{displaymath}
\textsf{5} represents the Church numeral $5$. Recall that in $\lambda$-calculus,
\textsf{5} is written $ `l f. `l x. (f (f (f (f (f \,x)))))$ and in
$\Lambda$, $`l `l (1 (1 (1 (1 (1 \,0)))))$.  Since $\udl{1}$ is repeated five
times, we need four duplications.  If we compute \textsf{5} other
ways, we can get other forms.  For instance, as the result of $\mathsf{3}
+ \mathsf{2}$:
\begin{displaymath}
  `l `l ((\dup{1}{`e} {\dup{1}{0} {\dup{1}{00}{ (\ind{1}{000} (\ind{1}{001}
    (\ind{1}{01} \dup{1}{1}{ (\ind{1}{10} (\ind{1}{11} \ind{0}{`e}))})))}}}))
\end{displaymath}
or as the result of $\mathsf{2} + \mathsf{3}$:
\begin{displaymath}
  `l `l (\dup{1}{`e}{\dup{1}{0}{  (\ind{1}{00} (\ind{1}{01} \dup{1}{1}{ \dup{1}{10}
  {(\ind{1}{100} (\ind{1}{101} (\ind{1}{11} \ind{0}{`e})))}}))}})
\end{displaymath}
or as the result of $\mathsf{3} + 1 + 1$:
\begin{displaymath}
  \lambda\lambda\dup{1}{\varepsilon}{\dup{1}{0}{\dup{1}{00}{(\ind{1}{000} \dup{1}{001}{(\ind{1}{0010} (\ind{1}{0011} (\ind{1}{01} (\ind{1}{1} \ind{0}{\varepsilon}))))})}}}
\end{displaymath}
The four above forms  correspond to the same term in $`L$, namely
\begin{math}
`l `l (1 (1 (1 (1 (1 \,0))))).
\end{math} The translation $\mathsf{readback}$ from $\LR$ to $`L$ will be defined in Section~\ref{sec:read}.

\medskip\noindent\textbf{The terms \textsf{ff} and \textsf{tt}}

The ®-term \textsf{ff} (i.e. the boolean false) is $\lambda(\era{0}{`e}{\lambda\ind{0}{`e}})$
and the ®-term \textsf{tt}  (i.e. the boolean true, that is also the
combinator \textsf{K}) is $\lambda(\lambda\era{0}{`e}{\ind{1}{`e}})$.

\medskip\noindent\textbf{The Curry fixpoint combinator}

The Curry fixpoint combinator \textsf{Y} is:
\begin{displaymath}
  \mathsf{Y} = \lambda\dup{0}{`e}((\lambda(\ind{1}{0}\,\dup{0}{`e}(\ind{0}{0}\,\ind{0}{1})))\,(\lambda(\ind{1}{1}\,\dup{0}{`e}(\ind{0}{0}\,\ind{0}{1}))))
\end{displaymath}
and in notations of~\cite{DBLP:journals/corr/GhilezanILL14}:
\begin{displaymath}
`l x.(\dupG{x}{x_{0}}{x_{1}}{`l y.(x_{0}\,(\dupG{y}{y_{0}}{y_{1}}{y_{0}\,y_{1}}))\,`l y.(x_{1}\,(\dupG{y}{y_{0}}{y_{1}}{y_{0}\,y_{1}}))})
\end{displaymath}
or using the notations of~\cite{DBLP:journals/iandc/KesnerL07}:
\begin{displaymath}
`l x.(\dupL{x}{x_{0}}{x_{1}}{`l y.(x_{0}\,(\dupL{y}{y_{0}}{y_{1}}{y_{0}\,y_{1}}))\,`l y.(x_{1}\,(\dupL{y}{y_{0}}{y_{1}}{y_{0}\,y_{1}}))})
\end{displaymath}

\subsection{$\Ell$-types for $\LRdBin$}
\label{sec:typ-LRdB}
  
In this setting, lists of  ®-indices are called $\Ell$-types for $\LRdBin$.

  \begin{definition}[$\Ell$-types for $\LRdBin$] The abstract syntax of $\Ell$-types for $\LRdBin$  is given by
\begin{eqnarray*}
    \ell  &::=& [] \mid (\udl{n}, \alpha) :: \ell
\end{eqnarray*}
where $(\ind{n}{\alpha}$ is a ®-index.
\end{definition}

Operations $\zip$ and $\mapdown$ are defined in Section~\ref{sec:LinLambda} for lists
on $\nat$. Here we apply $\zip$ to  lists of ®\=/indices. For that,
we have to define an order on the set of all ®\=/indices. We define first
an order on strings of booleans.
\begin{definition}[Order on strings of booleans]
  An order $<_L$ on strings of booleans is defined as
\begin{displaymath}
  \prooftree
  \justifies \bO :: \ell <_L \ \bI :: \ell
  \endprooftree
  \qquad\qquad
  \prooftree
  \justifies `e <_L b :: \ell
  \endprooftree
  \qquad\qquad
  \prooftree
  \ell_1 <_L \ell_2
  \justifies  b :: \ell_1 <_L b :: \ell_2
  \endprooftree
\end{displaymath}

\end{definition}

In other words, $<_L$ is the lexicographic extension on lists of the order
$\bO < \bI$. 

\begin{definition}[Order on ®-indices]
  An order $<^\circledR$ on  ®\=/indices is defined as 
  \begin{displaymath}
    \prooftree
    n_1 < n_2
    \justifies \ind{n_1}{`a_1} < ^\circledR \ind{n_2}{`a_2}
    \endprooftree
    \qquad \qquad 
    \prooftree
    `a_1 <_L `a_2
   \justifies \ind{n}{`a_1} <^{\circledR} \ind{n}{`a_2}
    \endprooftree
  \end{displaymath}
\end{definition}
In other words, $<^\circledR$ is the lexicographic product $< `* <_L$ of the orders
$<$, on the naturals and $<_L$ on strings of booleans. 
 By $\le^\circledR$ we denote the relation $<^\circledR$ \emph{or} $=$ and the relation $\le^\circledR$ is total.

\begin{definition}[Merge] A binary operation which merges two  lists of ®\=/indices is defined as follows:\label{def:merge}
\begin{eqnarray*}
[] \zip \ell & = & \ell \\
(\ind{n}{`a} :: \ell) \zip [] & = & \ind{n}{`a} :: \ell \\ 
(\ind{n_1}{`a_1} :: \ell_1) \zip (\ind{n_2}{`a_2} :: \ell_2) & = &
                                                                   \mathbf{~if~} \ind{n_1}{`a_1} <^\circledR \ind{n_2}{`a_2}  \mathbf{~then~} \ind{n_1}{`a_1} :: (\ell_1 \zip (\ind{n_2}{`a_2} :: \ell_2))   \\
&& \mathbf{~if~} \ind{n_2}{`a_2} <^\circledR \ind{n_1}{`a_1}  \mathbf{~then~} \ind{n_2}{`a_2} :: ((\ind{n_1}{`a_1} :: \ell_1) \zip  \ell_2)
\end{eqnarray*}
\end{definition}

\begin{remark} The function $\zip$ is {\bf not total}.

\end{remark}

If a list $\ell$ is an empty list or it contains only indices with strictly positive first component, we write $ \ell `: List^+$.

\begin{definition}[Decrement]\label{def:decrem} 
Given a list $\ell$, assume that we have a proof that $ \ell `: List^+$, we can define operation $\downarrow$ on this list:
\begin{eqnarray*}
\downarrow [] = []\\
  \downarrow (\ind{n+1}{`a} :: \ell) &=& \ind{n}{`a}  :: \downarrow \ell
\end{eqnarray*}

\end{definition}

All properties proved in Lemma~\ref{lem:on-zip} hold also for the lists of ®-indices. 
We omit the proof, due to the lack of space and the fact that it is analogous to the proof of Lemma~\ref{lem:on-zip}.
\\
By means of $\Ell$-typeability, we single out meaningful (well-formed) plain terms with resources and implicit names.

\begin{definition}[Terms $\LRdBin$]\label{def:Typ-LR} 
A $\lRdBin$-term is a plain $\lRdB$-term that can be $\Ell$-typed by the  rules of Figure~\ref{fig:typesForLRdB}.
\end{definition}

\begin{figure}[h]	
    \Ovalbox{
    \begin{minipage}{\textwidth}
  \begin{displaymath}
    (\mathsf{ind}) ~
    \prooftree \justifies \ind{n}{\alpha} : [\ind{n}{\alpha}]
    \endprooftree
    \qquad
    (\mathsf{abs}) ~
    \prooftree
    t : \ind{0}{`e} :: \ell
    \justifies `l t :~ \mapdown ~\ell
    \endprooftree
    \qquad 
    (\mathsf{app}) ~
    \prooftree
    t_1 : \ell_1 \qquad t_2 : \ell_2 
    \justifies t_1 \; t_2 : \ell_1 \zip \ell_2
    \endprooftree
  \end{displaymath}\medskip
  \begin{displaymath}
    (\mathsf{era})~
    \prooftree
    t : \ell
    \justifies \era{n}{\alpha}{t} : [\ind{n}{\alpha}]  \zip  \ell 
    \endprooftree
    \qquad\qquad
    (\mathsf{dup}) ~
    \prooftree
    t : \ell\, \zip\, [\ind{n}{\alpha\bO}, \ind{n}{\alpha\bI}]
    \justifies \dup{n}{\alpha}{t} : [ \ind{n}{\alpha} ] \zip \ell  
    \endprooftree   
  \end{displaymath}

  \medskip
\end{minipage}}
\caption{Typing rules for $\LRdBin$}\label{fig:typesForLRdB}
\end{figure}
The set of all $\lRdBin$-terms is denoted by $\LRdBin$.

The following example illustrates the Definition~\ref{def:Typ-LR} by $\Ell$-typing the $\LRdB$-term $\mathsf{S}\mathsf{K}$.
\begin{example} ~\\  
  \prooftree 
  \prooftree 
  \prooftree 
  \prooftree 
  \prooftree 
  \prooftree 
  \prooftree 
  \ind{2}{`e} : [\ind{2}{`e}] \qquad \ind{0}{\bO} : [\ind{0}{\bO}]
     \justifies \ind{2}{`e}\,\ind{0}{\bO} : [\ind{0}{\bO},\ind{2}{`e}]
  \endprooftree 
\qquad 
\prooftree 
 \ind{1}{`e} : [\ind{1}{`e}] \qquad \ind{0}{\bI} : [\ind{0}{\bI}]
   \justifies \ind{1}{`e}\,\ind{0}{\bI} : [\ind{0}{\bI}, \ind{1}{`e}]
  \endprooftree 
   \justifies \ind{2}{`e}\,\ind{0}{\bO}\,(\ind{1}{`e}\,\ind{0}{\bI}) : [\ind{0}{\bO}, \ind{0}{\bI}, \ind{1}{`e}, \ind{2}{`e}]
  \endprooftree 
   \justifies \dup{0}{`e}{(\ind{2}{`e}\,\ind{0}{\bO}\,(\ind{1}{`e}\,\ind{0}{\bI}))} : [\ind{0}{`e}, \ind{1}{`e}, \ind{2}{`e}]
    \endprooftree 
   \justifies `l (\dup{0}{`e}{(\ind{2}{`e}\,\ind{0}{\bO}\,(\ind{1}{`e}\,\ind{0}{\bI}))}) : [\ind{0}{`e}, \ind{1}{`e}]
   \endprooftree
  \justifies `l `l (\dup{0}{`e}{(\ind{2}{`e}\,\ind{0}{\bO}\,(\ind{1}{`e}\,\ind{0}{\bI}))}) : [\ind{0}{`e}]
 \endprooftree 
  \justifies `l `l `l (\dup{0}{`e}{(\ind{2}{`e}\,\ind{0}{\bO}\,(\ind{1}{`e}\,\ind{0}{\bI}))}) : []
 \endprooftree 
  \hspace*{-40pt}
  \prooftree 
  \prooftree 
    \prooftree 
   \ind{1}{`e} : [\ind{1}{`e}]
    \justifies \era{0}{`e}{\ind{1}{`e}} : [\ind{0}{`e}, \ind{1}{`e}]
    \endprooftree 
    \justifies `l \era{0}{`e}{\ind{1}{`e}} : [\ind{0}{`e}]
  \endprooftree 
  \justifies `l `l \era{0}{`e}{\ind{1}{`e}} : []
  \endprooftree 
  \justifies (`l`l`l
    (\dup{0}{`e}{(\ind{2}{`e}\,\ind{0}{\bO}\,(\ind{1}{`e}\,\ind{0}{\bI}))})\,(`l
    `l \era{0}{`e}{\ind{1}{`e}}) : []
  \endprooftree 
\end{example}

Notice that we abstract with $`l$ (see Definition~\ref{def:Typ-LR})
only ®\=/index of the form $(\ind{0}{\varepsilon}$. Further, the definition of $\zip$ ensures that in an $\Ell$-typed term an ®\=/index can occur  at most once (Definition~\ref{def:merge}).
The other binder, namely  duplication, binds two ®\=/indices of the form $\ind{n}{\alpha0}$ and $\ind{n}{\alpha1}$ and produces a new ®\=/index $\ind{n}{`a}$.
Closed terms are terms in which each ®\=/index is bound. 



\begin{proposition}[Closedness]
  If $t: []$ then $t$ is closed.
\end{proposition}
\begin{proof}
  If $t : \ell$, then $\ell$ is the set of free ®-indices in the
  term. Therefore, if $\ell$ is empty then $t$ has no free ®-index and $t$ is
  closed.
\end{proof}

  There are actually two rules which eliminate ®-indices, namely
  \textsf{abs} and \textsf{dup}.  But when \textsf{dup} eliminates two
  indices $\ind{n}{`a\bO}$ and $\ind{n}{`a\bI}$, it introduces
  $\ind{n}{`a}$.  Therefore if a term is closed, all the ®-indices are checked
  for linearity when abstracted by $`l$. This justifies the following definition of linearity.

\begin{definition}[Linearity of $\lRdB$-terms]\label{def:linear}
A $\lRdB$-term $t$ is said to be linear if $t:[]$.
\end{definition}


Similarly as in Section~\ref{sec:Lin-types} and \ref{sec:type-Luin}, the notion of $\Ell$-typeability has enabled the characterisation of 
linearity of $\lRdB$-terms. 

\subsection{Resource reduction in $\LRdB$}
\label{sec:resource}
We define rewriting rules for normal forms w.r.t. $\odot$ and
$\triangledown$, for which we prove $\Ell$-type preservation. Those rules which control resources are inspired
by~\cite{DBLP:journals/corr/GhilezanILL14}.  

First, we define replacement of an ®\=/index in a term. By $\repl{t}{\ind{n}{`a}}{\ind{m}{`b}}$  we denote a term obtained from term $t$ by replacing recursively the  ®\=/index $\ind{n}{`a}$ by $\ind{m}{`b}$.

\begin{definition}[Replacement]
  \emph{ Let
  us call} $cond(`a,`d,n,k)$ \emph{the condition}\\
\[{n\neq k ~\vee~ \forall `g`:\{0,1\}^* ~ `d \neq `a `g}.\] \emph{Notice that this
  can be written also} \[n\neq k ~\vee~ \neg (`a ~\textsf{prefix} ~`d).\]
  \emph{Replacement} $\repl{t}{\ind{n}{`a}}{\ind{m}{`b}}$\emph{ is defined as:}
  \begin{eqnarray*}
    \repl{\ind{n}{`a`g}}{\ind{n}{`a}}{\ind{m}{`b}} &=&\ind{m}{`b`g}\\
    \repl{\ind{k}{`d}}{\ind{n}{`a}}{\ind{m}{`b}} &=& \ind{k}{`d} \hspace{.33\textwidth} \textsf{if~} cond(`a,`d,n,k)\\
    \repl{(t_1\;t_2)} {\ind{n}{`a}}{\ind{m}{`b}} &=& \repl{t_1} {\ind{n}{`a}}{\ind{m}{`b}} \; \repl{t_2} {\ind{n}{`a}}{\ind{m}{`b}}\\
    \repl{`l t}{\ind{n}{`a}}{\ind{m}{`b}} &=&`l(\repl{t} {\ind{n+1}{`a}}{\ind{m+1}{`b}})\\
    \repl{(\ind{k}{\delta} * t)}{\ind{n}{`a}}{\ind{m}{`b}} &=& \repl{\ind{k}{\delta}}{\ind{n}{`a}}{\ind{m}{`b}} *  {\repl{t}{\ind{n}{`a}}{\ind{m}{`b}}}, \hspace{.03\textwidth} * \in \{\odot, \triangledown \}
  \end{eqnarray*}
\end{definition}
\medskip

The rewriting system controlling duplication and erasing in $\lRdB$-terms is given by the  rules in Figure~\ref{fig:rewriting rules for LRdB}. Basically, the resource rewriting rules propagate $\triangledown$ in the term and pull $\odot$ out of the term.

\begin{figure}[!h]
\begin{displaymath}
\begin{array}{rcl@{\quad}l}
\thickhline
 `l \era{n+1}{`a}{t} & \rightarrow_c & \era{n}{`a}{`l t} & ({\lambda} {-}\odot)\\
  \dup{n}{`a}(`l t) & \rightarrow_c  &  `l (\dup{n+1}{`a}{t}) &  (\triangledown {-} {\lambda})\\
(\era{n}{`a}{t_1}) \; t_2 & \rightarrow_c  & \era{n}{`a}{(t_1\;t_2)}  & (\text{AppL} {-} \odot)\\
 t_1 \; (\era{n}{`a}{t_2}) & \rightarrow_c  & \era{n}{`a}{(t_1\;t_2)} & (\text{AppR} {-}\odot)\\
\dup{n}{`a}{(t_1 \; t_2)} & \rightarrow_c  &  (\dup{n}{`a}{t_1}) \; t_2, \hspace*{.07\textwidth} \textbf{if } \ind{n}{`a0} `: t_1 \wedge \ind{n}{`a1} `: t_1  & (\text{AppL} {-} \triangledown) \\
\dup{n}{`a}{(t_1 \; t_2)} & \rightarrow_c  & t_1 \; (\dup{n}{`a}{t_2}), \hspace*{.07\textwidth} \textbf{if } \ind{n}{`a0} `: t_2 \wedge\ind{n}{`a1} `: t_2 & (\text{AppR} {-} \triangledown) \\
\era{n}{`a}{\era{m}{`b}{t}}  & \rightarrow_c  & \era{m}{`b}{\era{n}{`a}{t}}, \hspace*{.03\textwidth} \textbf{ if } n < m & (\odot {-}\odot)\\
  \dup{n}{`a}{\era{n}{`a1}{t}} & \rightarrow_c  & \repl{t}{\ind{n}{`a0}}{\ind{n}{`a}} & (\odot {-}\triangledown_1)\\
   \dup{n}{`a}{\era{n}{`a0}{t}} & \rightarrow_c  & \repl{t}{\ind{n}{`a1}}{\ind{n}{`a}} & (\odot {-}\triangledown_0)\\
  \dup{n}{`a}{\era{m}{`b}{t}} & \rightarrow_c  & \era{m}{`b}{\dup{n}{`a}{\,t}}, \hspace*{.04\textwidth} \mathbf{if~} n \neq m \vee (`b \neq  `a0 \wedge `b \neq `a1) & (\odot {-}\triangledown)\\
  \thickhline
\end{array}
\end{displaymath}
\caption{Resource rewriting system that controls $\odot$ and $\triangledown$ in  $\lRdB$-terms}\label{fig:rewriting rules for LRdB}
\end{figure}
\begin{theorem}[$\Ell$-type preservation]
If $t : \ell$ and $t \rightarrow_c  t' $, then $t' : \ell$.
\end{theorem}
\begin{proof}
Assume that $t$ matches the left-hand side of one of the rules in Figure~\ref{fig:rewriting rules for LRdB}. We consider the following two cases.
\begin{description}
\item[$(\bm{\lambda}\bm{-}\odot) : $]~ $`l \era{n+1}{`a}{t} \rightarrow_c  \era{n}{`a}{`l t}$

  Rule $\bm{\lambda}\bm{-}\odot$ preserves type. Indeed

\bigskip

\prooftree
\prooftree
t : \ind{0}{`e} :: \ell
\justifies \era{n+1}{`a}{t} :  [\ind{n+1}{`a}] \zip \ind{0}{`e} :: \ell
\endprooftree
\justifies `l \era{n+1}{`a}{t} :  [\ind{n}{`a}] \zip  (\;\downarrow \ell)
\endprooftree
\qquad\qquad
\prooftree
\prooftree
t : \ind{0}{`e} :: \ell
\justifies `l t : \downarrow \ell 
\endprooftree
\justifies \era{n}{`a}{`l t} :   [\ind{n}{`a}]\zip (\;\downarrow \ell)
\endprooftree

\vspace*{1pt}

\bigskip
Both the term on the left-hand side and the term on the right-hand side of the rule are typed with the same type.
\noindent

\medskip

\vspace*{1.5pt}

\item[$(\triangledown\bm{-}\bm{\lambda}) : $]~ $\dup{n}{`a}(`l t) \rightarrow_c  `l (\dup{n+1}{`a}{t})$

Rule $\triangledown\bm{-}\bm{\lambda}$ preserves type. Indeed


\bigskip

\prooftree
\prooftree
t :  [\ind{n+1}{\alpha\bO}), (\ind{n+1}{\alpha\bI}] \zip  (\ind{0}{`e} :: \ell) 
\justifies `l t :   [(\ind{n}{\alpha\bO}), (\ind{n}{\alpha\bI})] \zip  \;\downarrow \ell 
\endprooftree
\justifies \dup{n}{`a}{(`l t)} :  [\ind{n}{`a}] \zip  (\;\downarrow \ell)
\endprooftree
\qquad\qquad
\prooftree
\prooftree
t : [\ind{n+1}{\alpha\bO}, \ind{n+1}{\alpha\bI}] \zip  (\ind{0}{`e} :: \ell) 
\justifies \dup{n+1}{`a}{t} :  [\ind{n+1}{`a}] \zip (\ind{0}{`e} :: \ell)
\endprooftree
\justifies `l \dup{n+1}{`a}{t} : [\ind{n}{`a}] \zip (\;\downarrow \ell)
\endprooftree

\bigskip

Both the term on the left-hand side and the term on the right-hand side of the rule are typed with the same type.

\vspace*{1pt}

Proving that other rules preserve $\Ell$-type is straightforward.

\end{description}
\end{proof}


As a consequence, resource reduction preserves linearity.

\begin{corollary}[Preservation of linearity] If a $\lRdB$-term $t$ is linear and $t \rightarrow_c  t'$ then $t'$ is also linear.
\end{corollary}

\BigStepRep  
\subsection{Head-reduction and $\beta$-contraction in $\LRdB^{\Ell}$}
\label{sec:struc}
\noindent In $\LRdB^{\Ell}$, $`b_{\textrm{\textregistered}}$-contraction is
\begin{displaymath}
\begin{array}{rcl@{\quad}l}
(`l t_1)\, t_2 & "->" & u & \qquad \mathrm{\mathbf{where}\quad} t_1 [t_2 / \ind{0}{`e}] \Downs u \qquad  \qquad \qquad (`b_{\textrm{\textregistered}}) 
\end{array}
\end{displaymath}
The binary relation $\Downs$ is between an expression of the form $\sub{t_1}{t_2}{\ind{0}{`e}}$ and a
$\lRdB^{\Ell}$-term, known as the big-step semantics~\cite{DBLP:journals/scp/KokkeSW20} defined in Figure~\ref{fig:BSS}.  Intuitively the substitution
$[t_2 / \ind{0}{`e}]$ substitutes $t_2 : []$ (i.e., a closed and linear term) for the ®-indices $\ind{0}{`e}$. More generally
$[t / \ind{i}{`a}]$ substitutes $t : []$ for the ®\=/indices $\ind{i}{`a}$. The fact that the term $t$, which is
substituted, is of $\Ell$-type $[]$ (hence linear and closed) is crucial.  $`b_{\textrm{\textregistered}}$-contraction is the
key of head reduction in~$\LRdB^{\Ell}$. A~second operation close to \textsf{substitute} and called \textsf{replace} is defined in Figure~\ref{fig:BSR}. The goal is to have an operation which, unlike \textsf{substitute}, replaces ®-indices by terms without updating the  other ®-indices;  \textsf{replace} is used in the definition of \textsf{substitute}.   The reader may notice that the  premises of the rules of  Figure~\ref{fig:BSR} do not overlap and cover all the patterns, so the relation $\Downr$ is deterministic, similarly for the relation $\Downs$ of Figure~\ref{fig:BSS}.  Therefore the $u$ used in the definition of $(`b_{\textrm{\textregistered}})$ is unique and is always defined. 

This duality $\langle$\textsf{substitute} -- \textsf{replace}$\rangle$ reminds us \textsf{laziness} : the action done in the second run of the evaluation is lighter than the action done in the first run of the evaluation.  Clearly this connection should be deepened in order to investigate optimisation in the computation of the $`b_{\textrm{\textregistered}}$-contraction.  To our knowledge, this is the first time that such a mechanism is studied in the context of a $`l$-calculus with implicit names.

\BigStepSub 

\newcommand{\Sub}[3]{\ensuremath{#1[#2/#3]}}
\subsubsection{$\Ell$-type preservation for substitution}



\begin{figure}[!ht]
  \centering
  \Ovalbox{
    \begin{minipage}{.9\textwidth}
      \begin{displaymath}
        (Rep_{\in})~
        \prooftree
        t'\replace{t}{\ind{i}{`a}} \Downs u \quad t': \ell \zip [\ind{i}{`a}] \quad t : []
        \justifies u : \ell
        \endprooftree
      \end{displaymath}
      \bigskip
        \begin{displaymath}
               (Rep_{\notin})~
        \prooftree
        t'\replace{t}{\ind{i}{`a}} \Downs u \quad t': \ell \quad \ind{i}{`a} \notin \ell \quad t : []
        \justifies u : \ell
        \endprooftree
      \end{displaymath}
     \bigskip
        \begin{displaymath}
        (Sub_{\in})~
        \prooftree
        \sub{t'}{t}{\ind{i}{`a}} \Downs u \quad \  t': \ell\zip[\ind{i}{`a}] \quad t : []
         \justifies u :  (\le  i \mid \ell) \zip \down (> i \mid \ell)
         \endprooftree
       \end{displaymath}
       \bigskip
         \begin{displaymath}
 (Sub_{\notin})~
        \prooftree
        \sub{t'}{t}{\ind{i}{`a}} \Downs u \quad \  t': \ell \quad \ind{i}{`a} \notin \ell  \quad t : []
         \justifies u :  (\le  i \mid \ell) \zip \down (> i \mid \ell)
         \endprooftree
      \end{displaymath}
    \end{minipage}}
    \caption[(rep) and (sub)]{Preservation of $\Ell$-type by replace and substitute}
  \label{fig:Type-preserv}
\end{figure}

\begin{proposition}[Admissiblity of \mbox{\textit{Rep}}]\label{prop:lpresrep}
The rules $(Rep_{\in})$ and  $(Rep_{\notin})$ of Figure~\ref{fig:Type-preserv}  are admissible.
\end{proposition} 
\begin{proof}~
  The proof is by structural induction on $t'$ and by case for the rules of Figure~\ref{fig:typesForLRdB} used to type $t'$.

  * $(Rapp)$ Assume $t_1 : \ell_1 \zip [\ind{i}{`a}]$ and $t_2 : \ell_2$. By induction on $(Rep_{\in})$, $u_1 : \ell_1$. By induction on  $(Rep_{\notin})$, $u_2 : \ell_2$.  By (app) $u_1\,u_2 : \ell_1 \zip \ell_2$. The other cases  $t_1 : \ell_1$ and $t_2 : \ell_2 \zip [\ind{i}{`a}]$ and   $t_1 : \ell_1$ and $t_2 : \ell_2$ and $\ind{i}{`a} \notin \ell_1\zip \ell_2$ work the same way.

  * $(Rabs)$ If $`l\,t' : \down  \ell$ then $t':\ind{0}{`a} :: \ell$. By induction  $u:\ind{0}{`a} :: \ell$. Hence $`l\,u : \down \ell$.

  *  $(Rind_1)$ and $(Rind_2)$ and $(R\odot_1)$ are base cases, we have just to prove that conclusions are true statements.

  * $(R\odot_2)$ Assume $t':\ell \zip [\ind{i}{`a}]$. Hence $\era{j}{`b}{t'} : \ell \zip [\ind{i}{`a}] \zip [\ind{j}{`b}] $
  By induction $u :\ell$, hence $\era{j}{`b}{u}:\ell \zip[\ind{i}{`b}]$. If $\ind{i}{`a}\notin \ell$, the proof is similar.

  * $(R\triangledown_1)$  Assume $t':\ell$, $t' \replace{t}{\ind{i}{`a0}} \Downs v$ and $v \replace{t}{\ind{i}{`a1}} \Downr u$. Since $\dup{i}{`a}{t'}$ is a correct $\lRdB$-term, $\ind{i}{`a0} \in \ell$, $\ind{i}{`a1} \in \ell$ and $\ind{i}{`a} \notin \ell$.  Hence we may write $t' : \ell' \zip [\ind{i}{`a0}] \zip [\ind{i}{`a1}]$.
Then by induction $v : \ell \zip  [\ind{i}{`a1}]$. Then 
by $(Rep_{\in}$), $u:  \ell$ and by $(dup)$, $\sub{(\dup{i}{`a}{t'})}{t}{\ind{i}{`a}} : \ell$.

* $(R\triangledown_2)$ is on the same model.
  
\end{proof}
\begin{proposition}[Admissiblity of \mbox{\textit{Sub}}]\label{prop:lpressub}
The rules $(Sub_{\in})$ and $(Sub_{\notin})$ of Figure~\ref{fig:Type-preserv} are admissible.
\end{proposition}
\begin{proof}
  
  * $(Sapp)$ Assume $t_1 : \ell_1 \zip [\ind{i}{`a}]$ and $t_2:\ell_2$. Then $t_1\,t_2 : \ell_1\zip\ell_2 \zip  [\ind{i}{`a}]$. By induction for $(Sub_{\in})$ on $t_1$, $u_1: (\le i\mid\ell_1) \zip (> i\mid\ell_1) $ and by induction for $(Sub_{\notin})$ on $t_2$ $u_2: (\le i\mid\ell_2) \zip (> i\mid\ell_2)$ then $u_1 \,u_2 : (\le i\mid\ell_1 \zip\ell_2) \zip (> i\mid\ell_1 \zip \ell_2) $.   The case $t_1 : \ell_1$ and $t_2:\ell_2 \zip [\ind{i}{`a}]$ works alike.

  * $(Sabs)$ If $`l\, t' : \, \down \ell$, then $t':  \ind{0}{`e} :: \ell$. By induction $u: \ind{0}{`e} :: (\le i \mid \ell) \zip (>i \mid\ell)$, hence
  \begin{displaymath}
`l\,u: \, \down (\le i \mid \ell) \zip \down (>i \mid\ell) ~=~ (\le i \mid \down \ell) \zip \down (>i \mid \down\ell).
\end{displaymath}

  * $(Sind_1)$ and $(Sind_2)$ and $(Sind_3)$ and $(S\odot_1)$ are base cases, we have just to prove that conclusions are true statements. 

  * $(S\odot_2)$ Assume $i>j$ and $t':\ell$, then since $\era{i}{`a}{t'}$ is a correct $\lRdB$-term, $\ind{i}{`a}\notin \ell$ and $(Rep_{\notin})$ applies, then ${u:(\le j \mid \ell) \zip \,\down (>j \mid \ell)}$. Moreover $\sub{(\era{i}{`a}{t'})}{t}{\ind{j}{`b}} : \ell \zip [\ind{i}{` a}]$. Then
\begin{displaymath}
  \era{i-1}{` a}{u} : (\le j \mid \ell) \,\zip \,\down (>j \mid \ell) \zip [\ind {i-1}{`a}] ~=~ (\le j \mid \ell \zip [\ind{i}{` a}]) \zip \,\down (>j \mid \ell\zip [\ind{i}{` a}])
\end{displaymath}

  * $(S\odot_3)$ Assume $i \le j $ and $t':\ell$, then since $\era{i}{`a}{t'}$ is a correct $\lRdB$-term, $\ind{i}{`a}\notin \ell$ and $(Rep_{\notin})$ applies, then ${u:(\le j \mid \ell) \zip \,\down (>j \mid \ell)}$. Moreover $\sub{(\era{i}{`a}{t'})}{t}{\ind{j}{`b}} : \ell \zip [\ind{i}{` a}]$. Then
\begin{displaymath}
  \era{i}{` a}{u} : (\le j \mid \ell) \,\zip \,\down (>j \mid \ell) \zip [\ind {i}{`a}] ~=~ (\le j \mid \ell \zip [\ind{i}{` a}]) \zip \,\down (>j\mid \ell\zip [\ind{i}{` a}])
\end{displaymath}

* $(Sub\triangledown_1)$ Assume $t':\ell$, $\sub{t'}{t}{\ind{i}{`a0}} \Downs v$ and $v \replace{t}{\ind{i}{`a1}} \Downr u$. Since $\dup{i}{`a}{t'}$ is a correct $\lRdB$-term, $\ind{i}{`a0} \in \ell$, $\ind{i}{`a1} \in \ell$ and $\ind{i}{`a} \notin \ell$.  Hence we may write $t' : \ell' \zip [\ind{i}{`a0}] \zip [\ind{i}{`a1}]$.

\noindent Then by induction $v : (\le i \mid \ell') \zip \down (>i \mid \ell') \zip  [\ind{i}{`a1}]$. Then 
by $(Rep_{\in}$), $u:  (\le i \mid \ell') \zip \down (> i \mid \ell')$ and by $(dup)$, $\sub{(\dup{i}{`a}{t'})}{t}{\ind{i}{`a}} :  (\le i \mid \ell') \zip \down (> i \mid \ell')$.

* $(Sub\triangledown_2)$ Assume $\sub{t'}{t}{\ind{j}{`b}} \Downs u$ and $i > j$ and $t':\ell$. Since $\dup{i}{`a}{t'}$ is a correct $\lRdB$-term, $\ind{i}{`a0} \in \ell$ and  $\ind{i}{`a0} \in \ell$. Let us write $t':\ell' \zip [\ind{i}{`a0}] \zip  [\ind{i}{`a1}]$.
By (dup), $\dup{i}{`a}{t'} : \ell'\zip  [\ind{i}{`a}]$.
Beside, by induction
\begin{displaymath}
u: (\le j \mid \ell) \zip \down\, (>j \mid \ell) ~=~  (\le j \mid \ell') \zip \down\, (>j \mid \ell') \zip  [\ind{i-1}{`a0}] \zip  [\ind{i-1}{`a1}].
\end{displaymath}
Hence by (dup) $\dup{i-1}{`a}{u} :  (\le j \mid \ell') \zip \down\, (>j \mid \ell') \zip  [\ind{i-1}{`a}]$.

* $(Sub\triangledown_3)$ Assume $\sub{t'}{t}{\ind{j}{`b}} \Downs u$ and $i < j \vee (i=j= \wedge `a \neq `b)$ and $t':\ell$. Since $\dup{i}{`a}{t'}$ is a correct $\lRdB$-term, $\ind{i}{`a0} \in \ell$ and  $\ind{i}{`a0} \in \ell$. Let us write $t':\ell' \zip [\ind{i}{`a0}] \zip  [\ind{i}{`a1}]$.
By (dup), $\dup{i}{`a}{t'} : \ell'\zip  [\ind{i}{`a}]$.
Beside, by induction
\begin{displaymath}
  u: (\le j \mid \ell) \zip \down\, (>j \mid \ell) ~=~  (\le j \mid \ell') \zip \down\, (>j \mid \ell') \zip  [\ind{i}{`a0}] \zip  [\ind{i}{`a1}].
\end{displaymath}
Hence by (dup) $\dup{i-1}{`a}{u} :  (\le j \mid \ell') \zip \down\, (>j \mid \ell') \zip  [\ind{i}{`a}]$.
\end{proof}

\begin{proposition}[$\Ell$-type preservation for $`b_{\textrm{\textregistered}}$]
  $ `b_{\textrm{\textregistered}}$-contraction preserves $\Ell$-type, i.e., if $(\lambda t_1)t_2 \; \longrightarrow \; u$ and $(\lambda t_1)t_2: \ell $ then $u:\ell$.
  
\end{proposition}
\begin{proof}
The proof uses the admissible rule $(Sub_{\in})$
  \begin{displaymath}
    \prooftree
    \prooftree
    t_1 : \ind{0}{`e} :: \ell
    \justifies `l\, t_1 : \,\down \ell
    \endprooftree
    \qquad t_2 : []
    \justifies (`l\,t_1)~t_2 : (\down \ell) \zip []
    \endprooftree
    \qquad \qquad 
    \prooftree
 \Sub{t_1}{t_2}{\ind{0}{`e}} \Downarrow_s u \qquad   t_1 : \ind{0}{`e} :: \ell \qquad  t_2 : [] 
	\justifies u: (< 0 \mid \ell) \zip \down ( > 0 \mid \ell)
	\using(Sub_{\in})
    \endprooftree
  \end{displaymath}\\

Notice that $ (\down \ell) \zip [] ~=~ \down \ell ~=~  (< 0 \mid \ell) \zip \down ( > 0 \mid \ell)$, because $ (< 0 \mid \ell) = []$ and $( > 0 \mid \ell) = \ell$.
\end{proof}
As a consequence, $`b_{\textrm{\textregistered}}$-reduction preserves linearity.

\subsection{Examples of reductions in $\lRdB$-calculus : {\sf SKK} and booleans}

The operational semantics of the $\lRdB$-calculus is equipped with the $`b_{\textrm{\textregistered}}$-reduction, resource reduction and their reflexive, transitive and contextual closure. Hereinbelow,  we present examples of reductions in $\lRdB$-calculus of the terms {\sf SKK}, negation of a boolean, disjunction of booleans. The examples demonstrate the fine-grained structure of the reductions obtained by the resource reductions.

\begin{example}[Reducing {\sf SKK} in $\lRdB$]\label{ex:redSKK}
  As an example, we propose to reduce the term \textsf{SKK}, from Subsection~\ref{sec:bestiary} :
  \begin{displaymath}
    \begin{array}{lll@{\quad}l}
      ((\lambda(\lambda(\lambda\dup{0}{`e}{((\ind{2}{`e}\,\ind{0}{0})\,(\ind{1}{`e}\,\ind{0}{1}))})))\,(\lambda(\lambda\era{0}{`e}{\ind{1}{`e}})))\,(\lambda(\lambda\era{0}{`e}{\ind{1}{`e}})) & \rightarrow&& \textrm{using~} (`b_{\textrm{\textregistered}})\\
      (\lambda(\lambda\dup{0}{`e}{(((\lambda(\lambda\era{0}{`e}{\ind{1}{`e}}))\,\ind{0}{0})\,(\ind{1}{`e}\,\ind{0}{1}))}))\,(\lambda(\lambda\era{0}{`e}{\ind{1}{`e}}))& \rightarrow&& \textrm{using~} (`b_{\textrm{\textregistered}})\\
      \lambda\dup{0}{`e}{(((\lambda(\lambda\era{0}{`e}{\ind{1}{`e}}))\,\ind{0}{0})\,((\lambda(\lambda\era{0}{`e}{\ind{1}{`e}}))\,\ind{0}{1}))}& \rightarrow&& \textrm{using~} (`b_{\textrm{\textregistered}}) \\
      \lambda\dup{0}{`e}{\era{0}{1}{\ind{0}{0}}}&\rightarrow_c& \lambda\ind{0}{`e}&  \textrm{using~}  (\odot {-}\triangledown_1)
    \end{array}
  \end{displaymath}
As expected, \textsf{SKK} reduces do \textsf{I}.
\end{example}

\begin{example}[Reducing the negation of a boolean in $\lRdB$]\label{ex:negbool}
  Boolean \textsf{tt} and boolean \textsf{ff} are given in Section~\ref{sec:bestiary}.
  Negation is term $\lambda(\lambda(\lambda((\ind{2}{`e}\,\ind{0}{`e})\,\ind{1}{`e})))$.
  \begin{displaymath}
  \begin{array}{lll@{\quad}l}
 (\lambda(\lambda(\lambda((\ind{2}{`e}\,\ind{0}{`e})\,\ind{1}{`e}))))\,(\lambda(\lambda\era{0}{`e}{\ind{1}{`e}})) 
    & \rightarrow&& \textrm{using~} (`b_{\textrm{\textregistered}})\\
\lambda(\lambda(((\lambda(\lambda\era{0}{`e}{\ind{1}{`e}}))\,\ind{0}{`e})\,\ind{1}{`e})) 
    & \rightarrow&& \textrm{using~} (`b_{\textrm{\textregistered}})\\
 \lambda(\lambda((\lambda\era{0}{`e}{\ind{1}{`e}})\,\ind{1}{`e})) 
    &\rightarrow&& \textrm{using~} (`b_{\textrm{\textregistered}})\\
\lambda(\lambda\era{1}{`e}{\ind{0}{`e}}) 
    & \rightarrow_c& \lambda\era{0}{`e}{(\lambda\ind{0}{`e})} & \textrm{using~} ({\lambda} {-}\odot)
  \end{array}
\end{displaymath}
In other words, the negation of \textsf{tt} reduces to \textsf{ff}. 
\begin{displaymath}
  \begin{array}{lll@{\quad}l}
 (\lambda(\lambda(\lambda((\ind{2}{`e}\,\ind{0}{`e})\,\ind{1}{`e}))))\,(\lambda\era{0}{`e}{(\lambda\ind{0}{`e})}) 
    & \rightarrow&& \textrm{using~} (`b_{\textrm{\textregistered}})\\
    \lambda(\lambda(((\lambda\era{0}{`e}{(\lambda\ind{0}{`e})})\,\ind{0}{`e})\,\ind{1}{`e})) 
    & \rightarrow&& \textrm{using~} (`b_{\textrm{\textregistered}})\\
    \lambda(\lambda(\era{0}{`e}{(\lambda\ind{0}{`e})}\,\ind{1}{`e})) 
    &\rightarrow_c&  &  \textrm{using~} ({\lambda} {-}\odot)\\
    \lambda(\lambda\era{0}{`e}{((\lambda\ind{0}{`e})\,\ind{1}{`e})}) 
    & \rightarrow & \lambda(\lambda\era{0}{`e}{\ind{1}{`e}}) & \textrm{using~} (`b_{\textrm{\textregistered}}) 
  \end{array}
\end{displaymath}
\end{example}
In other words, the negation of \textsf{ff} reduces to \textsf{tt}. The reader may notice that rule $({\lambda} {-}\odot)$ pulls out a $\odot$ beyond a $`l$ and therefore enables a $ (`b_{\textrm{\textregistered}})$ contraction.

\begin{example}[Reducing the disjunction of booleans in $\lRdB$]\label{ex:orbool}
  This example deals with disjunction in booleans. The disjunction in $\LRdB$ is
  \begin{displaymath}
    \lambda(\lambda(\lambda(\lambda\dup{1}{`e}{((\ind{3}{`e}\,\ind{1}{0})\,((\ind{2}{`e}\,\ind{1}{1})\,\ind{0}{`e}))})))
  \end{displaymath}
  Let us start with the disjunction of \textsf{tt} with itself. In other hands, let us reduce ``\textsf{tt} \textsf{and} \textsf{tt}''. 
  \begin{displaymath}
    \begin{array}{lll@{\quad}l}
      (\lambda(\lambda(\lambda(\lambda\dup{1}{`e}{((\ind{3}{`e}\,\ind{1}{0})\,((\ind{2}{`e}\,\ind{1}{1})\,\ind{0}{`e}))})))) (\lambda(\lambda\era{0}{`e}{\ind{1}{`e}})) (\lambda(\lambda\era{0}{`e}{\ind{1}{`e}})) 
       & \rightarrow&& \textrm{using~} (`b_{\textrm{\textregistered}})\\
      (\lambda(\lambda(\lambda\dup{1}{`e}{(((\lambda(\lambda\era{0}{`e}{\ind{1}{`e}}))\,\ind{1}{0})\,((\ind{2}{`e}\,\ind{1}{1})\,\ind{0}{`e}))})))\,(\lambda(\lambda\era{0}{`e}{\ind{1}{`e}})) 
       & \rightarrow&& \textrm{using~} (`b_{\textrm{\textregistered}})\\
      \lambda(\lambda\dup{1}{`e}{(((\lambda(\lambda\era{0}{`e}{\ind{1}{`e}}))\,\ind{1}{0})\,(((\lambda(\lambda\era{0}{`e}{\ind{1}{`e}}))\,\ind{1}{1})\,\ind{0}{`e}))}) 
      & \rightarrow&& \textrm{using~} (`b_{\textrm{\textregistered}})\\ 
     \lambda(\lambda\dup{1}{`e}{((\lambda\era{0}{`e}{\ind{2}{0}})\,(((\lambda(\lambda\era{0}{`e}{\ind{1}{`e}}))\,\ind{1}{1})\,\ind{0}{`e}))}) 
      & \rightarrow&& \textrm{using~} (`b_{\textrm{\textregistered}})\\
      \lambda(\lambda\dup{1}{`e}{\era{1}{1}{\era{0}{`e}{\ind{1}{0}}}}) 
 & \rightarrow_c && \textrm{using~} (\odot {-}\triangledown_1)\\
\lambda(\lambda\era{0}{`e}{\ind{1}{`e}})
    \end{array}
  \end{displaymath}
Let us see now the disjunction of \textsf{ff} with itself. 
\begin{displaymath}
  \begin{array}{ll@{~}l}
    ((\lambda(\lambda(\lambda(\lambda\dup{1}{`e}{((\ind{3}{`e}\,\ind{1}{0})\,((\ind{2}{`e}\,\ind{1}{1})\,\ind{0}{`e}))}))))\,(\lambda\era{0}{`e}{(\lambda\ind{0}{`e})}))\,(\lambda\era{0}{`e}{(\lambda\ind{0}{`e})})
    & \rightarrow& \textrm{using~} (`b_{\textrm{\textregistered}})\\
    (\lambda(\lambda(\lambda\dup{1}{`e}{(((\lambda\era{0}{`e}{(\lambda\ind{0}{`e})})\,\ind{1}{0})\,((\ind{2}{`e}\,\ind{1}{1})\,\ind{0}{`e}))})))\,(\lambda\era{0}{`e}{(\lambda\ind{0}{`e})})
    & \rightarrow& \textrm{using~} (`b_{\textrm{\textregistered}})\\
    \lambda(\lambda\dup{1}{`e}{(((\lambda\era{0}{`e}{(\lambda\ind{0}{`e})})\,\ind{1}{0})\,(((\lambda\era{0}{`e}{(\lambda\ind{0}{`e})})\,\ind{1}{1})\,\ind{0}{`e}))})
    & \rightarrow& \textrm{using~} (`b_{\textrm{\textregistered}})\\
    \lambda(\lambda\dup{1}{`e}{(\era{1}{0}{(\lambda\ind{0}{`e})}\,(((\lambda\era{0}{`e}{(\lambda\ind{0}{`e})})\,\ind{1}{1})\,\ind{0}{`e}))})
    & \rightarrow& \textrm{using~} (`b_{\textrm{\textregistered}})\\
    \lambda(\lambda\dup{1}{`e}{(\era{1}{0}{(\lambda\ind{0}{`e})}\,(\era{1}{1}{(\lambda\ind{0}{`e})}\,\ind{0}{`e}))})
    & \rightarrow_c& \textrm{using~}  (\text{AppL} {-} \odot)\\
    \lambda(\lambda\dup{1}{`e}{\era{1}{0}{((\lambda\ind{0}{`e})\,(\era{1}{1}{(\lambda\ind{0}{`e})}\,\ind{0}{`e}))}})
     & \rightarrow_c& \textrm{using~}  (\odot {-}\triangledown_0)\\
    \lambda(\lambda((\lambda\ind{0}{`e})\,(\era{1}{`e}{(\lambda\ind{0}{`e})}\,\ind{0}{`e})))
    & \rightarrow& \textrm{using~} (`b_{\textrm{\textregistered}})\\
    \lambda(\lambda(\era{1}{`e}{(\lambda\ind{0}{`e})}\,\ind{0}{`e}))
     &\rightarrow_c& \textrm{using~}  (\text{AppL} {-} \odot)\\
    \lambda(\lambda\era{1}{`e}{((\lambda\ind{0}{`e})\,\ind{0}{`e})})
    & \rightarrow& \textrm{using~} (`b_{\textrm{\textregistered}})\\
    \lambda(\lambda\era{1}{`e}{\ind{0}{`e}})
     &\rightarrow_c& \textrm{using~}  ({\lambda} {-}\odot)\\
    \lambda\era{0}{`e}{(\lambda\ind{0}{`e})}
      \end{array}
    \end{displaymath}
    As expected, the disjunction of \textsf{ff} and \textsf{ff}
    reduces to \textsf{ff}.  The examples were made possible by the
    \texttt{Haskell} implementation (see Section~\ref{sec:Haskell}).
    To summarise: the above examples use 19 times rule~$(`b_{\textrm{\textregistered}})$ and 8 times resource control rewriting rules from Figure~\ref{fig:rewriting rules for LRdB}.
    Those ``statistics'' give an intuition on the
    cost of resource control. But a study grounded on analytic combinatorics~\cite{flajolet08:_analy_combin}, like this
    of~\cite{DBLP:conf/ppdp/BendkowskiL18,DBLP:journals/lmcs/BendkowskiL19},
    especially like Table~1 of~\cite{DBLP:conf/ppdp/BendkowskiL18},
    might give even more precise information on the ratio of
    $\triangledown$-contractions and $\odot$-contractions over
    $`b_{\textrm{\textregistered}}$-contractions.  This would require
    an approach of the substitution of ®-indices by explicit
    substitutions, which is not done yet. Perhaps, sampling of
    $\lRdB$-terms associated with empiric counting of contractions
    when reducing to the normal form would be valuable.
  \end{example}

\subsection{Examples of reductions in $\lRdB$-calculus :  Church numerals}
\label{section-example}
\newcommand{\sucR}{\mathbf{s}_{\textrm{\textregistered}}}
\newcommand{\Chtwo}{\textbf{\textsf{two}}}
\newcommand{\Chthree}{\textbf{\textsf{three}}}
\newcommand{\Chfive}{\textbf{\textsf{five}}}
\newcommand{\Sum}{\textbf{\textsf{sum}}}

Here we show examples based on Church numerals (see Section~\ref{sec:bestiary}).  The archetype of a non linear $`l$-term is the Church numeral for $\mathbf{2}$ which corresponds to the $\lRdB$-term:
\begin{displaymath}
  \Chtwo \equiv \lambda(\lambda\dup{1}{`e}{(\ind{1}{0}\,(\ind{1}{1}\,\ind{0}{`e}))})
\end{displaymath}
$\Chtwo$ has one occurrence of $\triangledown$, that is one duplication.  Another non linear $`l$-term is \textbf{successor} on Church numeral, which we write $\sucR$ and which is the $\lRdB$-term:
  \begin{eqnarray*}
    \sucR &\equiv& \lambda(\lambda(\lambda\dup{1}{`e}{((\ind{2}{`e}\,\ind{1}{0})\,(\ind{1}{1}\,\ind{0}{`e}))}))
  \end{eqnarray*}
The successor of $\Chtwo$ is $\sucR~\Chtwo$. But  $\sucR~\Chtwo$ is not in $\lRdB$-normal form. It reduces as follows:
\begin{displaymath}
\begin{array}{rcl}
    \sucR~\Chtwo &\equiv& (\lambda(\lambda(\lambda\dup{1}{`e}{((\ind{2}{`e}\,\ind{1}{0})\,(\ind{1}{1}\,\ind{0}{`e}))})))\,(\lambda(\lambda\dup{1}{`e}{(\ind{1}{0}\,(\ind{1}{1}\,\ind{0}{`e}))})) \\
                 &\stackrel{3}{\rightarrow}& \lambda(\lambda\dup{1}{`e}{\dup{1}{0}{(\ind{1}{00}\,(\ind{1}{01}\,(\ind{1}{1}\,\ind{0}{`e})))}})%
                 \qquad \textrm{using three times~} (`b_{\textrm{\textregistered}})\\
  &\equiv& \Chthree
\end{array}
\end{displaymath}
We notice that $\Chthree$ has two $\triangledown$'s.

In $`l$-calculus, $\Sum$  is  yet another non linear $`l$-term.  In $\LRdB$, it is
\begin{displaymath}
\lambda(\lambda(\lambda(\lambda\dup{1}{`e}{((\ind{3}{`e}\,\ind{1}{0})\,((\ind{2}{`e}\,\ind{1}{1})\,\ind{0}{`e}))}))).
\end{displaymath}
Let use the $\Sum$ to define a Church numeral $\Chfive$ as the $\lRdB$-normal form of $\Sum~(\sucR~\Chtwo)~\Chtwo$.
\begin{displaymath}
    \Sum~(\sucR~\Chtwo)~\Chtwo  \stackrel{6}{\rightarrow}
    \lambda(\lambda\dup{1}{`e}{\dup{1}{0}{\dup{1}{00}{(\ind{1}{000}\,(\ind{1}{001}\,(\ind{1}{01}\,\dup{1}{1}{(\ind{1}{10}\,(\ind{1}{11}\,\ind{0}{`e}))})))}}})
  \end{displaymath}
  using 6 times $`b_{\textrm{\textregistered}}$ and the term can no more be reduced. It has $4$ duplications.  If, instead of  $\Sum~(\sucR~\Chtwo)~\Chtwo$  we take $\Sum~\Chtwo~(\sucR~\Chtwo)$ we get:
  \begin{displaymath}
    \lambda(\lambda\dup{1}{`e}{\dup{1}{0}{(\ind{1}{00}\,(\ind{1}{01}\,\dup{1}{1}{\dup{1}{10}{(\ind{1}{100}\,(\ind{1}{101}\,(\ind{1}{11}\,\ind{0}{`e})))}}))}})
  \end{displaymath}
  after $6$ $`b_{\textrm{\textregistered}}$-reductions.
 There is another $\lRdB$-term for the Church numeral for $\mathbf{5}$ :
  \begin{displaymath}
    \lambda(\lambda\dup{1}{`e}{(\ind{1}{0}\,\dup{1}{1}{(\ind{1}{10}\,\dup{1}{11}{(\ind{1}{110}\,\dup{1}{111}{(\ind{1}{1110}\,(\ind{1}{1111}\,\ind{0}{`e}))})})})})
  \end{displaymath}
  The two above $\lRdB$-terms are different by the order, in which the duplications are performed.
  \begin{itemize}
  \item In the first term,$\ind{1}{`e}$ is first duplicated, then $\ind{1}{0}$ is duplicated, then $\ind{1}{00}$ is duplicated, then $\ind{1}{1}$ is duplicated.
  \item In the second term, $\ind{1}{`e}$ is first duplicated, then $\ind{1}{0}$ is duplicated, then $\ind{1}{1}$ is duplicated, then $\ind{1}{10}$ is duplicated.
  \item In the third term, $\ind{1}{`e}$ is first duplicated, then $\ind{1}{1}$ is duplicated,  $\ind{1}{11}$ is duplicated, then $\ind{1}{111}$ is duplicated.
  \end{itemize}
  We can define a notion of \textsf{representative} term which can be computed by a program and which supposes a somewhat ''canonical'' way of performing the duplications. But the there is several ways to define the notion of ``\textsf{canonical}''. Which notion of canonicity to choose is a matter of taste.  The term obtained by the function \textsf{read} is \textsf{canonical}, for the notion canonicity we choose.  Therefore a natural way to define the \textsf{representative} of the $\lRdB$-term  $t$ is $\textsf{read}(\textsf{readback}(t))$.  \textsf{read} and \textsf{readback} are defined in Section~\ref{sec:read}.
  However, there is no rewrite system to compute it, like the system of Figure~\ref{fig:rewriting rules for LRdB}.  Among the three $\lRdB$-normal forms above, representing the Church numeral for $\mathbf{5}$ in~$\LRdB$, the third term is the \textsf{representative}.

\subsection{Correspondence with $\Lambda$}
\label{sec:read}
In order to establish a correspondence between the introduced language
$\LR$ and the well-known system $\Lambda$, we define two translations:
$\mathsf{read} : `L "->" \LR$ and $\mathsf{readback} : \LR "->" `L$.
\begin{definition}[\textsf{read}] \label{def:read}
  $\mathsf{read} : `L "->" \LR$
  \begin{itemize}
  \item $\mathsf{read}~ \udl{n}  = \ind{n}{`e}$
  \item $\mathsf{read}~ (`l\,t) = \mathsf{let~} u = \mathsf{read} ~t \\
    \hspace*{50pt} \mathsf{~in~} \mathsf{if~} \ind{0}{`e} `: u \mathsf{~then~} `l\,u\\
    \hspace*{65pt} \mathsf{else~} `l\,\era{0}{`e}{u}$
  \item $\mathsf{read} (t_1\,t_2) = \multiDup{t_1^\star\cap t_2^\star} (\mathsf{rename}~ 0~ (t_1^\star\cap t_2^\star) (\mathsf{read}~ t_1)\, \mathsf{rename}~ 1~ (t_1^\star\cap t_2^\star) (\mathsf{read}~ t_2)$
  \end{itemize}
  where
  \begin{itemize}
  \item[$-$] $\textsf{rename}~0~\ell$ replaces every ®-index of the form
    $\ind{n}{`a}$ in the list $\ell$ of ®-indices by the corresponding
    ®-index of the form $\ind{n}{`a0}$ and similarly
    $\textsf{rename}~1~\ell$ replaces all ®-index of the form $\ind{n}{`a}$ in
    the list by the corresponding ®-index of the form $\ind{n}{`a1}$.
  \item[$-$] $t_1^\star\cap t_2^\star$ is a short notation for the list of ®-indices that occur both in $\mathsf{read}(t_1)$ and in $\mathsf{read(t_2)}$.
  \item[$-$] $\multiDup{\ell}~t$ is the iteration of the duplicator $\triangledown$.  In other words, if $\ell$ is the list $[\ind{k_1}{`g_1},\ind{k_2}{`g_2},...,\ind{k_n}{`g_n}]$, then $\multiDup{\ell}~t = \dup{k_1}{`g_1}{\dup{k_2}{`g_2}{...\dup{k_n}{`g_n}{t}}}$.
\end{itemize}
\end{definition}

The translation \textsf{read} is the formalisation of the translations presented in Example~\ref{exa:R} and  corresponds to the \textsf{Haskell} function \textsf{readLR} (l. 151 in \href{https://github.com/PierreLescanne/LambdaCalculusWithDuplicationsAndErasures/blob/master/Lambda_R_dB.hs}{Lambda$\_$R$\_$dB.hs}).
\begin{definition}[\textsf{readback}] \label{def:readback}
  $\mathsf{readback}  : \LR "->" `L$
  \begin{itemize}
  \item $\mathsf{readback}~\ind{n}{`a} = \udl{n}$
  \item $\mathsf{readback}~(`l t) = `l (\mathsf{readback}~t)$
  \item $\mathsf{readback}~(t_1\,t_2) = (\mathsf{readback}~t_1)\,(\mathsf{readback}~t_2)$
  \item $\mathsf{readback}~(\era{n}{`a}{t}) = \mathsf{readback}~t$
  \item $\mathsf{readback}~(\dup{n}{`a}{t}) = \mathsf{readback}~t$.
  \end{itemize}
\end{definition}
\begin{proposition}[Correctness of \textsf{read}]\label{th:correctness}
  \begin{math}
    \bm{\lambda}\, t . \mathsf{readback}~ (\mathsf{read}~ t) : `L "->" `L 
  \end{math}
  is the identity on~$`L$.  In other words,
  \begin{displaymath}
    \mathsf{readback}~ (\mathsf{read}~ t) = t.
  \end{displaymath}
\end{proposition}

The function \begin{math}
    \bm{\lambda}\, t . \mathsf{read}~ (\mathsf{readback}~ t) : \LR "->" \LR 
  \end{math}
is an interesting function which associates with a term $t$ another term
with a somewhat standard disposition of $\odot$ and $\triangledown$,
which we call \textsf{standardisation} of the term.

\thefigurelOOO

Evidently, the same non-linear $`l$-term may correspond to several $\lRdB$-terms. For instance, this
is the case for the term $`l ((\udl{0}\,\udl{0})\,\udl{0})$ (a $\LRdB$ instance of $`l x. x x x$) illustrated by the following example and pictured in Figure~\ref{fig:Lxxx}.

  \begin{example}\label{exa:Lxxx}
Consider the term $\lambda\dup{0}{\varepsilon}{\dup{0}{1}{(\ind{0}{0}
\ind{0}{10}) \ind{0}{11}}}$.  This is actually term $`l x . xxx$ in explicit name and no duplication.
\begin{displaymath}
  \mathsf{readback}(\lambda\dup{0}{\varepsilon}{\dup{0}{1}{(\ind{0}{0}
\ind{0}{10}) \ind{0}{11}}} = `l (\udl{0}\,\udl{0})\, \udl{0}
\end{displaymath}
but 
\begin{displaymath}
  \mathsf{read}(`l (\udl{0}\,\udl{0})\, \udl{0}) = \lambda\dup{0}{\varepsilon}{\dup{0}{0}{((\ind{0}{00} \ind{0}{01}) \ind{0}{1})}}
\end{displaymath}
Hence
\begin{displaymath}
\mathsf{read} \mathop{`o} \mathsf{readback}(\lambda\dup{0}{\varepsilon}{\dup{0}{1}{\ind{0}{0} \ind{0}{10}}}
\ind{0}{11})) = \lambda\dup{0}{\varepsilon}{\dup{0}{0}{(\ind{0}{00}
\ind{0}{01}) \ind{0}{1}}}
\end{displaymath}
The reader may notice that, in both terms, the first duplication is
$\dup{0}{`e}{\_}$.  But the reader may also notice that the second duplication is
$\dup{0}{1}{\_}$ in the first term and $\dup{0}{0}{\_}$ in the second term. So
they are not the same.  Choosing $\dup{0}{0}{\_}$ over $\dup{0}{1}{\_}$ is
somewhat canonical.  This corresponds to choosing the leftmost diagram in Figure~\ref{fig:Lxxx}.
The
fourth diagram corresponds to
\begin{displaymath}
\lambda\dup{0}{\varepsilon}{\dup{0}{0}{\ind{0}{01} \ind{0}{00}}
\ind{0}{1}}
\end{displaymath}
and the fifth diagram corresponds to
\begin{displaymath}
\lambda\dup{0}{\varepsilon}{\dup{0}{1}{\ind{0}{11} \ind{0}{10}}
\ind{0}{0}}.
\end{displaymath}
We let the reader write the $\LRdB$ term corresponding to the third diagram of Figure~\ref{fig:Lxxx}. There are $12$ ways to write the term  $`l (\udl{0}\,\udl{0})\, \udl{0}$ in $\LRdB$ and to draw corresponding diagrams.  The reader may devise the omitted cases. 
\end{example}

\subsection{Implementation of $\LRdB$ in \textsf{Haskell}}\label{sec:Haskell}
We implemented the whole $\lRdB$ in \textsf{Haskell}, where the data type for $\LRdB$ is as follows:
\begin{minted}{haskell}
data RTerm = App RTerm RTerm
           | Abs RTerm
           | Ind Int [Bool]
           | Era Int [Bool] RTerm
           | Dup Int [Bool] RTerm
\end{minted}

We give here a flavor of the implementation. 
We have defined functions $\mathsf{read}$ and $\mathsf{readback}$. As presented in the previous section $\mathsf{readback}$ is relatively easy to define, by just forgetting duplications and erasures. Function $\mathsf{read}$, denoted by $\mathsf{readLR}$, is defined in Haskell as follows:

\begin{minted}{haskell}
-- Given a list of indices and a term,
-- dupTheIndices applies all the duplications of that list to that term
dupTheIndices :: [(Int,[Bool])] -> RTerm -> RTerm
dupTheIndices [] t = t
dupTheIndices ((i,alpha):l) t = Dup i alpha  (dupTheIndices l t)

-- `consR` is a function used in `readLR`
-- given a boolean and an index, put the boolean (0 or 1)
-- in front of all the alpha parts associated with the index
consR :: Bool -> Int -> RTerm -> RTerm
consR b i (App t1 t2) = App (consR b i t1) (consR b i t2)
consR b i (Abs t) = Abs (consR b (i+1) t)
consR b i (Ind j beta) = if i==j
                        then Ind j (b:beta)
                        else Ind j beta
consR b i (Era j beta t) = if i==j
                          then Era j (b:beta) (consR b i t)
                          else Era j beta (consR b i t)
consR b i (Dup j beta t) = if i==j
                          then Dup j (b:beta) (consR b i t)
                          else Dup j beta (consR b i t)

\end{minted}
\textsf{indOf} is a function that extracts the indices of a term;
\textsf{?} is an infix operator which returns a boolean, \textsf{i ? t}
returns True if and only if \textsf{i} occurs in \textsf{t}.
\begin{minted}{haskell}
readLR :: Term -> RTerm
readLR (Ap t1 t2) =
  let rt1 = readLR t1
      rt2 = readLR t2
      indToIndR i = (i,[])
      commonInd = sort (indOf t1 `intersect`  indOf t2)
      pt1 = foldl (.) id (map (consR False) commonInd) rt1
      pt2 = foldl (.) id (map (consR True) commonInd) rt2
  in dupTheIndices (map indToIndR commonInd) (App pt1 pt2) 
readLR (Ab t) = if 0 ? t then Abs (readLR t) else Abs (Era 0 [] (readLR t))
readLR (In i) = Ind i []
\end{minted} 

We also present the Haskell code for test of linearity and closedness:

\begin{minted}{haskell}
-- (iL t) returns the list of free (R)-de Bruijn indices of t
-- if all the binders of the term binds one and only one (R)-index.
remove :: Eq a => a -> [a] -> Maybe [a]
remove _ [] = Nothing
remove x (y:l) = if x == y then Just l
                 else case (remove x l) of
                   Nothing -> Nothing
                   Just l' -> Just (y:l')
\end{minted}

\begin{minted}{haskell}
iL :: RTerm -> Maybe [(Int,[Bool])]
iL (Ind n alpha) = Just [(n,alpha)]
iL (Abs t) =
  case iL t of
    Nothing -> Nothing
    Just u -> case remove (0,[]) u of 
      Nothing -> Nothing
      Just u' -> case find (((==) 0).fst) u' of
        Just _ -> Nothing
        Nothing -> Just \$ map ((i,a)->(i-1,a)) u

  \end{minted}

\begin{minted}{haskell}
iL (App t1 t2) =
  case iL t1 of
    Nothing -> Nothing
    Just u1 -> case iL t2 of
                         Nothing -> Nothing
                         Just u2 -> if null (u1 `intersect` u2)
                                    then Just(u1 ++ u2)
                                    else Nothing
                                  
iL (Era n alpha t) = case iL t of
                       Nothing -> Nothing
                       Just u -> Just ((n,alpha):u)
iL (Dup n alpha t) =
  case iL t of
    Nothing -> Nothing
    Just u -> if (n,alpha++[False]) `elem` u &&
                 (n,alpha++[True]) `elem` u
              then Just ((n,alpha):(delete (n,alpha++[False]) (delete (n,alpha++[True]) u)))
              else Nothing
      
    \end{minted}
      
      \begin{minted}{haskell}           
-- is linear in the sense that all the binders bound one and only one index. 
isLinearAndClosed t = case iL t of
  Nothing -> False
  Just u -> u == []
\end{minted}

\section{Discussion and Related work}
\label{sec:related}

Compared to languages with explicit names, like
$`l\mathsf{lxr}$~\cite{DBLP:journals/iandc/KesnerL07} or the language
of~\cite{DBLP:journals/corr/GhilezanILL14}, $\lRdB$ is a simpler  
calculus, because, we can tell exactly how the ®-indices are
duplicated, since we have a tight control on the way those indices are
built.  As consequences, there are fewer basic rules and a simple
implementation is possible.  For instance, if we consider a rough
quantitative aspect, the calculus
of~\cite{DBLP:journals/iandc/KesnerL07} has $19$ rules and $6$
congruences, the system of~\cite{DBLP:journals/corr/GhilezanILL14} has
$18$ rules ($9$~basic rules and $8$ rules for substitution) and $4$
congruences, whereas our system $\lRdB$ has $12$ rules and no congruences.

Linear logic~\cite{DBLP:journals/tcs/Girard87} is, among others, an
approach to linearity of $`l$-calculus and certainly the most popular and there is a vast literature on the subject.  By
1998, people at CMU~\cite{pfenning98:_html} collected 463 entries. Let us cite a few papers that address the
implementation of the linear $\lambda$-calculus by a calculus of explicit substitution related to linear
logic~\cite{DBLP:conf/fossacs/GhaniPR99,DBLP:journals/igpl/GhaniPR00,cervesato99}.  More specifically, in those calculi,
the type system is this of linear logic, with connectors like $"-o"$ (linear implication), $\otimes$~(tensor) and $!$
(exponential modality).

The two approaches share a common focus on linear $`l$-terms, addressing similar objects. However, they diverge in their treatment of types: the authors of ~\cite{DBLP:conf/fossacs/GhaniPR99,DBLP:journals/igpl/GhaniPR00,cervesato99} use a different set of types, namely the specific set of types of linear logic, whereas, in this paper, a given $`l$-term can receive a pair of types, namely a first type usual in $`l$-calculus (simple type, system F, calculus of construction, etc.) to characterise its computing behaviour and a second type ($\Ell$-type) telling its linearity. \Jel{A different approach to tracking resource usage in type systems is presented by McBride~\cite{DBLP:conf/birthday/McBride16}, where resource-annotated contexts distinguish between data that is observed and data that is consumed. Unlike our systems, which assign $\Ell$-types to control linearity, McBride uses annotations to track usage constraints, ensuring that certain terms can be referenced without explicit duplication or erasure. This perspective aligns with our treatment of $\Ell$-types, which regulate resource usage in a different type-theoretic setting.}

The $\Ell$-types of our systems address a notion of \emph{correctness} which
is somewhat orthogonal to this of 
standard types (such as simple types or
higher order types).  A term is well $\Ell$-typed if it is linear and
we prove, thanks to $\Ell$-type preservation, that linearity is
preserved by reduction.  The two notions of types are orthogonal in
the sense that standard types say something about the result (the term
is an integer or a boolean, for instance) whereas $\Ell$-types say
something about the internal features of the terms (the term is linear). 
Since we do not characterise the ``result'' of a computation, but only the structure of
the term, there is no notion of ``progress'' associated with
$\Ell$-types, there is only a notion of ``preservation'' (terms stay linear along their reduction, i.e. $\Ell$-type is preserved). However, it is possible to introduce other standard type systems, such as simple types, intersection types or system F, to further characterise computational properties of $\Ell$-typed terms. As the continuation of the presented research, we intent to explore such a hierarchy of type systems for the language $\LRdB$.

We have developed the concept of $\Ell$-types with the aim to characterise 
linear $\lambda$-terms. 
Linear $\lambda$-calculus with implicit names, namely de Bruijn indices, is also in the focus of research presented in \cite{Schack-NielsenS10,Schack-Nielsen11}. More precisely, in \cite{Schack-NielsenS10} the authors have introduced the calculus of explicit substitutions with the affine and linear types.  They have extended simple type system with affine and linear functional types, in order to control the consumption of resources.  However, these types are classic types in the sense explained above.

The language $\LRdB$ has connection with the \emph{differential $`l$-calculus} of Erhard and Regnier~\cite{DBLP:journals/tcs/EhrhardR08} where the \emph{duplicator} $\triangledown$ is a non-commutative differential operator (similar to their~$D$, which is commutative) and the \emph{black-hole} $\odot$ corresponds to an empty iteration of $\triangledown$ (like~$D^0$).  Therefore $\lRdB$ can be considered as a \emph{non-commutative differential $`l$-calculus}, where iterations are no more done on natural numbers, but on lists of \texttt{Bool}.  These observations merit  to be deepened.

 Paul Tarau and Valeria de Pavia address a similar problem (\cite{DBLP:journals/corr/abs-2009-10241} Section 4.3) in their attempt to generate closed linear  $\lambda$-terms.
Anyway in functional programming, programs of interest are those with no free (undeclared) variable.

\Sil{The implementation in Agda~\cite{githubrepo}, currently under development, covers some topics, e.g. explicit substitution, which are not present in~\cite{DBLP:journals/scp/KokkeSW20}.}

\section{Conclusion}
\label{sec:conc}
This paper introduces a novel approach to dealing with resource-related properties of formal languages, specifically focusing on the concept of  linearity. This concept, although informally clear and intuitive, can be challenging to formally define and successfully implement. Our approach relies on using implicit names (such as de Bruijn indices or novel ®-indices) and $\Ell$-types (which represent the list of free indices in a term). This approach enables simple formal definitions of 
linearity through $\Ell$-typeability. 

To illustrate our approach, we introduced three new languages with implicit names. The first language is the most straightforward, as it comprises traditional $`l$-terms with unique occurrences of each variable (\textsf{BCI$`l$}-terms). The second language addresses an abstract implementation of $`b$\=/reduction through explicit substitution. The third language is concerned with resource control, featuring explicit duplication and explicit erasure.

For each of those languages, we introduce a specific list type system, which is used to give a simple and manageable definition of 
linearity. Those types represent lists of free implicit names: de Bruin indices for the first two calculi, and new ®-indices for the calculus with resource control.

To summarise, we have addressed the concepts of 
linearity, implicit names and explicit resource control (duplication and erasure) by introducing and implementing \Sil{in Haskell \cite{githubrepo2}} list types, ®-indices and three new languages: $\Lin$, $\Luin$ and~$\LRdBin$.

\Sil{As for ongoing and further work, we plan to develop a direct method for charaterising  \emph{affiness and relevance} in the implicit names framework. Moreover, the \textsf{Agda} implementation, which is currently under  development~\cite{githubrepo}, will be furthered}, as well as a study of the statistics (the ratio) of the use of $(`b_{\textrm{\textregistered}})$-contractions over control operations. This would require a definition of ®-substitution in terms of explicit ®-substitution.

\bibliographystyle{plainurl}
\bibliography{ref}

\begin{thebibliography}{10}

\bibitem{DBLP:books/daglib/0095289}
Andrea Asperti and Stefano Guerrini.
\newblock {\em The optimal implementation of functional programming languages},
  volume~45 of {\em Cambridge tracts in theoretical computer science}.
\newblock Cambridge University Press, 1998.

\bibitem{DBLP:conf/ppdp/BendkowskiL18}
Maciej Bendkowski and Pierre Lescanne.
\newblock Combinatorics of explicit substitutions.
\newblock In {\em PPDP'18}, pages 7:1--7:12, 2018.
\newblock \href {https://doi.org/10.1145/3236950.3236951}
  {\path{doi:10.1145/3236950.3236951}}.

\bibitem{DBLP:journals/lmcs/BendkowskiL19}
Maciej Bendkowski and Pierre Lescanne.
\newblock On the enumeration of closures and environments with an application
  to random generation.
\newblock {\em Log. Methods Comput. Sci.}, 15(4), 2019.
\newblock \href {https://doi.org/10.23638/LMCS-15(4:3)2019}
  {\path{doi:10.23638/LMCS-15(4:3)2019}}.

\bibitem{DBLP:journals/entcs/BerghoferU07}
Stefan Berghofer and Christian Urban.
\newblock A head-to-head comparison of de {Bruijn} indices and names.
\newblock In Alberto Momigliano and Brigitte Pientka, editors, {\em Proceedings
  of the First International Workshop on Logical Frameworks and Meta-Languages:
  Theory and Practice, LFMTP@FLoC 2006, Seattle, WA, USA, August 16, 2006},
  volume 174 of {\em Electronic Notes in Theoretical Computer Science}, pages
  53--67. Elsevier, 2006.
\newblock URL: \url{https://doi.org/10.1016/j.entcs.2007.01.018}, \href
  {https://doi.org/10.1016/J.ENTCS.2007.01.018}
  {\path{doi:10.1016/J.ENTCS.2007.01.018}}.

\bibitem{DBLP:conf/concur/Boudol93}
G{\'{e}}rard Boudol.
\newblock The lambda-calculus with multiplicities (abstract).
\newblock In Eike Best, editor, {\em {CONCUR} '93, 4th International Conference
  on Concurrency Theory, Hildesheim, Germany, August 23-26, 1993, Proceedings},
  volume 715 of {\em Lecture Notes in Computer Science}, pages 1--6. Springer,
  1993.
\newblock \href {https://doi.org/10.1007/3-540-57208-2\_1}
  {\path{doi:10.1007/3-540-57208-2\_1}}.

\bibitem{Bourbaki39}
Nicolas Bourbaki.
\newblock {\em {\'E}l{\'e}ments de {M}ath{\'e}matiques. {L}ivre {I}.
  {T}h{\'e}orie des ensembles. Fascicule de r{\'e}sultats}.
\newblock Hermann \& Cie, Paris, 1939.
\newblock translation in \cite{bourbaki68:_theor_sets}.

\bibitem{bourbaki68:_theor_sets}
Nicolas Bourbaki.
\newblock {\em Theory of Sets}.
\newblock Springer, 1968.

\bibitem{DBLP:conf/concur/CairesP10}
Lu{\'{\i}}s Caires and Frank Pfenning.
\newblock Session types as intuitionistic linear propositions.
\newblock In Paul Gastin and Fran{\c{c}}ois Laroussinie, editors, {\em {CONCUR}
  2010 - Concurrency Theory, 21th International Conference, {CONCUR} 2010,
  Paris, France, August 31-September 3, 2010. Proceedings}, volume 6269 of {\em
  Lecture Notes in Computer Science}, pages 222--236. Springer, 2010.
\newblock \href {https://doi.org/10.1007/978-3-642-15375-4\_16}
  {\path{doi:10.1007/978-3-642-15375-4\_16}}.

\bibitem{DBLP:conf/sas/CalvertM12}
Peter Calvert and Alan Mycroft.
\newblock Control flow analysis for the join calculus.
\newblock In Antoine Min{\'{e}} and David Schmidt, editors, {\em Static
  Analysis - 19th International Symposium, {SAS} 2012, Deauville, France,
  September 11-13, 2012. Proceedings}, volume 7460 of {\em Lecture Notes in
  Computer Science}, pages 181--197. Springer, 2012.
\newblock \href {https://doi.org/10.1007/978-3-642-33125-1\_14}
  {\path{doi:10.1007/978-3-642-33125-1\_14}}.

\bibitem{cervesato99}
Iliano Cervesato, Valeria de~Paiva, and Eike Ritter.
\newblock Explicit substitutions for linear logical frameworks: Preliminary
  results.
\newblock In {\em Workshop on Logical Frameworks and Metalanguages LFM'99, Held
  as part of the Colloquium on Principles, Logics, and Implementations of
  High-Level Programming Languages, Paris, France, September 28, 1999,
  Proceedings}, 1999.

\bibitem{curien93:categ_combin}
Pierre-Louis Curien.
\newblock {\em Categorical Combinators, Sequential Algorithms and Functional}.
\newblock Birkhaüser, 2nd edition, 1993.

\bibitem{NGDeBruijn108}
Nicolaas~Govert de~Bruijn.
\newblock Lambda calculus with nameless dummies, a tool for automatic formula
  manipulation, with application to the {Church-Rosser} theorem.
\newblock {\em Proc. Koninkl. Nederl. Akademie van Wetenschappen},
  75(5):381--392, 1972.

\bibitem{DS-H93}
Kosta Do\v{s}en and Peter Schroeder-Heister eds.
\newblock {\em Substructural Logics}.
\newblock Perspectives in logic. Oxford University Press, 1993.

\bibitem{DBLP:journals/tcs/EhrhardR08}
Thomas Ehrhard and Laurent Regnier.
\newblock Uniformity and the {Taylor} expansion of ordinary lambda-terms.
\newblock {\em Theor. Comput. Sci.}, 403(2-3):347--372, 2008.
\newblock \href {https://doi.org/10.1016/j.tcs.2008.06.001}
  {\path{doi:10.1016/j.tcs.2008.06.001}}.

\bibitem{flajolet08:_analy_combin}
Philippe Flajolet and Robert Sedgewick.
\newblock {\em Analytic Combinatorics}.
\newblock Cambridge University Press, 2008.

\bibitem{gent35}
Gerhard Gentzen.
\newblock Unterschungen \"uber das logische {S}chliessen, {M}ath {Z}. 39
  (1935), 176--210.
\newblock In M.E. Szabo, editor, {\em Collected papers of Gerhard Gentzen},
  pages 68--131. North-Holland, 1969.

\bibitem{DBLP:conf/fossacs/GhaniPR99}
Neil Ghani, Valeria de~Paiva, and Eike Ritter.
\newblock Categorical models of explicit substitutions.
\newblock In Wolfgang Thomas, editor, {\em Foundations of Software Science and
  Computation Structure, Second International Conference, FoSSaCS'99, Held as
  Part of the European Joint Conferences on the Theory and Practice of
  Software, ETAPS'99, Amsterdam, The Netherlands, March 22-28, 1999,
  Proceedings}, volume 1578 of {\em Lecture Notes in Computer Science}, pages
  197--211. Springer, 1999.
\newblock \href {https://doi.org/10.1007/3-540-49019-1\_14}
  {\path{doi:10.1007/3-540-49019-1\_14}}.

\bibitem{DBLP:journals/igpl/GhaniPR00}
Neil Ghani, Valeria de~Paiva, and Eike Ritter.
\newblock Linear explicit substitutions.
\newblock {\em Log. J. {IGPL}}, 8(1):7--31, 2000.
\newblock \href {https://doi.org/10.1093/jigpal/8.1.7}
  {\path{doi:10.1093/jigpal/8.1.7}}.

\bibitem{DBLP:conf/esop/GhicaS14}
Dan~R. Ghica and Alex~I. Smith.
\newblock Bounded linear types in a resource semiring.
\newblock In Zhong Shao, editor, {\em Programming Languages and Systems - 23rd
  European Symposium on Programming, {ESOP} 2014, Held as Part of the European
  Joint Conferences on Theory and Practice of Software, {ETAPS} 2014, Grenoble,
  France, April 5-13, 2014, Proceedings}, volume 8410 of {\em Lecture Notes in
  Computer Science}, pages 331--350. Springer, 2014.
\newblock \href {https://doi.org/10.1007/978-3-642-54833-8\_18}
  {\path{doi:10.1007/978-3-642-54833-8\_18}}.

\bibitem{DBLP:journals/corr/GhilezanILL14}
Silvia Ghilezan, Jelena Ivetic, Pierre Lescanne, and Silvia Likavec.
\newblock Resource control and intersection types: an intrinsic connection.
\newblock {\em CoRR}, abs/1412.2219, 2014.
\newblock URL: \url{http://arxiv.org/abs/1412.2219}, \href
  {http://arxiv.org/abs/1412.2219} {\path{arXiv:1412.2219}}.

\bibitem{DBLP:journals/tcs/Girard87}
Jean{-}Yves Girard.
\newblock Linear logic.
\newblock {\em Theor. Comput. Sci.}, 50:1--102, 1987.
\newblock \href {https://doi.org/10.1016/0304-3975(87)90045-4}
  {\path{doi:10.1016/0304-3975(87)90045-4}}.

\bibitem{DBLP:conf/tapsoft/GirardL87}
Jean{-}Yves Girard and Yves Lafont.
\newblock Linear logic and lazy computation.
\newblock In {\em TAPSOFT'87}, pages 52--66. Springer, 1987.
\newblock \href {https://doi.org/10.1007/BFb0014972}
  {\path{doi:10.1007/BFb0014972}}.

\bibitem{DBLP:conf/popl/GonthierAL92}
Georges Gonthier, Mart{\'{\i}}n Abadi, and Jean{-}Jacques L{\'{e}}vy.
\newblock The geometry of optimal lambda reduction.
\newblock In {\em POPL'92}, pages 15--26. {ACM} Press, 1992.
\newblock \href {https://doi.org/10.1145/143165.143172}
  {\path{doi:10.1145/143165.143172}}.

\bibitem{DBLP:journals/jfrea/Grimm10}
Jos{\'{e}} Grimm.
\newblock Implementation of {Bourbaki's Elements of Mathematics in Coq: Part
  One, Theory of Sets}.
\newblock {\em J. Formaliz. Reason.}, 3(1):79--126, 2010.
\newblock \href {https://doi.org/10.6092/issn.1972-5787/1899}
  {\path{doi:10.6092/issn.1972-5787/1899}}.

\bibitem{hindley97:_basic_simpl_theor}
J.~Roger Hindley.
\newblock {\em Basic Simple Type Theory}.
\newblock Number~42 in Cambridge Tracts in Theoretical Computer Science.
  Cambridge University Press, 1997.

\bibitem{DBLP:journals/iandc/KesnerL07}
Delia Kesner and St{\'e}phane Lengrand.
\newblock Resource operators for lambda-calculus.
\newblock {\em Inf. Comput.}, 205(4):419--473, 2007.

\bibitem{DBLP:journals/tcs/KesnerR11}
Delia Kesner and Fabien Renaud.
\newblock A prismoid framework for languages with resources.
\newblock {\em Theor. Comput. Sci.}, 412(37):4867--4892, 2011.
\newblock \href {https://doi.org/10.1016/j.tcs.2011.01.026}
  {\path{doi:10.1016/j.tcs.2011.01.026}}.

\bibitem{DBLP:journals/scp/KokkeSW20}
Wen Kokke, Jeremy~G. Siek, and Philip Wadler.
\newblock Programming language foundations in agda.
\newblock {\em Sci. Comput. Program.}, 194:102440, 2020.
\newblock \href {https://doi.org/10.1016/j.scico.2020.102440}
  {\path{doi:10.1016/j.scico.2020.102440}}.

\bibitem{DBLP:conf/popl/Lamping90}
John Lamping.
\newblock An algorithm for optimal lambda calculus reduction.
\newblock In {\em POPL'90}, pages 16--30. {ACM} Press, 1990.
\newblock \href {https://doi.org/10.1145/96709.96711}
  {\path{doi:10.1145/96709.96711}}.

\bibitem{DBLP:journals/iandc/LengrandLDDB04}
St{\'{e}}phane Lengrand, Pierre Lescanne, Daniel~J. Dougherty, Mariangiola
  Dezani{-}Ciancaglini, and Steffen van Bakel.
\newblock Intersection types for explicit substitutions.
\newblock {\em Inf. Comput.}, 189(1):17--42, 2004.
\newblock \href {https://doi.org/10.1016/j.ic.2003.09.004}
  {\path{doi:10.1016/j.ic.2003.09.004}}.

\bibitem{DBLP:conf/popl/Lescanne94}
Pierre Lescanne.
\newblock From lambda-sigma to lambda-upsilon a journey through calculi of
  explicit substitutions.
\newblock In Hans{-}Juergen Boehm, Bernard Lang, and Daniel~M. Yellin, editors,
  {\em Conference Record of POPL'94: 21st {ACM} {SIGPLAN-SIGACT} Symposium on
  Principles of Programming Languages, Portland, Oregon, USA, January 17-21,
  1994}, pages 60--69. {ACM} Press, 1994.
\newblock \href {https://doi.org/10.1145/174675.174707}
  {\path{doi:10.1145/174675.174707}}.

\bibitem{Lescanne95WADT}
Pierre Lescanne.
\newblock {The lambda calculus as an abstract data type}.
\newblock In Magne Haveraaen, Olaf Owe, and Ole-Johan Dahl, editors, {\em
  Recent Trends in Data Type Specification}, volume 1130 of {\em Lecture Notes
  in Computer Science}, pages 74--80. Springer Verlag, 1996.

\bibitem{githubrepo2}
Pierre Lescanne.
\newblock Lambda calculus with duplications and erasures in haskell.
\newblock
  \url{https://github.com/PierreLescanne/LambdaCalculusWithDuplicationsAndErasures},
  2024.
\newblock Last accessed: May 24, 2024.

\bibitem{githubrepo}
Pierre Lescanne.
\newblock Lambda-r in agda.
\newblock \url{https://github.com/PierreLescanne/Lambda-R}, 2024.
\newblock Last accessed: May 12, 2024.

\bibitem{DBLP:conf/lics/LincolnM92}
Patrick Lincoln and John~C. Mitchell.
\newblock Operational aspects of linear lambda calculus.
\newblock In {\em LICS'92}, pages 235--246. {IEEE} Computer Society, 1992.
\newblock \href {https://doi.org/10.1109/LICS.1992.185536}
  {\path{doi:10.1109/LICS.1992.185536}}.

\bibitem{DBLP:conf/birthday/McBride16}
Conor McBride.
\newblock I got plenty o' nuttin'.
\newblock In Sam Lindley, Conor McBride, Philip~W. Trinder, and Donald
  Sannella, editors, {\em A List of Successes That Can Change the World -
  Essays Dedicated to Philip Wadler on the Occasion of His 60th Birthday},
  volume 9600 of {\em Lecture Notes in Computer Science}, pages 207--233.
  Springer, 2016.
\newblock \href {https://doi.org/10.1007/978-3-319-30936-1\_12}
  {\path{doi:10.1007/978-3-319-30936-1\_12}}.

\bibitem{pfenning98:_html}
Frank Pfenning, Iliano Cervesato, and Carsten Schürmann.
\newblock {HTML} bibliography on linear logic.
\newblock Web site, 1998.
\newblock URL: \url{https://www.cs.cmu.edu/~carsten/linearbib/llb.html}.

\bibitem{pierce02:_types_progr_languag}
Benjamin Pierce.
\newblock {\em Types and Programming Language}.
\newblock The MIT Press, 2002.

\bibitem{DBLP:journals/iandc/Pitts03}
Andrew~M. Pitts.
\newblock Nominal logic, a first order theory of names and binding.
\newblock {\em Inf. Comput.}, 186(2):165--193, 2003.
\newblock \href {https://doi.org/10.1016/S0890-5401(03)00138-X}
  {\path{doi:10.1016/S0890-5401(03)00138-X}}.

\bibitem{Restall00}
Greg Restall.
\newblock {\em An Introduction to Substructural Logics}.
\newblock Routledge, 2000.

\bibitem{rose:LIPIcs:2011:3130}
Kristoffer~H. Rose.
\newblock {CRSX - Combinatory Reduction Systems with Extensions}.
\newblock In Manfred Schmidt-Schau{\ss}, editor, {\em 22nd International
  Conference on Rewriting Techniques and Applications, RTA'11}, volume~10 of
  {\em Leibniz International Proceedings in Informatics (LIPIcs)}, pages
  81--90. Schloss Dagstuhl--Leibniz-Zentrum fuer Informatik, 2011.
\newblock \href {https://doi.org/http://dx.doi.org/10.4230/LIPIcs.RTA.2011.81}
  {\path{doi:http://dx.doi.org/10.4230/LIPIcs.RTA.2011.81}}.

\bibitem{rose11:_implem_trick_that_make_crsx_tick}
Kristoffer~H. Rose.
\newblock {Implementation Tricks That Make CRSX Tick}.
\newblock IFIP 1.6 workshop, 6th International Conference on Rewriting,
  Deduction, and Programming, {RDP '11}, 2011.

\bibitem{Schack-Nielsen11}
Anders Schack-Nielsen.
\newblock {\em Implementing substructural logical frameworks}.
\newblock PhD thesis, IT University of Copenhagen, 2011.

\bibitem{Schack-NielsenS10}
Anders Schack{-}Nielsen and Carsten Sch{\"{u}}rmann.
\newblock Curry-style explicit substitutions for the linear and affine lambda
  calculus.
\newblock In J{\"{u}}rgen Giesl and Reiner H{\"{a}}hnle, editors, {\em
  Automated Reasoning, 5th International Joint Conference, {IJCAR} 2010,
  Edinburgh, UK, July 16-19, 2010. Proceedings}, volume 6173 of {\em Lecture
  Notes in Computer Science}, pages 1--14. Springer, 2010.
\newblock \href {https://doi.org/10.1007/978-3-642-14203-1\_1}
  {\path{doi:10.1007/978-3-642-14203-1\_1}}.

\bibitem{DBLP:journals/corr/abs-2009-10241}
Paul Tarau and Valeria de~Paiva.
\newblock Deriving theorems in implicational linear logic, declaratively.
\newblock In Francesco Ricca, Alessandra Russo, Sergio Greco, Nicola Leone,
  Alexander Artikis, Gerhard Friedrich, Paul Fodor, Angelika Kimmig,
  Francesca~A. Lisi, Marco Maratea, Alessandra Mileo, and Fabrizio Riguzzi,
  editors, {\em Proceedings 36th International Conference on Logic Programming
  (Technical Communications), {ICLP} Technical Communications 2020, (Technical
  Communications) UNICAL, Rende (CS), Italy, 18-24th September 2020}, volume
  325 of {\em {EPTCS}}, pages 110--123, 2020.
\newblock \href {https://doi.org/10.4204/EPTCS.325.18}
  {\path{doi:10.4204/EPTCS.325.18}}.

\bibitem{Troelstra:2000:BPT:351148}
Anne~S. Troelstra and Helmut Schwichtenberg.
\newblock {\em Basic Proof Theory ($2^{nd}$ Ed.)}.
\newblock Cambridge University Press, New York, NY, USA, 2000.

\bibitem{walk05}
David Walker.
\newblock Substructural type systems.
\newblock In Benjamin Pierce, editor, {\em Advanced Topics in Types and
  Programming Languages}, pages 3--44. MIT Press, Cambridge, 2005.

\end{thebibliography}


\end{document}